\documentclass[]{article}
\usepackage{amsthm}
\usepackage{lmodern}
\usepackage{amssymb,amsmath}
\usepackage{ifxetex,ifluatex}
\usepackage{listings}
\usepackage{enumerate}
\usepackage{subcaption}
\usepackage{setspace}
\usepackage{hyperref}
\hypersetup{
	colorlinks=true,
	linkcolor=blue,
	filecolor=magenta,      
	urlcolor=cyan,
	citecolor=cyan,
}

\newtheorem*{theorem*}{Theorem}
\newtheorem*{corollary*}{Corollary}
\usepackage{tikz}
\usepackage{verbatim}
\usepackage{tikz-qtree}
\usetikzlibrary{trees}

\usepackage{fixltx2e} 
\ifnum 0\ifxetex 1\fi\ifluatex 1\fi=0 
\usepackage[T1]{fontenc}
\usepackage[utf8]{inputenc}
\else 
\ifxetex
\usepackage{mathspec}
\else
\usepackage{fontspec}
\fi
\defaultfontfeatures{Ligatures=TeX,Scale=MatchLowercase}
\fi
\IfFileExists{upquote.sty}{\usepackage{upquote}}{}
\IfFileExists{microtype.sty}{%
	\usepackage{microtype}
	\UseMicrotypeSet[protrusion]{basicmath} 
}{}
\usepackage[margin=1in]{geometry}
\usepackage{hyperref}
\hypersetup{unicode=true,
	pdftitle={Generalized Bayes Quantification Learning under Dataset Shift},
	pdfauthor={Jacob Fiksel},
	pdfborder={0 0 0},
	breaklinks=true}
\usepackage{graphicx,grffile}
\makeatletter
\def\maxwidth{\ifdim\Gin@nat@width>\linewidth\linewidth\else\Gin@nat@width\fi}
\def\maxheight{\ifdim\Gin@nat@height>\textheight\textheight\else\Gin@nat@height\fi}
\makeatother

\setkeys{Gin}{width=\maxwidth,height=\maxheight,keepaspectratio}
\IfFileExists{parskip.sty}{%
	\usepackage{parskip}
}{
	\setlength{\parindent}{0pt}
	\setlength{\parskip}{6pt plus 2pt minus 1pt}
}
\ifx\paragraph\undefined\else
\let\oldparagraph\paragraph
\renewcommand{\paragraph}[1]{\oldparagraph{#1}\mbox{}}
\fi
\ifx\subparagraph\undefined\else
\let\oldsubparagraph\subparagraph
\renewcommand{\subparagraph}[1]{\oldsubparagraph{#1}\mbox{}}
\fi

\let\rmarkdownfootnote\footnote%
\def\footnote{\protect\rmarkdownfootnote}

\usepackage{amsmath}
\usepackage{float}
\usepackage{xcolor}

\newcommand{\ee}{\end{equation}}
\newcommand{\bea}{\begin{eqnarray}}
\newcommand{\eea}{\end{eqnarray}}
\newcommand{\mbf}[1]{\mathbf{#1}}
\newcommand{\mbs}[1]{\boldsymbol{#1}}

\newcommand{\ba}{\mbf{a}}

\newcommand{\bO}{\mbf{O}}
\newcommand{\bb}{\mbf{b}}

\newcommand{\bd}{\mbf{d}}

\newcommand{\bq}{\mbf{q}}

\newcommand{\bu}{\mbf{u}}

\newcommand{\bv}{\mbf{v}}
\newcommand{\bx}{\mbf{x}}

\newcommand{\bz}{\mbf{z}}
\newcommand{\be}{\mbf{e}}

\newcommand{\bp}{\mbs{p}}
\newcommand{\bA}{\mbf{A}}

\newcommand{\bB}{\mbf{B}}
\newcommand{\bJ}{\mbf{J}}

\newcommand{\bM}{\mbf{M}}
\newcommand{\tM}{\widetilde {\bM}}
\newcommand{\tp}{\widetilde {\bp}}

\newcommand{\bU}{\mbf{U}}
\newcommand{\bV}{\mbf{V}}
\newcommand{\bD}{\mbf{D}}

\newcommand{\bg}{\mbf{g}}

\newcommand{\bzeta}{\mbs{\zeta}}

\newcommand{\bOmega}{\mbs{\Omega}}

\newcommand\numberthis{\addtocounter{equation}{1}\tag{\theequation}}

\newcommand{\balpha}{{\mbs{\alpha}}}

\newcommand{\bgamma}{{\mbs{\gamma}}}
\newcommand{\btheta}{{\mbs{\theta}}}
\newcommand{\bTheta}{{\mbs{\Theta}}}

\newcommand{\bI}{\mbf{I}}
\newcommand{\bH}{\mbf{H}}

\newcommand{\bone}{\mbs{1}}
\newcommand{\bzero}{\mbs{0}}

\newcommand{\ben}{\begin{equation*}}
\newcommand{\een}{\end{equation*}}
\newcommand{\bean}{\begin{eqnarray*}}
\newcommand{\eean}{\end{eqnarray*}}
\newcommand{\bsm}{\begin{smallmatrix}}
\newcommand{\esm}{\end{smallmatrix}}
\newcommand{\bmat}{\begin{matrix}}
\newcommand{\emat}{\end{matrix}}

\newcommand{\ind}{\overset{\mbox{ind}}{\sim}}

\newcommand{\given}{|\;}
\newcommand{\calU}{\mathcal U}

\newcommand{\calL}{\mathcal L}
\newcommand{\calS}{\mathcal S}
\newcommand{\calA}{\mathcal A}
\newcommand{\elll}{\ell_{\calL,n}}
\newcommand{\ellN}{\ell_{\calL,\xi N}}
\newcommand{\ellu}{\ell_{\calU,N}}

\newtheorem{theorem}{Theorem}
\newcommand{\eps}{\epsilon}
\newtheorem{corollary}{Corollary}
\newtheorem{lemma}{Lemma}

\usepackage{mathtools}
\DeclarePairedDelimiter{\ceil}{\lceil}{\rceil}
\usepackage{natbib}
\newcommand{\blind}{1}
\begin{document}
\def\spacingset#1{\renewcommand{\baselinestretch}%
	{#1}\small\normalsize} \spacingset{1}

\if1\blind
{
	\title{\bf Generalized Bayes Quantification Learning under Dataset Shift}
	\author{Jacob Fiksel\thanks{
			The authors gratefully acknowledge \textit{Bill \& Melinda Gates Foundation through the grant number OPP1163221 to Johns Hopkins University for the Countrywide Mortality Surveillance for Action project in Mozambique and Epi/Biostats of Aging Training Grant, Funded by National Institute of Aging T32AG000247}}\hspace{.2cm} \\
		Department of Biostatistics Johns Hopkins University\\
		and \\
		Abhirup Datta\footnote{abhidatta@jhu.edu}\\
		Department of Biostatistics Johns Hopkins University\\		and \\
		Agbessi Amouzou\\
		Department of International Health Johns Hopkins University\\
		and \\
		Scott Zeger \\
		Department of Biostatistics Johns Hopkins University}
	\maketitle
} \fi

\if0\blind
{
	\bigskip
	\bigskip
	\bigskip
	\begin{center}
		{\LARGE\bf Generalized Bayes Quantification Learning under Dataset Shift}
	\end{center}
	\medskip
} \fi

\bigskip

\begin{abstract}
	Quantification learning is the task of prevalence estimation for a test population using predictions from a classifier trained on a different population. 
	Quantification methods assume that the sensitivities and specificities of the classifier are either perfect or transportable from the training to the test population. 
	These assumptions are inappropriate in the presence of dataset shift, when the misclassification rates in the training population are not representative of those for the test population. Quantification under dataset shift has been addressed only for single-class (categorical) predictions and assuming perfect knowledge of the true labels on a small subset of the test population. We propose generalized Bayes quantification learning (GBQL) that uses the entire compositional predictions from probabilistic classifiers and allows for uncertainty in true class labels for the limited labeled test data. Instead of positing a full model, we use a model-free Bayesian estimating equation approach to compositional data using Kullback-Leibler loss-functions based only on a first-moment assumption. The idea will be useful in Bayesian compositional data analysis in general as it is robust to different generating mechanisms for compositional data and allows 0's and 1's in the compositional outputs thereby including categorical outputs as a special case. 
	We show how our method yields existing quantification approaches as special cases. 
	Extension to an ensemble GBQL that uses predictions from multiple classifiers yielding inference robust to inclusion of a poor classifier is discussed. We outline a fast and efficient Gibbs sampler using a rounding and coarsening approximation to the loss functions. We establish posterior consistency, asymptotic normality and valid coverage of interval estimates from GBQL, which to our knowledge are the first theoretical results for a quantification approach in presence of local labeled data. We also establish finite sample posterior concentration rate. Empirical performance of GBQL is demonstrated through simulations and analysis of real data with evident dataset shift.
\end{abstract}

\noindent%
\vfill

\newpage
\spacingset{1.45} 

\hypertarget{introduction}{%
	\section{Introduction}\label{introduction}}
Classifiers are most commonly developed with the goal of obtaining accurate predictions for individual units. 
However, in some applications, the objective is not individual level predictions, but rather to learn about population-level distributions of a given outcome. Examples include sentiment analysis for Twitter users
\citep{giachanou2016}, estimating the prevalence of chronic fatigue syndrome \citep{valdez2018estimating}, and cause
of death distribution estimation from verbal autopsies \citep{king2008, mccormick2016insilico, tariff2_0, interva4, nbc}.

The task of predicting the population distribution $p(y)$ of unobserved true outcomes (labels) $y$ based on observed (and possibly high-dimensional) covariates $\bx$ has been termed {\em quantification} \citep{forman2005, bella2010quantification, gonzalez2017review, perez2019dynamic} in the machine learning literature. Since the covariates are usually passed through a ``trained'' classification algorithm $A$ to obtain predicted labels $\ba:=\ba(\bx)$, quantification can be viewed as prevalence estimation using these predicted labels, and mathematically can be formulated as solving for $p(y)$ from the identity 
\begin{equation}\label{eq:ql}
p(\ba) = \sum_y p(\ba \given y) p(y)\;.
\end{equation} 
Here $p(\ba)$ can be estimated as the mean $\widehat{p(\ba)}$ from a representative sample of predicted labels from the population of interest. However, as both $p(\ba \given y)$ and $p(y)$ on the right hand side are unknowns, assumptions need to be made on $p(\ba \given y)$ to identify $p(y)$. 

Quantification approaches like {\em Classify and Count} (CC) \citep{forman2005} and {\em Probabilistic average} (PA) \citep{bella2010quantification} simply estimate $p(y)$ by $\widehat{p(\ba)}$. This solution is justified only under the  assumption that the {\em misclassification rates} $p(\ba \given y)$ are perfect (i.e., the classifier has $100\%$ sensitivity and specificity). As no classifier is perfect, this assumption is always violated.

{\em Adjusted Classify and Count} (ACC) or {\em adjusted Probabilistic Average} (APA)  
adjust for classifier inaccuracy \citep{forman2008, bella2010quantification}. 
They estimate the classifier's true and false positive rates (or their multi-class equivalents), i.e., $p(\ba \given y)$ from the training data and assumes that these rates are the same in the population of interest (test data).
This assumption of $p_{tr}(\ba \given y)=p_{test}(\ba \given y)$, i.e., the sensitivity and specificity of the classifier is same in the training and test dataset, can be viewed as a {\em transportability} assumption. 
This is similar to {\em transportability} of clinical trial results 
\citep{westreich2017transportability, cole2010generalizing}. 
As 
\begin{equation}\label{eq:pred}
p(\ba \given y) = \int_\bx p(\ba \given \bx) p(\bx \given y) d\bx\;,
\end{equation} 
and $p(\ba \given \bx)$, the prediction map for the trained classifier is same for a given $\bx$ irrespective of whether $\bx$ is in the training or the test populations, 
implicit in the transportability assumption  
is the assumption that $p_{tr}(\mathbf{x}|y)= p_{test}(\mathbf{x}|y)$ \citep{perez2019dynamic}. 
The marginal distributions of the outcomes $p_{tr}(y)$ and $p_{test}(y)$ are allowed to be different.

{\em Dataset shift} occurs when the classifier is trained using data or information \citep[like expert knowledge, ][]{kalter2015direct} from a population different from the test population of interest resulting in both $p_{tr}(y) \neq p_{test}(y)$ and $p_{tr}(\mathbf{x}|y) \neq  p_{test}(\mathbf{x}|y)$  \citep{moreno2012unifying} (as illustrated in Figure \ref{fig:phmrc_resp_rates} for the real application of Section \ref{phmrc-dataset-analysis}). It is evident from (\ref{eq:pred}) that under dataset shift as $p_{tr}(\mathbf{x}|y) \neq  p_{test}(\mathbf{x}|y)$, we will not generally have $p_{tr}(\ba \given y) = p_{test}(\ba \given y)$, and all the aforementioned quantification learning methods 
will be biased. 

When limited validation data with known labels is available from the test set, \cite{datta2018local} proposed population-level {\em Bayesian Transfer Learning} (BTL) -- a quantification approach for dataset shift. BTL uses this limited labeled data to estimate the misclassification rates $p(\ba \given y)$ on the test set, while using classifier predicted labels for the abundant unlabeled test data to estimate $p(\ba)$. The two estimation pieces are combined to solve for $p(y)$ from (\ref{eq:ql}) using a hierarchical Bayesian model. 
BTL only assumes 
transportability of the misclassification rates 
from the labeled test data to the unlabeled test data. Even the marginal distribution of $y$ in the labeled test set is allowed to be different from that in the unlabeled test set. 

BTL uses a multinomial model requiring a single-class (categorical) prediction for each instance. Statistical classifiers are often probabilistic \citep{mccullagh1989generalized, murphy2006naive, specht1990probabilistic} producing a compositional prediction -- a vector of prediction probabilities for every class. To apply BTL, these compositional predictions need to be transformed to categorical predictions by using the plurality rule (most probable category). 
This categorization leads to information loss and \cite{bella2010quantification} showed, in a setting without dataset shift, that quantification using the compositional class
probability estimates can outperform such a practice. 
To our knowledge, there is no quantification method for dataset shift that utilizes the compositional predictions from probabilistic classifiers.

In this manuscript, we generalize Bayesian quantification under dataset shift 
to use entire compositional prediction distributions from classifiers. Rather than positing a complete likelihood for compositional data, we use a Kullback-Leibler divergence loss equivalent to a Bayesian-style estimating equation for compositional data. The advantages of using this loss function over proper likelihoods for compositional data are many fold including robustness to model misspecification, coherence of the estimating equations for the labeled and unlabeled set, and allowing of 0's and 1's in the compositions ensuring use of the same loss-function for categorical, compositional or mixed-type predictions from the classifiers. Estimates of $p(y)$ can be obtained using generalized or {\em Gibbs posteriors} that updates prior beliefs using loss-functions without requiring full distributional specification \cite{shawe1997pac,mcallester1999some,chernozhukov2003mcmc,bissiri2016}. 

Our second innovation concerns allowing uncertainty in true labels in the labeled test set. This is not uncommon. For example, physicians may be
uncertain in the final cause of death \citep{mccormick2016insilico}, or labels may be produced by aggregating crowd sourced responses \citep{bragg2013crowdsourcing}. Existing quantification approaches do not allow for uncertainty in the true labeled test instances, as it is not clear how to define and estimate the misclassification rates where the true labels are probabilistic. 

We use belief-based mixture modeling \citep{szczurek2010introducing} to represent true-label uncertainty as a priori class probabilities, and extend the notion of misclassification rates for such compositional true labels. 
This contribution is of independent importance, as it offers a generalized Bayes estimating equation approach for regressing compositional outcome on compositional covariate without requiring full model specification or transformation of the data, and allowing 0's and 1's in both variables, and with an efficient Gibbs sampler. To our knowledge, this is novel. There has been very little work on using KLD-loss based generalized Bayes for compositional data. The few existing approaches \citep{kessler2015bayesian,yuan2007continual} only consider compositional outcome, not compositional predictors, and have not been theoretically studied. 

We refer to our method as {\em Generalized Bayes Quantification Learning (GBQL)}. We show that GBQL subsumes existing quantification approaches (CC,PA,ACC,APA,BTL) as special cases. 
Like BTL, we extend GBQL to an ensemble approach that can utilize predicted labels from multiple classifiers to produce an ensemble quantification that is robust to inclusion of poor classifiers in the group. We device an fast and efficient Gibbs sampler for GBQL, harmonizing the KL loss function with conjugate priors, and using 
a simple coarsening and rounding approximation to the likelihood. 

Our final contribution is a thorough theoretical study of GBQL. To our knowledge, there is no supporting theory about the accuracy of quantification under dataset shift. There is substantial existing large sample theory for generalized posteriors and related approaches \citep{chernozhukov2003mcmc,zhang2006,jiang2008gibbs,miller2019asymptotic,bhattacharya2019bayesian}. These results can be applied contingent upon identifiability of the parameters specifying the loss function. 
Identifiability is challenging for quantification under dataset shift, as it involves two different loss functions -- one each for the labeled and unlabeled datasets which on their own are both incapable of identifying the parameters. 
Our central result is identifiability of parameters for GBQL that uses both loss functions. Using this we prove asymptotic consistency of the Gibbs posterior, asymptotic normality of the posterior mean, and provide asymptotically well calibrated confidence intervals. 
We also prove a finite sample rate result on posterior concentration. 
All the theory only relies on a correct first-moment assumption and is thus robust to misspecification of the full model. 
We also extend the theory to accommodate the practical modifications used to implement the Gibbs sampler, and to ensemble GBQL for multiple classifiers.

The rest of the manuscript is organized as follows. 
Our method, various extensions, and connection to existing approaches are offered in 
Section \ref{method}. Theoretical properties are discussed in Section \ref{sec:th}. Bayesian implementation and computational considerations are presented in Section \ref{a-gibbs-sampler-for-the-posterior-belief-distribution}.
We show the robustness of our method through simulations in Section \ref{simulations}, and in Section \ref{phmrc-dataset-analysis} we 
demonstrate its performance on the problem of deriving the cause-specific rates of children deaths using the PHMRC dataset.

\section{Method}\label{method}
Let $\calU$ denote an unlabeled dataset of \(N\) instances  randomly sampled from our test population of interest. 
These data do not come with the true class labels $y_r \in \{1,\ldots, C\}$ where $C$ is the total number of categories, but using some pre-trained classifier algorithm $A$, one can predict labels $\ba_r=\ba(\bx_r)$ for $r = 1, \ldots,N$. We do not assume that the training data for the algorithm is available, nor do we assume the knowledge of the covariates $\bx$ for the test set, as long as $\ba(\bx)$ is available to us. 
. 
Our target of interest is $\mathbf{p} = p_{\mathcal{U}}(y) = (p_1, \ldots, p_C)'$, the distribution of the outcome $y$ in our population of interest $\calU$, i.e,   $p_{i} = p(y_{r} = i | r \in \mathcal{U})$.

We further assume availability of a dataset $\calL$ of size $n$  
from our population of interest with both true labels $y_r$ and predicted labels $\ba_r$. 
Because true labels are potentially expensive to obtain, we assume $n \ll N$. 
We do not assume the distribution of $y$ in $\calL$ to be representative of our whole population as true labels may only be available for a convenient sample, i.e., $p_\calL(y) \neq p_\calU(y)$. 
We only assume transportability of the conditional distribution  $p(\bx \given y)$ from $\calL$ to $\calU$. 
This transportability assumption for $p_\calL(\bx \given y)=p_\calU(\bx \given y)$  
is more likely to hold even if the marginal distributions of $y$ are different between $\calL$ and $\calU$.
For example, even if the marginal cause of death distributions are different for hospital and community deaths, given a cause $y$, the symptoms $\bx$ observed in the patient are likely to have similar distribution in both settings. The transportability assumption implies from (\ref{eq:pred}) that $p(\ba \given y)$ is also same for $\calL$ and $\calU$ as the prediction $p(\ba \given \bx)$ from a trained classifier remains the same given $\bx$ irrespective of the population $\bx$ is drawn from. 

Bayesian Transfer Learning (BTL) \citep{datta2018local} assumes that the predictions are deterministic, i.e., $\ba_r$'s are categorical. 
Under transportability, if $\bM=(M_{ij})=(p(\ba_r = j \given y_r=i, r \in \calU \cup \calL))$ is the misclassification matrix of the classifier on the test population, then marginal distribution of $\ba_r$, $r \in \calU$ will be given by $\bM'\bp$. BTL essentially uses the labeled data $\calL$ to estimate $\bM$, the unlabeled data $\calU$ to estimate $\bM'\bp$ and uses the two pieces to solve for $\bp$. 
This is done in using the data model:
\begin{equation}\label{eq:btl}
\begin{array}{cc}
&  \sum_{r \in \calU} \ba_r \sim Multinomial(N,\bM'\bp) \\
& \ba_r \given y_r=i \ind Multinomial(1,\bM_{i*}) \mbox{ for } r \in \calL, i=1,\ldots,C.,
\end{array}
\end{equation}
with $\bM_{i*}$ denoting the $i^{th}$ row of $\bM$. A Bayesian framework is used with priors for $\bp$ and $\bM$. 

\subsection{Bayesian estimating equations for compositional data}\label{sec:main}
The major limitation of BTL is its reliance on multinomial distributions for modeling the data in (\ref{eq:btl}). This  restricts its use to cases where the predicted labels $\ba_r$ are categorical. Classifiers are often probabilistic producing a compositional prediction $\ba_r=(a_{r1},\ldots,a_{rC})'$ with $0 \leq a_{rj}$ and $\sum_j a_{rj}=1$. To use BTL for such probabilistic outputs, one would require unnecessary categorization. 
Instead, we will generalize the model based BTL to a Bayesian estimation equation based quantification method for compositional labels. 

Central to BTL's estimation of population class probabilities (``quantification'')
\begin{equation}
p(y_r=i) = p_i,\; \forall r \in \calU \label{eq:marg}
\end{equation}
is the assumption of transportability of conditional distribution between $\calL$ and $\calU$, i.e.,

\begin{align}
p(\ba_r \given y_r=i) = \bM_{i*} \; \forall r \in \calU \cup \calL . \label{eq:cond}
\end{align} 
The distributional assumption (\ref{eq:cond}) can also be viewed as a first-moment assumption
\begin{equation}
E(\ba_r \given y_r=i) = \bM_{i*}    \; \forall r \in \calU \cup \calL . \label{eq:condmean}
\end{equation}
The two viewpoints are equivalent for categorical $\ba_r$ used in BTL, but (\ref{eq:condmean}) is more general as it is no longer restricted to categorical data. 
For compositional $\ba_r$, rather than specifying $p(\mathbf{a}_{r} | y_{r} = i)$, we only make the general first moment assumption (\ref{eq:condmean}). This is similar to the first-moment assumption in the PA and APA approaches. The challenge is of course how to do 
Bayesian estimation 
without a full model specification.  

First focusing on labeled instances \(r \in \mathcal{L}\), we consider 
the following loss function to connect the parameter \(\mathbf{M}\) to
our data \(\mathbf{a}_{r}, y\)
\begin{align}\label{eq:lossl}
\ell_\calL(\bM \given \{\ba_r,y_r\}_{r \in \calL}) =  \sum_{r \in \mathcal{L}}D_{KL}(\mathbf{a}_{r} || \sum_{i=1}^C \bM_{i*}I(y_r=i)) 
\end{align}
where \(D_{KL}(\bp || \bq)\) is the Kullback--Leibler divergence (KLD) between two distributions $\bp$ and $\bq$.
There are several reasons to choose the KLD loss functions. First, if (\ref{eq:condmean}) is true for some $\bM=\bM^0$, then 
\begin{equation}\label{eq:eel}
E_{\bM^0} \left(\frac {d \ell_\calL}{d \bM} \right) = 0\;.
\end{equation} 
To see this, observe that $-d\ell_\calL/d\bM$ is the derivative of a multinomial log-likelihood. Hence, $E_{\bM^0} (d\ell_\calL/d\bM)=0$ when $\ba_r$ are categorical. However, this derivative is only a linear function of $\ba_r$ and hence the expectation remains unchanged when we switch to compositional $\ba_r$ with the same conditional mean. So, the loss function $\ell_\calL$ leads to a set of unbiased estimating equations \citep{liang1986longitudinal} for compositional data. The second advantage of using KLD is that, as $x \log x = 0$, it seamlessly accommodates instances $0$'s and $1$'s in $\ba_r$. 
Finally, minimizing (\ref{eq:lossl}) is equivalent to
maximizing
\[
\prod_{r \in \mathcal{L}} \prod_{j=1}^{C}\left(\sum_{i}I(y_{r} =i)M_{ij} \right)^{a_{rj}}
\]
which is the exact form of the multinomial quasi-likelihood (MQL). So, when $\ba_r$ are all categorical, this reduces to the likelihood from the second row of (\ref{eq:btl}). 

If only inference on $\bM$ was of interest, frequentist optimization on (\ref{eq:lossl}) or GEE using its derivative can be executed. Using the rich theory of estimating equations, the estimate \(\widehat{\mathbf{M}}\) has been shown to be a consistent estimator for \(\mathbf{M}\) \citep{papke1996econometric, mullahy2015multivariate}, and such frequentist approaches have been commonly used in the econometrics literature for regression with a compositional outcome. 

However, the primary interest in quantification learning is estimation of $\bp$ and accurate estimation of the nuisance parameter $\bM$ is only an important intermediate step. The unlabeled dataset $\calU$ is the only one informing estimation of $\bp$, and using (\ref{eq:marg}) and (\ref{eq:condmean}), the marginal first-moment condition for $\ba_r$ in $\calU$ is given by:
\begin{align}\label{eq:margmean}
E[\ba_{r}] &= E[E[\ba_{r} | y_{r}]] 
= \sum_{i}p_{i}E[\ba_{r} | y_{r} = i] = \bM'\bp, \forall r \in \calU.
\end{align}
This harmonizes with the loss-function 
\begin{align}\label{eq:lossu}
\ell_\calU(\bp, \bM \given \{\ba_r\}_{r \in \calU}) &=  \sum_{r \in \mathcal{L}}D_{KL}(\mathbf{a}_{r} || \bM'\bp) \;.
\end{align}

The loss function $\ell_\calU$ for the marginal distribution of the predicted labels is coherent with the loss-function $\ell_\calL$ for their conditional distribution, as they are based off of coherent moment conditions (\ref{eq:condmean}) and (\ref{eq:margmean}). Assuming (\ref{eq:marg}) and (\ref{eq:condmean}) holds for some true $\bp^0$ and $\bM^0$, following the same logic used in (\ref{eq:eel}), we can show
\begin{equation}\label{eq:eeu}
E_{\bM^0,\bp^0}\left(\frac{d \ell_\calU }{ d(\bM,\bp)}\right)=0, 
\end{equation} i.e., the derivative is once again an estimating equation. However, if we only considered $\ell_\calU$ without bringing in $\ell_\calL$, $\bM$ and $\bp$ cannot be identified. For example, $\ell_\calU(\bM,\bp) = \ell_\calU(\bI,\bM'\bp)$. Hence, we will consider the joint loss-function $\ell = \ell_\calL + \ell_\calU$ as adding $\ell_\calL$ helps to identify $\bM$ which in turns makes $\bp$ identifiable. 
 
Bayesian inference using only loss functions, without full model specification, is now well-established. 
For any reasonable choice of a loss-function $\ell(\btheta \given data)$ and prior $\Pi(\btheta)$, a  {\em Gibbs posterior} is defined as the distribution 
\begin{equation}\label{eq:gen}
\Pi (\btheta \given data) \propto \exp\left(-\alpha\ell(\btheta \given data)\right) \Pi(\btheta). 
\end{equation} 
for some $\alpha > 0$, provided the normalizing constant exists. The idea of updating prior beliefs through loss functions via (\ref{eq:gen}) has developed independently in multiple fields, dating back atleast to \cite{vovk1990aggregating}. This posterior is interpreted as the distribution $\nu$ for $\btheta$ minimizing the loss function $\alpha E_{\nu}(\ell(\btheta \given data)) + D_{KL}(\nu,\Pi)$.  Gibbs posteriors (also known as pseudo- or generalized posteriors) have been widely used to derive generalization errors in the PAC-Bayesian framework \citep{shawe1997pac,mcallester1999some,catoni2003pac}. Functionals of the posterior in (\ref{eq:gen}) has been referred to as Laplace-type estimators (LTE) or quasi-Bayesian estimators (QBE) in \cite{chernozhukov2003mcmc}. \cite{jiang2008gibbs} used Gibbs posteriors for high-dimensional variable selection. The case where the loss-function $\exp(-\ell)$ is a fractional likelihood has received extra attention with the literature demonstrating the utility of fractional posteriors over full posteriors especially under model misspecification \citep{zhang2006,walker2001bayesian,bhattacharya2019bayesian}. \cite{bissiri2016} showed that, given the loss and the prior, (\ref{eq:gen}) is the unique update that is invariant to sequentially updating with each additional data point or joint updating using all data points. The parameter $\alpha$ is related to calibration of credible intervals based on Gibbs posteriors and its choice will be discussed in Section \ref{sec:cp}. The problem of quantification learning under dataset shift using compositional predicted labels have not been studied using a Bayesian or generalized Bayes framework. 

We use $\ba^\calL$ and $\ba^\calU$ to respectively denote $\{\ba_r\}_{r \in \calL}$ and $\{\ba_r\}_{r \in \calU}$, and similar notations for collections of the other variables. The two loss functions $\ell_L$ and $\ell_U$  have same functional form leading to the Gibbs posterior:
\begin{align*}
\Pi(\bp, \bM \given \mathbf{a}^{\mathcal{U}}, \mathbf{a}^{\mathcal{L}},y^{\mathcal{L}}) & \propto \exp\left( - \alpha \sum_{r \in \mathcal{U}}D_{KL}(\mathbf{a}_{r} || E[\mathbf{a}_{r}]) - \alpha\sum_{r \in \mathcal{L}}D_{KL}(\mathbf{a}_{r} || E[\mathbf{a}_{r} | y_{r}])\right)\Pi(\bp,\bM) \\
& \propto \exp\left( \alpha \sum_{r \in \mathcal{U}}\sum_{j=1}^{C} a_{rj}\log\frac{\sum_{i}p_{i}M_{ij}}{a_{rj}} + \alpha\sum_{r \in \mathcal{L}}\sum_{j=1}^{C} a_{rj}\log\frac{\sum_{i=1}^{C}I(y_{r} = i)M_{ij}}{a_{rj}}\right) \Pi(\bp,\bM)
\end{align*}

If all $\ba_r$ were categorical, this posterior with $\alpha=1$ is identical to the one from the BTL model (\ref{eq:btl}). However, using estimating equations and generalized Bayes, we now have an unified framework for Bayesian quantification for both categorical, compositional or mixed-type $\ba_r$ without having to specify the full models for the different data types. In subsequent sections we illustrate how this generalized Bayes framework naturally lends itself to accommodating uncertainty in true labels, multiple classifiers, and shrinkage priors. 

\subsection{Advantages over full Bayesian modeling of compositional data}\label{sec:adv}
Before discussing these extensions, we highlight the advantages of a generalized Bayes framework over a fully Bayesian approach for quantification learning. There are fundamental hurdles to extend the model in (\ref{eq:btl}) when some or all $\ba_r$ are compositional. The Dirichlet distribution and its generalizations \citep{hijazi2009modelling, wong1998generalized, tang2018zero} are standard models for compositional data. However, there are several issues with specifying a Dirichlet model for $\ba_r$ in quantification learning.  

\begin{enumerate}
	\item We allow the $\ba_r$ to take 0 and 1 values as the predictions can be categorical or sparse-compositional (prediction for some classes to be exactly $0$).  
	Dirichlet distributions do not support 0's and 1's, and would require forcing the $a_{rj}$'s to lie strictly in $(0,1)$ using some arbitrary cutoff.  
	Alternatively, one can use the zero-inflated Dirichlet distribution \citep{tang2018zero} to formally account for the presence of 0's, which leads to a significant increase in the number of parameters.
	
	A related point is that single-class classifiers can be viewed as a subclass of probabilistic classifiers, with the predicted distribution being degenerate. Hence, if using two classifiers, one with compositional predictions and one with single-class predictions, use of the Dirichlet model for the former and a multinomial model for the latter is discordant. 
	
	\item Our generalized Bayes approach has a coherence property required for quantification learning. The conditional expectation model (\ref{eq:condmean}) for the labeled data 
	leads to the marginal model (\ref{eq:margmean}) for the unlabeled data. This is central to identification of $\bp$. 
	Specifying $\ba_r \given y=i$ as a Dirichlet distribution (or its variants), will endow $\ba_r$ with a mixture-Dirichlet marginal distribution which presents a computational challenge in posterior sampling. Our pseudo-likelihood for $\ba_r$ nicely harmonizes with conjugate Dirichlet priors for the parameters $\bM$ and $\bp$ leading to an efficient Gibbs sampler. 
	
	\item Fully specified Dirichlet distributions are susceptible to model misspecification. The generalized Dirichlet distribution \citep{wong1998generalized} can be used to broaden the model class, however increased model complexity comes with added computational burden. 
\end{enumerate} 

Finally, as an alternate to Dirichlet-based likelihoods,  
one can log-transform the data and use multivariate normal or skew-normal to fully model the log-ratio coordinates of the compositional $\ba_r$ \citep{comas2016log}. However, a transformation-free approach is generally more desirable. Also, a model on the transformed compositional $\ba_r$ will  be discordant with the multinomial model for the categorical $\ba_r$. The transformations also generally do not allow for 0's and 1's. 

\subsection{Quantification using uncertain true labels}\label{sec:belieflabels}

As stated in Section \ref{introduction}, in many applications, there is uncertainty in some or all of the true labels in the labeled test set $\calL$. For example, a panel of physicians may fail to unanimously agree on a single cause of death, and only provide a subset of the list of causes from which they believe the individual was equally likely to die. No existing quantification approach can work with uncertainty in true labels. In this Section, we generalize the notion of misclassification rates to uncertain true labels and extend GBQL accordingly. 

Following the belief based modeling framework of \cite{szczurek2010introducing}, we let \(b_{ri}\)
represent the a priori probability that instance \(r\) belongs to label
\(i\). Then \(\bb_{r}\) is constrained such that
\( 0 \leq {b}_{ri}\) and \(\sum_{i=1}^{C} b_{ri}=1\). For  \(r \in \mathcal{L}\) we no longer observe the $y_r$'s but observe the belief vector $\bb_r$. Cases where the true label is identified with complete certainty can be subsumed by writing $\bb_r=\be_i$ when $y_r=i$, $\be_i$ denoting the vector with $1$ at the $i^{th}$ component and zeros elsewhere. 
We can generalize the conditional first-moment condition (\ref{eq:condmean}) to 
\begin{equation}\label{eq:compreg}
E[\ba_{r} | \bb_{r}] = E[E[\ba_{r} | y_r, \bb_{r}] \given \bb_r] = E\left(\sum_i M_{i*}I(y_r=i) \given \bb_r\right)  = \bM'\bb_r.
\end{equation}
So, our loss function for $\calL$ is now
$\ell_\calL(\bM \given \ba^\calL,\bb^\calL) = \sum_{r \in \mathcal{L}}D_{KL}(\mathbf{a}_{r} || \mathbf{M}^{'}\bb_{r}) = \sum_{r \in \mathcal{L}}\sum_{j=1}^{C} a_{rj}\log\left(\frac{\sum_{i=1}^{C}b_{ri}M_{ij}}{a_{rj}}\right)$.
The loss for the unlabeled data remains the same, and generalized Bayes proceeds using the likelihood $\ell_\calL + \ell_\calU$ with this generalized choice of $\ell_\calL$. 
Appealing to the motivation of generalized Bayes \citep{chernozhukov2003mcmc,zhang2006,bissiri2016}, we can see that the Gibbs posterior 
\(\nu = \Pi(\mathbf{p},\mathbf{M} |\mathbf{a}^{\mathcal{U}}, \mathbf{a}^{\mathcal{L}},  \bb^{\mathcal{L}})\)
is the probability measure which, as \(n, N \rightarrow \infty\) and \(\frac{n}{N} \rightarrow \xi\),
minimizes the Bayes risk
\[
E_{\nu}\left[E_{r \in \mathcal{U}}[D_{KL}(\mathbf{a}_{r} || \mathbf{M}^{'}\mathbf{p})] + \xi E_{r \in \mathcal{L}}[D_{KL}(\mathbf{a}_{r} || \mathbf{M}^{'}\bb_{r})]\right].
\]

\hypertarget{incorporating-multiple-predictions}{%
	\subsection{Ensemble Quantification Incorporating Multiple
		Classifiers}\label{incorporating-multiple-predictions}}

There may be \(k=1,\ldots, K\) predictions for each instance corresponding to $K$ classifiers. 
\cite{datta2018local} has shown the advantage of incorporating multiple algorithms for quantification when only categorical predictions are available, and their ensemble quantification can easily be extended to compositional settings. For the ensemble approach, a fundamental observation is that each algorithm is expected to have their own sensitivities and specificities. Representing the $k^{th}$ algorithm prediction 
for instance \(r\) as \(\mathbf{a}_{r}^{k}\) and the corresponding misclassification matrix as $\bM^k$, 
the conditional first moment assumption (\ref{eq:condmean}) becomes 
\begin{equation}
E(\ba_r^k \given y_r=i) = \bM^k_{i*}    \; \forall r \in \calU \cup \calL. \label{eq:condmeank}
\end{equation}
For the unlabeled data, we will now have the labels satisfying the marginal first moment condition $E(\ba_r^k) = {\bM^k}'\bp$. Hence, each of the $K$ predictions for the unlabeled test data $\calU$ informs about the same parameter $\bp$ (our estimand) and we define {\em ensemble GBQL} using sum of the losses for the individual algorithms:

\[
\sum_{k=1}^{K}\left[\sum_{r \in \mathcal{U}}D_{KL}(\mathbf{a}_{r}^{k} || {\mathbf{M}^{(k)}}^{'}\mathbf{p}) + \sum_{r \in \mathcal{L}}D_{KL}(\mathbf{a}_{r}^{k} || {\mathbf{M}^{(k)}}^{'}\bb_{r})\right].
\]

Ensemble GBQL offers a unified framework for combining information from probabilistic classifiers (compositional $\ba_r$) and deterministic ones (categorical $\ba_r$) like clinical classifiers for cause of deaths. 

\hypertarget{shrinkage-towards-the-source-probability-predictions}{%
	\subsection{Shrinkage towards default quantification methods}\label{sec:shrink}
}
We now discuss how existing quantification approaches are special cases of GBQL with specific choices of degenerate priors for $\bM$. We will leverage this property to construct shrinkage priors in data-scarce settings. 

The simplest quantification approach is called Classify \& Count (CC) \citep{forman2005}. CC requires a single predicted class $j$ for each instance, so that $a_{rj} \in \{0, 1\}$. The CC estimate of $p_{i}$ is simply
$
\hat{p}_{i}^{CC} = \frac{\sum_{r \in \mathcal{U}}a_{ri}}{N}
$.
Probabilistic Average (PA) \citep{bella2010quantification} extended this to allow probabilistic predictions. The PA estimate, $\hat{p}_{i}^{PA}$, is obtained in the same manner as $\hat{p}_{i}^{CC}$, but does not require $a_{rj} \in \{0,1\}$. It is clear from above that CC (or PA) produces a biased estimate as $E(\bp^{CC}) = E(\bp^{PA}) = E(\ba_r) = \bM'\bp$ which is generally does not equal $\bp$ unless $\bM=\bI$ (i.e., the classifier is perfect) or $\bp$ is a stationary distribution for $\bM$. 

Adjusted Classify \& Count (ACC) \citep{forman2005}  accounts for the classifier being not perfect 
even for the training population. ACC relies on cross-validation using training data splits to estimate the true positive and false positive rates (tpr and fpr) of the classifier (for the base case of $C=2$), and propose 
\begin{equation}\label{eq:acc}
\hat{p}_{i}^{ACC} = \frac{\hat{p}_{i}^{CC} - fpr}{tpr-fpr}.
\end{equation}

\cite{bella2010quantification} developed 
an adjusted version of the PA estimate (APA) similar to ACC but for probabilistic predictions. ACC, APA and their multi-class extensions \citep{hopkins2010method} are inappropriate for quantification in the presence of dataset shift, as the $fpr$ and $tpr$ estimated from the training data will not be representative of those in the test population \citep{perez2019dynamic}. Also, $\hat{p}_{i}^{ACC}$ is not guaranteed to be in $[0,1]$, although \cite{hopkins2010method} correct for this using constrained optimization. 

To make the connection between these methods and GBQL, we first consider the scenario where $n=0$, i.e., when there is no labeled test set to estimate dataset shift. Consider a sequence $\{\Pi_u(\bM) \given u=1,2,\ldots\}$ of priors for $\bM$ such that $\Pi_u$ converges in distribution to $\delta(\bM^{pr})$, a degenerate prior at some pre-fixed transition matrix $\bM^{pr}$. Then the Gibbs posterior $\nu_u$ of GBQL using the prior $\Pi(\bp)\Pi_u(\bM)$ converges in distribution to 
\[ \lim_{u \rightarrow \infty}\nu_u(\bp) \propto \exp\left(- \sum_{r \in \calU} D_{KL}(\ba_r || {\bM^{pr}}'\bp)\right)\Pi(\bp)\;. \]

If $\bM^{pr}=\bI$, then for any prior choice of $\bp$, $\lim_{u \rightarrow \infty}\nu_u(\bp) \propto Dirichlet(\bp ; \sum_{r \in \calU} \ba_r) \Pi(\bp)$. In particular, if $\Pi(\bp)=Dirichlet(\bp ; \bzero)$ or as $N \rightarrow \infty$, then $\lim_{u \rightarrow \infty}\nu_u(\bp) = Dirichlet(\sum_{r \in \calU} \ba_r)$. For categorical $\ba_r$, this result was proved in \cite{datta2018local}, and shows that $E_{\lim_u \nu_u}(\bp)=\bp^{CC}$, i.e., using priors $\Pi_u(\bM)$ shrinking towards the degenerate prior at $\bI$, inference from GBQL becomes identical to inference from Classify and Count \citep{forman2005} when there is no labeled dataset. Analogously, for the same settings, when $\ba_r$ are compositional, posterior mean from GBQL becomes identical to Probabilistic Average \citep{bella2010quantification}. Extending, the argument to the settings with multiple predictions, it is straightforward to see that $E_{\lim_u \nu_u}(\bp)=1/K \sum_{k=1}^K \bp^{k,PA}$, i.e., the posterior mean from our ensemble classifier coincides with the average of the CC or PA estimates for the $K$ classifiers. 

Alternatively, if the misclassification matrix $\bM^{tr}$ for the training data is available and can be trusted for test data, one can use 
$\bM^{pr}=\bM^{tr}$. Then the posterior $\lim_{u \rightarrow \infty}\nu_u(\bp)$ coincides with the implicit likelihood in Adjusted Classify and Count (for categorical $\ba_r$) and in Adjusted Probabilistic Average (for compositional $\ba_r$). To see this, note that the ACC estimate (\ref{eq:acc}) for 2 classes and categorical $\ba_r$'s, relies on the principle that 
$$p_1 = \frac{ E(a_{r1}) - M^{tr}_{21}}{M^{tr}_{11} - M^{tr}_{21}}.$$
This is equivalent to $E(a_{r1}) = p_1 M^{tr}_{11} + p_2 M^{tr}_{21}$ or $E(\ba_r)=\bM^{tr'}\bp$, i.e., assuming (\ref{eq:cond}) with $\bM=\bM^{tr}$. 
Thus using $\Pi(\bM)\approx \delta(\bM=\bM^{tr})$ in GBQL is a better implementation of ACC or APA, ensuring that the posterior mean of $\bp$ is guaranteed to be a vector of probabilities lying in $[0,1]$. This is not assured in their current implementation based on the direct correction (\ref{eq:acc}). 

Hence, in absence of local labeled set, a prior for $\bM$ concentrated around $\bI$ or $\bM^{tr}$, makes estimates from GBQL nearly coincide with these existing methods (Figure \ref{fig:gbql_schematic}). GBQL in fact provides a probabilistic framework around these existing quantification approaches.
 
When labeled data is present, instead of using $\bM=\bI$ (i.e., no adjustment as in CC,PA) or $\bM=\bM_{tr}$ (i.e., transportability of the conditional distributions between the training and test data as used in ACC,APA), GBQL estimates an unstructured $\bM$
only assuming transportability of the conditional distributions from the limited labeled test data $\calL$ to all test data. However, quantification projects like burden of disease estimation using nationwide surveys are often multi-year endeavors. At the initial stages of such projects, $\calL$, consisting of hospital deaths with clinically diagnosed causes, can be very small. With very limited labeled data, estimating both \(\mathbf{M}\) and \(\mathbf{p}\) precisely with vague priors is ill-advised as $\bM$ involves $C(C-1)$ parameters. 
In such settings, the above-established link between GBQL and the existing quantification methods can be exploited to choose shrinkage priors for stabilizing estimation of $\bM$. 
For example, one can use the priors $\bM_{i*} \sim Dirichlet(\gamma_{ui} (\bM^{pr}_{i*}+\eps_u\bone))$ for small $\eps_u$ or large $\gamma_{ui}$. This prior concentrates around $\delta(\bM=\bM^{pr})$ if either $\eps_u \rightarrow 0$ or $\gamma_{ui} \rightarrow \infty$, hence for no or small labeled dataset, the estimate for GBQL will shrink to those from CC or PA (if $\bM^{pr}=\bI$) or to ACC or APA (if $\bM^{pr}=\bM^{tr}$). With limited labeled data, these shrinkage priors make a bias-variance trade-off yielding estimates with higher precision. The benefits of such shrinkage priors over non-informative priors have been demonstrated in \cite{datta2018local}. Finally as more and more labeled data is collected, in the next section we show that any reasonable choice of prior (including all these shrinkage priors) leads to desirable asymptotic and finite-sample properties of the GBQL estimate. 

\begin{figure}[t]
	\centering
	\resizebox{0.8\textwidth}{!}{%
		\begin{tikzpicture}[grow'=right,level distance=1.25in,sibling distance=.25in]
		\tikzset{edge from parent/.style= 
			{thick, draw, edge from parent fork right},
			every tree node/.style=
			{draw,minimum width=1in,text width=1in,align=center}}
		\Tree [.\textbf{GBQL}
		[.{No labeled test data}
		[.{Single class predictions}
		[.{One classifier}
		[.{Prior: $\Pi(\mathbf{M})=\delta(\mathbf{I})$} 
		{CC\\ \citep{forman2005}} ]
		[.$\Pi(\mathbf{M})=\Pi(\mathbf{M}^{tr})$ 
		{ACC\\ \citep{forman2005}} ] ]
		[.{Multiple classifiers }
		[.$\Pi(\mathbf{M^k})=\delta(\mathbf{I})$ 
		[.{Average CC} ]] \edge[red]; [.$\Pi(\mathbf{M}^k)=\delta(\mathbf{M}^{k,tr})$ ] ]]
		[.{Compositional predictions}
		[.{One classifier}
		[.{$\Pi(\mathbf{M})=\delta(\mathbf{I})$} 
		{PA\\ \citep{bella2010quantification}} ]
		[.$\Pi(\mathbf{M})=\Pi(\mathbf{M}^{tr})$ 
		{APA\\ \citep{bella2010quantification}} ] ]
		[.{Multiple classifiers }
		[.$\Pi(\mathbf{M^k})=\delta(\mathbf{I})$ 
		[.{Average PA} ]] \edge[red]; [.$\Pi(\mathbf{M}^k)=\delta(\mathbf{M}^{k,tr})$ ] ]]]
		[.{Labeled test data} 
		[.{With true labels} 
		[.{Single class predictions} 
		[.{One classifier} 
		[.{BTL\\ \citep{datta2018local}} ]]
		[.{Multiple classifiers} 
		[.{Ensemble BTL\\ \citep{datta2018local}} ]]]
		\edge[red];[.{Compositional predictions} ]]\edge[red];
		[.{Uncertainty in true labels} ]]]
		\end{tikzpicture}
	}
	\caption{GBQL includes and extends the common quantification methods through different classifier outputs and choices of priors for $\mathbf{M}$. \textcolor{red}{Red lines} indicate the settings where  GBQL extends current methods, while black lines indicate where GBQL subsumes existing methods.}\label{fig:gbql_schematic}
\end{figure}
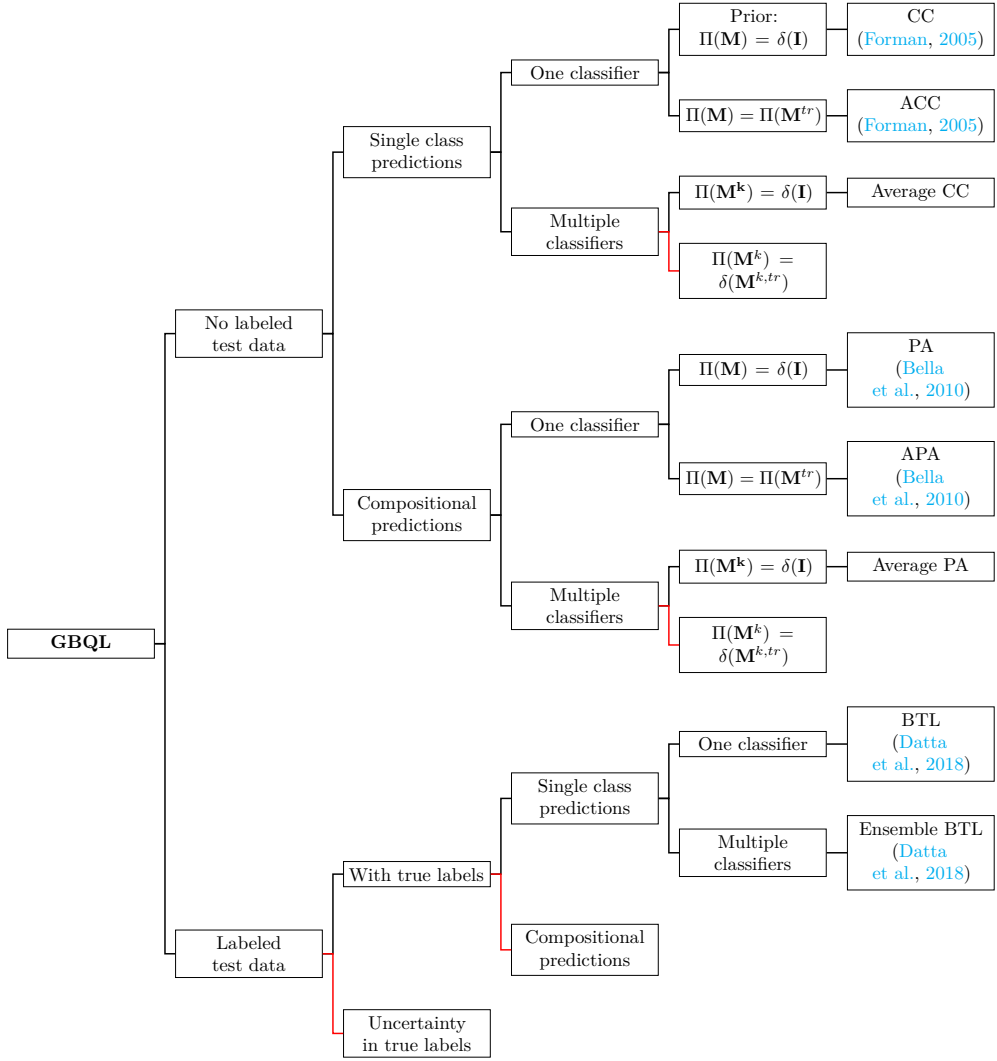

\subsection{Parametric modeling of misclassification rates}\label{sec:para} 
An alternative to using the shrinkage priors would be to incorporate domain-knowledge about the misclassification rates via informative priors in a parametric model for $\bM$ . 
One example of such prior knowledge is impossibility of occurrence of some true-label-predicted-label class pairs. This can especially be true for a domain-knowledge driven classifier. A clinically driven cause-of-death classifier like the Expert Algorithm \citep{kalter2015direct} is unlikely to produce entirely improbable predicted causes given a true cause-of-death. Hence, a misclassification rates for a large set of class-label pairs can likely be set to $0$ a priori. Such sparse models for $\bM$ can drastically reduce the number of parameters. 

If $\calS_i=\{j \given M_{ij} \neq 0\}$ denote the known support set for true class $i$. Then 
we can use an uninformative prior for the non-zero entries of $\bM$ via the uniform sparse-Dirichlet prior as follows:
\begin{equation*}
M_{i,\calS[i]} \ind Dirichlet(\delta \bone_{|\calS_i|}), \mbox{ for } i=1,\ldots,C.
\end{equation*}
where $\bone_d$ denotes the $d$-dimensional vector of ones and $\delta >0$. We can easily modify the Gibbs sampler presented in Section \ref{a-gibbs-sampler-for-the-posterior-belief-distribution} for such a sparse model to ensure conjugate updates. Only change would involve replacing the Dirichlet updates for $M_{i*}$ with Dirichlet update for $M_{i,\calS[i]}$ as $M_{ij}$  for $j \notin \calS_i$ are set to $0$. 

Strategies other than sparsity can also be adopted to reduce the number of parameters. Consider an example where the classifier is nearly perfect for the training population, and the test population is a mixture population where some cases are similar to the training data for which the classifier will be accurate, and some are from a different population for which the classifier reduces to a random guess. In such cases, $\bM$ can be modeled as an equicorrelated (only one parameter) or row-equicorrelated ($C$ parameters) matrix.  

Finally, assuming homogeneous misclassification rates for the entire population maybe inappropriate for some applications. We can then model $M_{ij}$ as function of a small set of covariates $U$. \cite{datta2018local} considered this extension for categorical labels, using a logit regression model of $\bM$ on $U$ and proposed a Gibbs sampler using the Polya-Gamma sampler of \cite{polson2013bayesian}. We can adopt such a model for GBQL, and formulate a Gibbs sampler with conjugate updates using the data augmentation scheme we present in Section \ref{a-gibbs-sampler-for-the-posterior-belief-distribution} combined with the Polya-Gamma data augmentation. 

\section{Theory}\label{sec:th} In this section we establish asymptotic and finite sample guarantees for GBQL 
for the general case from Section \ref{sec:belieflabels} where the true labels in $\calL$ are observed with uncertainty $\bb_r$. This subsumes the case of Section \ref{sec:main} with exact labels $y_r$. The Gibbs posterior for GBQL is given by:
\begin{equation}\label{eq:nu}
\nu_N = \Pi(\bp, \bM \given \mathbf{a}^{\mathcal{U}}, \mathbf{a}^{\mathcal{L}},\bb^{\mathcal{L}}) \propto \exp\left( -\alpha \sum_{r \in \mathcal{U}}D_{KL}(\mathbf{a}_{r} || \bM'\bp) - \alpha\sum_{r \in \mathcal{L}}D_{KL}(\mathbf{a}_{r} || \bM'\bb_{r})\right)\Pi(\bp,\bM) .
\end{equation}

Theory of Gibbs posteriors is well developed. \cite{chernozhukov2003mcmc} developed very general results for asymptotic posterior consistency, normality and coverage of confidence/credible intervals of Gibbs posteriors. Similar results were developed in \cite{miller2019asymptotic}.  \cite{walker2001bayesian,zhang2006} theoretically demonstrated benefits of using fractional posteriors over full posteriors. \cite{bhattacharya2019bayesian} developed finite-sample concentration results for fractional posteriors under model misspecification. 

Our quantification approach is grounded only in correct specification of the conditional first moment assumption 
(\ref{eq:condmean}), i.e, we work in an M-free (model-free) setting. Most applications of the aforementioned theoretical results have been demonstrated in the M-closed (true data model within the class of models considered), or M-open (misspecified model) case \citep{bernardo2009bayesian}. Previous applications of Gibbs posteriors to the M-free case include the M-estimation examples of \cite{chernozhukov2003mcmc}, and the misclassification-loss based approach for variable selection of classifiers of \cite{jiang2008gibbs}.

A central component driving the theory of Gibbs posteriors is some assumption about identifiability of the parameters. Parameters are identifiable if they maximize the likelihood for the M-closed case, minimize the KL divergence to the true distribution for the M-open case, and minimize the loss function 
for M-free case.  
As we see in (\ref{eq:nu}), we have two loss functions for GBQL. The loss function $\ell_\calL$ for the labeled data $(\ba_r,\bb_r)$ in $\calL$ is different from the loss-function $\calU$ for the unlabeled data $\ba_r$ in $\calU$. Each loss on their own is incapable of identifying the parameter of interest $\bp$, as $\ell_\calL$ doesn't even depend on $\bp$, and $\ell_\calU(\bM,\bp) = \ell_\calU(\bI,\bM'\bp)$. 
To our knowledge, 
there hasn't been any application of the general theory of Gibbs posterior to a setting similar to quantification learning requiring more than one type of loss-function to identify the estimand. 

More generally, there is no theory on model-free Gibbs posteriors for compositional data using cross-entropy (KLD) loss. Related Bayesian methodological work are the Gibbs samplers for compositional regression  \citep{kessler2015bayesian}, and for multiple toxicity grades in the context of early-phase clinical trials  \citep{yuan2007continual}. However, these methods only consider a compositional outcome, and not a compositional predictor as we do in Section \ref{sec:belieflabels}. Also, their approach was motivated from fractional multinomial regression and not from loss-function-based generalized Bayes, and did not come with any theoretical guarantees. 

Our main result which leads to all the subsequent asymptotic and finite-sample guarantees is identifiability of $\bp$ from the loss $\ell_\calL + \ell_\calU$. We introduce the following notations for the theory. 
We will use $\tM$ and $\tp$ to denote the free parameters in $\bM$ and $\bp$ respectively, i.e., $\tM$ excludes the last column of $\bM$, $\tp$ excludes the last element of $\bp$. $\bM$ and $\bp$ are bijective functions of $\tM$ and $\tp$ respectively, so we will use them interchangeably.
Let $\btheta=(\tM,\tp)$, then $\btheta$ is supported on the compact set $\Theta = \calS_{C-1}^C \otimes \calS_{C-1}$ where $\calS_d = \{ \bx \in \mathbb{R}^d \given x_i \geq 0, \bone'\bx \leq 1\}$. 
Switching to $\tM$ and $\tp$ ensures that the parameter space $\Theta$ has a non-empty interior.  

Let $\bp^0$ and $\bM^0$ denote the true values and $\btheta^0=(\tM^0,\tp^0)$, an interior point in $\Theta$. We first state and discuss our assumptions, for the theory:
\begin{enumerate}
	\item (Positivity) Let $\widetilde S_C=\{\bx \in \mathbb R^C \given x_i \geq 0, \bone'\bx =1\}$ denote the $C$-dimensional probability simplex with corners $\be_i$ -- the $C \times 1$ vector with $1$ at the $i^{th}$ position and $0$'s elsewhere, then for any arbitrary small neighborhood $N_i \subset \widetilde S_C$ containing $\be_i$, $F_{b,\calL}(N_i) > 0$ where $F_{b,\calL}$ is the true distribution of $b_r$, $r \in \calL$. 
	\item (Separability) $\bM^0$ is non-singular.
\end{enumerate}

Assumption 1 states that the true data-generation distribution of the compositional labels $\bb_r$ for $\calL$, has positive mass at each of the $C$ corners of the simplex $\widetilde S_C$. 
Each corner of the simplex represents a cause-category. Mass at the $i^{th}$ corner is needed to estimate the $i^{th}$ row of $\bM$ from $\calL$. So the assumption ensures that there is data to estimate each row of $\bM$. To interpret Assumption 1, consider the special case where we observe the true labels $y$, and the predicted labels $\ba$ are categorical. Then Assumption 1, along with $\bM^0$ being an interior point, ensure that for large enough $n$, for every $(i,j)$ pair, there are cases in $\calL$ for whom the true class is $i$ and the predicted class is $j$. This is of course necessary to estimate the misclassification rate $M_{ij}$. Thus, Assumption 1 can be interpreted as a {\em positivity assumption} ensuring that the limited labeled test set can estimate the sensitivities and specificities of the classifier for all class-pairs. 

Assumption 2 
is a {\em separability assumption} necessary for quantification. If there exists two probability vectors $\bp^0$ and $\bp^1$ such that ${\bM^0}'\bp^0={\bM^0}'\bp^1$ then $\ell_\calU(\bM^0,\bp^0)=\ell_\calU(\bM^0,\bp^1)$. So it will be impossible to identify $\bp$ based on predicted labels. A trivial example of this is a 2-class setting with $M_{11}=M_{21}=1$ and $M_{12}=M_{22}=0$. Then all labels will be predicted as class 1 and it would not be possible to distinguish between the true positives from class 1 and the false positives from class 2, i.e., classes 1 and 2 will not be separable. This separability assumption has long been discussed in the finite mixture model literature \citep{teicher1963identifiability, yakowitz1968identifiability}, but has not been explicitly discussed in the context of 
quantification. 

Under these two assumptions, we have the following result asserting that, with enough data, the loss function is minimized close to the true parameter value. 
\begin{theorem}[Identifiability for quantification learning]\label{lem:id} Let $f_N(\btheta) = ( \ell_\calU(\bM,\bp) + \ell_\calL(\bM))/N$, and $f(\btheta)= E_{\ba \in \calU} \left[D_{KL}(\mathbf{a}|| \mathbf{M}^{'}\mathbf{p})\right] + \xi E_{(\ba,\bb) \in \mathcal{L}}\left[D_{KL}(\mathbf{a} || \mathbf{M}^{'}\mathbf{b})\right]$ where $\xi=\lim n/N$. 
	Under Assumptions 1 and 2, for any $\eps>0$ there exists $\kappa > 0$ such that, the following holds for $f_N$. 
	\begin{enumerate}[(i)]
		\item $\liminf_{N}\inf_{\|\btheta-\btheta^0\|_1 > \eps}f_{N}(\tM, \tp) - f(\tM^{0}, \tp^{0}) \geq \kappa$ a.s. 
		\item $\liminf_N P(\inf_{\|\btheta-\btheta^0\|_1 \geq \eps}\, f_{N}(\tM, \tp) - f_N(\tM^{0}, \tp^{0}) \geq \kappa) = 1$. 
	\end{enumerate} 
\end{theorem}

The formal proof is provided in the appendix, we briefly outline the ideas used. 
We can write  $\nu_N(\btheta) \propto \exp(-\alpha\elll(\tM)-\alpha\ellu(\btheta))\Pi(\tp,\tM)$ where the subscripts $n$ and $N$ are added to indicate dependence of $\ell_\calL$, $\ell_\calU$ and $\nu$ on the sample size. 
When the loss-functions $f_N$'s are convex, and converges pointwise to some $f$ one can use the rich theory of convex functions to establish identifiability by showing that $f$ is minimized at the true value $\btheta^0$. 
In our case, $f_N=(\elll+\ellu)/N$ converges point-wise to $f=\xi E_\calL(D_{KL}(\ba || \bM'\bb)) + E_\calU(D_{KL}(\ba || \bM'\bp))$ where $\xi = \lim n/N$. However, neither $f_N$'s nor $f$ is convex because of the $\bM'\bp$ term, ruling out direct application of this result. 

We first show in Lemma \ref{lem:fnl} of the appendix that $\ell_\calL$ is convex in $\tM$ and use the convexity results to show that $\ell_\calL$ can identify $\tM^0$, i.e., outside of any neighborhood around the true value $\bM^0$, the empirical loss-function $\elll/n$ has higher value than the limiting loss-function $E_\calL(D_{KL}(\ba || \bM'\bb))$. 
A complementary result to this is that of weak identifiability of $(\tM^0,\tp^0)$ from the non-convex $\ell_\calU$. 
Lemmas \ref{lem:fu} and \ref{lem:fnu} of the appendix state that $\btheta^0=(\tM^0,\tp^0)$ is one of the minimizers of the loss function $\ell_\calU$ and its limit $E_\calU(D_{KL}(\ba || \bM'\bp))$. 
Combining, these results for the two losses, we have that for any $\btheta=(\tM,\tp)$, $f_N(\btheta)$ is greater than $f(\btheta^0)$ unless $\tM$ lies in an infinitesimally small neighborhood around $\tM^0$, and $\btheta$ is also a minimizer of $\ell_\calU$. Thus, use of the local labeled set $\calL$ via the loss function $\elll$ helps to identify $\bM$, as the posterior is guaranteed to concentrate around $\bM^0$. As $\bM$ concentrates around $\bM^0$, the loss $\ellu(\bM,\bp)$ becomes capable of identifying $\bp^0$ as $(\tM^0,\tp^0)$ is a minimizer of $\ell_\calU$ and $\bM$ is non-singular by 
the separability assumption. 

Subsequent,to establishing identifiability, one can leverage general results of Gibbs posteriors \citep{chernozhukov2003mcmc,miller2019asymptotic} to show posterior asymptotic consistency of the Gibbs posterior for GBQL (Theorem \ref{th:post}), and asymptotic normality of the Gibbs posterior mean (Theorem \ref{th:norm}). 
\begin{theorem}[Posterior consistency]\label{th:post}
	Let $B_\eps(\btheta^0)$ be an $\ell_1$ ball of radius $\eps$ around $\btheta^0$,  and $\Pi(\bp,\bM)$ be any prior which gives positive support to $B_\eps(\btheta^0)$ for any $\eps >0$. Then, under assumptions 1-2, as $N,n \rightarrow \infty$ and $n/N$ to some limit, 
	for any $\eps > 0$, $P_{\nu_N}(B_\eps(\btheta^0)) \rightarrow 1$. 
\end{theorem}

\begin{theorem}[Asymptotic normality]\label{th:norm} The mean $\hat\btheta$ of  the Gibbs posterior distribution (\ref{eq:nu}) of $\btheta=(\tM, \tp)$ is asymptotic normal i.e., 
	$\sqrt N \bOmega(\btheta^0)^{-1/2}\bJ(\btheta^0)(\hat \btheta - \btheta^0) \to_d N(0,\bI)$ for positive definite matrices $\bOmega(\btheta^0)$ satisfying $\bOmega(\btheta^0)^{-1/2}\nabla f_N(\btheta^0)/\sqrt N \to_d N(0,\bI)$ and $\bJ(\btheta^0)=\nabla^2_{\btheta^0} f$.
\end{theorem}

The proof of the results are provided in the Supplement. Proof of existence and positiveness of the matrices $\bOmega(\btheta^0)$ and $\bJ(\btheta^0)$ of Theorem \ref{th:norm} are provided in Lemma \ref{lem:f} of the Supplement.

\subsection{Coverage of Interval estimates}\label{sec:cp}

Gibbs posteriors or even full posteriors under model misspecification generally do not offer well calibrated credible intervals \citep{kleijn2012bernstein}. Calibration of credible intervals can be guaranteed under {\em generalized information equality}, i.e., when $\bOmega(\btheta^0) = \bJ(\btheta^0)$ in Theorem \ref{th:norm} \citep{chernozhukov2003mcmc}. This equality is satisfied for correctly specified likelihoods and other estimators like generalized method of moments, but typically is not satisfied for fractional posteriors or M-estimation type approaches like ours. Consequently, there is a large body of work on choice of the scaling parameter $\alpha$ (also known as the {\em learning rate} or {\em inverse temperature}) in (\ref{eq:nu}) to ensure desirable properties of interval estimates from Gibbs posteriors. 

\cite{bissiri2016} proposed choosing $\alpha$ by matching expected unit losses from the data and the prior. This strategy does not work with flat priors. Other strategies considered include endowing $\alpha$ with a hyper-prior which requires choosing the hyper-parameters, or using a prior that has conjugate structure to the loss which requires establishing the equivalence of prior loss with $m$ units of data loss. 
\cite{holmes2017assigning} considered the special case of power likelihoods and developed related ideas of choosing the power parameter based on matching expected information from Gibbs power posterior with that from a full Bayesian update. This relies on knowledge of a full parametric model for a Bayesian update and does not generalize from power posteriors to loss-function based Gibbs posteriors. More importantly, choosing $\alpha$ based on information matching offer no guarantee about calibration of credible intervals.

Choice of $\alpha$ is also discussed in the PAC-Bayes literature \citep{guedj2019primer}. While theoretical bounds often use some oracle value of $\alpha$, practical strategies include cross-validation which can be computationally demanding, or integration over $\alpha$. SafeBayes \citep{grunwald2011safe,grunwald2012safe,grunwald2017inconsistency} recommend optimizing over $\alpha$ using expected posterior loss or posterior predictive loss from power likelihoods that may not generalize to arbitrary loss functions. Again, these strategies provide no guarantees about coverage probabilities of the  credible intervals. 

\cite{syring2019calibrating} calibrates $\alpha$ directly  using bootstrapped data distributions to get the desired coverage of credible intervals. However, the algorithm requires posterior sampling for each choice of $\alpha$, each confidence level, and each bootstrap sample. 
Also their theory proves existence of an $\alpha^*$ which ensures calibrated intervals, but does not guarantee reaching $\alpha^*$ with the proposed iterative algorithm. 

As an alternative to choosing $\alpha$ to get well-calibrated credible intervals, \cite{chernozhukov2003mcmc} proposed fixing $\alpha$ and obtaining delta-method style  confidence intervals around the Gibbs posterior mean. These guarantee well-calibrated coverage probability of all functionals of parameters for any confidence level. 

For GBQL, while we can use any of the aforementioned strategies for choosing $\alpha$, we adopt the delta-method approach of \cite{chernozhukov2003mcmc} due to its minimal computational overhead and established asymptotic coverage guarantee. We prove the following result on guaranteed coverage of all parameter functionals for GBQL using fully data-driven confidence intervals. The technical statement of the result is provided in Section \ref{sec:addth} of the supplement, and proved therein. 

\begin{theorem}\label{th:cp} 
	For any differentiable function $g(\btheta)$, one can compute $t(g,\calL,\calU)$ --- a deterministic function of $g$ and the data $(\calL,\calU)$, such that with $\widehat\btheta$ denoting the Gibbs posterior mean, and $C_{g,s}=z_{1-s/2}t(g,\calL,\calU)$ for any $0 < s <1$, where $z_s$ is the $s^{th}$ quantile of a standard normal variable, we have 
	$$P(g(\hat\btheta) - C_{g,s} < g(\btheta^0) < g(\hat\btheta) + C_{g,s} ) \to 1 - s.$$
\end{theorem}

\subsection{Finite-sample concentration rate}\label{sec:rate}
In addition to the asymptotic results above, we also prove a finite sample result on posterior concentration rate. 
Let $P_N^0$ denote the probability measure for generating $a_r \in \calU$, and $(a_r,b_r) \in \calL$. We define 
\begin{equation}\label{eq:alphad}
D_{N,\alpha}(\btheta,\btheta^0) = - 
\log \int 
\left(\frac{\tilde p_N(\btheta)}{\tilde p_N(\btheta^0)}\right)^\alpha 
dP_N^0 
\mbox{ where } \tilde p_N(\btheta)=\exp(-N f_N(\btheta)). 
\end{equation}

For the M-closed case, i.e., when the class $
\{\tilde p_N(\theta) \given \btheta \in \Theta\}$ contains the true data-generation model, $D_{N,\alpha}(\btheta,\btheta^0)$ is (upto a constant) the widely used R\'enyi-divergence. For the M-open case, i.e., when $
\{\tilde p_N(\theta) \given \btheta \in \Theta\}$ is a family of misspecified likelihoods, \cite{bhattacharya2019bayesian} used (\ref{eq:alphad}) as a valid divergence measure of $\btheta$ from $\btheta^0$ to derive finite-sample posterior concentration rates. Two conditions needed for the validity of (\ref{eq:alphad}) as a divergence were $\btheta^0$ being an interior point of $\Theta$, and that $\btheta^0$ minimizes -$\int \log(\tilde p_N(\btheta)) dP_N^0$.

For GBQL, we are using estimating equations (M-free case) for the compositional true and predicted labels. We also consider $\btheta$ to be an interior point of the parameter space, and under Assumptions 1 and 2, we show in Lemma \ref{lem:f} part (iv) that 
$$- \int \log(\tilde p_N(\btheta)) dP_N^0 = N \int f_N(\btheta) dP_N^0 = N f(\btheta)$$
is minimized at $\btheta^0$. This ensures validity of $D_{N,\alpha}(\btheta,\btheta^0)$ as a divergence. 
We now have the following finite sample result:
\begin{theorem}[Posterior concentration rate]\label{th:rate}
	Let $\eps:=\eps_N > 0$ be such that $B_\eps(\btheta^0)$ -- an $\ell_1$ ball of radius $\eps$ around $\btheta^0$ lies in the interior of $\Theta$, and $\Pi_N(\bp,\bM)$ be any (possibly $N$-dependent) prior that gives positive mass of atleast $\exp(-NR\eps)$ to $B_\eps(\btheta^0)$ for some universal constant $R$. Then, under Assumptions 1-2, 
	we have for any $\alpha \in (0,1)$, $D > 1$, and $t > 0$, 
	$$P_{\nu_N}\left(\frac{D_{N,\alpha}(\btheta,\btheta^0)}N \geqslant 
	(D+3t)R\eps\right) \leq \exp(-tNR\eps) \mbox{
		with } P_N^0\mbox{-probability atleast }1 - 
	\frac 2{NR\eps\min\{{(D-1+t)^2,t}\}}.$$
\end{theorem}
Theorem \ref{th:rate} establishes the (outer)-probability of $1-O(1/N)$ of the Gibbs posterior for GBQL concentrating in $\alpha$-divergence neighborhood of the truth at an exponential ($1-\exp(-O(N))$) rate. The rate is same as that for fractional posteriors established in \cite{bhattacharya2019bayesian}. We show that the prior mass condition of assigning atleast $\exp(-NR\eps)$ prior probability for the ball $B_\eps(\btheta^0)$ is satisfied by the Dirichlet priors for $\bM$ and $\bp$ for $\eps_N=O(\frac {\log{N}} N)$. This leads to the following {\em nearly parametric rate} (upto logarithmic terms) for the concentration of the posterior around $\alpha$-divergence neighborhoods. 

\begin{corollary}\label{cor:paramrate}
Using independent Dirichlet priors for $\bp$ and for (the rows of) $\bM$, we have for some universal constant $M$, 
$$P_{\nu_N}\left(\frac 1N D_{N,\alpha}(\btheta,\btheta^0) \geqslant 
M \frac{\log N}N \right) \to 0 \mbox{ in } P_N^0\mbox{-probability}.$$
\end{corollary}

\subsection{Ensemble GBQL}\label{sec:ensth} Finally, all theory extends to the ensemble quantification of Section \ref{incorporating-multiple-predictions}. One thing to note for ensemble GBQL is that the same dataset is used with different classifiers to get predictions. Hence, these $K$ sets of predictions will not be independent. Also, the labeled data $\bb_r$ for $r \in \calL$ is going to be the same one used in the loss function for each classifier. The theory accommodates these dependencies. We state the result informally here, and provide the technical version in Section \ref{sec:addth} and the proof in Section \ref{sec:proofs}. 
\begin{corollary}[Ensemble GBQL]\label{cor:ens} If there are $K$ predictions are available for each instance from $K$ classifiers, and Assumptions 1 and 2 are satisfied for each classifier, then with  $\btheta=(\tM^{(1)},\ldots,\tM^{(K)}, \tp)$ we can establish posterior consistency, asymptotic normality, coverage of confidence intervals, and posterior concentration rate results for ensemble GBQL analogous to Theorems \ref{th:post} - \ref{th:rate}. 
\end{corollary}

\hypertarget{a-gibbs-sampler-for-the-posterior-belief-distribution}{%
	\section{Gibbs Sampler using rounding and coarsening}\label{a-gibbs-sampler-for-the-posterior-belief-distribution}}

We first outline the Gibbs sampler steps when only using one classifier. 
The sampler for ensemble GBQL is detailed in the Supplement. The Gibbs posterior 
\(\nu\) is given by

\begin{equation}\label{eq:gibbspost}
\nu \propto \left[\prod_{r \in \mathcal{U}} \prod_{j=1}^{C}\left(\sum_{i}p_{i}M_{ij} \right)^{a_{rj}}\prod_{r \in \mathcal{L}} \prod_{j=1}^{C}\left(\sum_{i}b_{ri}M_{ij} \right)^{a_{rj}}\right]\pi(\mathbf{p}, \mathbf{M}).
\end{equation}

When all $\ba_r$ are categorical, the polynomial expansion of $(\sum_i p_i M_{ij})^{\sum_r a_{rj} }$ enabled an efficient latent variable Gibbs sampler in \cite{datta2018local}. When $a_{rj}$ are fractions, this advantage is lost as fractional polynomials do not have such convenient expansions. Additionally, since we now allow uncertainty in the true labels, we also need to consider the extra fractional expansion terms $(\sum_{i}b_{ri}M_{ij})^{a_{rj}}$. 

We can use any of-the-shelf sampler to generate samples from (\ref{eq:gibbspost}). However, to enable fast and efficient sampling, we propose a data-augmented Gibbs sampler. We first switch from $\nu$ to $\nu_{round}$ where the 
probabilistic output $a_{rj}$ is replaced by $\ceil{T a_{rj}}$ where \(T\)
is an integer, and $\ceil{\cdot}$ denotes the ceiling of any real number. 
Consider now the following generative model:  

\begin{equation}\label{eq:round}
\begin{aligned}
&z_{rt} \overset{\text{ind}}{\sim} \begin{cases}
Multinomial(1, \mathbf{p}) & \text{if }r \in \mathcal{U}\\
Multinomial(1, \bb_{r})  &\text{if }r \in \mathcal{L}
\end{cases},\ t=1,\ldots, T_r=\sum_j \ceil{Ta_{rj}}\\
& d_{rt} | z_{rt} = i \overset{\text{ind}}{\sim} Multinomial(1, \mathbf{M}_{i*}), r \in \calL \cup \calU
\end{aligned}
\end{equation}

The rounded generalized posterior $\nu_{round}$ is then the proper Bayesian posterior using the likelihood $p(\bd^\calU, \bd^\calL \given \bb^\calL, \bM, \bp)$ for any realization of $\bd_{rt}$'s satisfying $\sum_{t} I(\bd_{rt}=j)=\ceil{Ta_{rj}}$. To obtain samples of $\bp$ and $\bM$ from $\nu_{round}$, instead of using this marginalized likelihood, we can equivalently introduce $\bz^\calL$, and $\bz^\calU$ as latent variables and use the joint likelihood $p(\bd^\calU, \bd^\calL, \bz^\calL, \bz^\calU \given \bb^\calL, \bM, \bp)$. This joint likelihood decomposes nicely and will be conducive to a Gibbs sampler with standard Dirichlet priors on $\bM$ and $\bp$. 

However, since we artificially inflate sample size by an order of $T$ by switching from $\ba_r$ to $\ceil{T\ba_{r}}$, instead of sampling from $v_{round}$ we sample from the coarsened likelihood 
\begin{equation}\label{eq:coarse}
\nu_{coarse} \propto p(\mathbf{d}^{\mathcal{U}}, \mathbf{d}^{\mathcal{L}}| \bb^{\mathcal{L}},\mathbf{M}, \mathbf{p})^{\frac{1}{T}}\pi(\mathbf{p}, \mathbf{M})
\end{equation}

As
\( \nu_{round} = p(\mathbf{d}^{\mathcal{U}}, \mathbf{d}^{\mathcal{L}}| \bb^{\mathcal{L}},\mathbf{M}, \mathbf{p})\)
is a proper likelihood arising from a mixture of categorical distributions, \(\nu_{coarse}\) can be expressed as a
fractional (coarsened) posterior \citep{bhattacharya2019bayesian,ibrahim2015power}. The Conditional Coarsening
Algorithm \citep{miller2019robust} introduces latent
variables to device Gibbs samplers for such coarsened posteriors for mixture likelihoods, just as one would do for proper posteriors. While the coarsened posterior has not been established to be exactly equal to the stationary distribution of the Conditional Coarsening Gibbs sampler, \cite{miller2019robust} has provided heuristic justification for the algorithm by drawing connection with geometric mean of posteriors based on subsampled data. Empirical evidence comparing conditional coarsening with direct importance sampling. As our $\nu_{round}$ is also a proper mixture likelihood, we outline below a conditional coarsening-based Gibbs sampler algorithm for sampling from $\nu_{coarse}$. Our own simulations, detailed in Section \ref{sec:simstan} reinforces the claim of \cite{miller2019robust} about the accuracy of conditional coarsening to sample from coarsened posteriors. The GBQL conditional coarsening Gibbs sampler generates posteriors nearly indistinguishable from direct samples from the coarsened posterior, while being substantially faster. 

We use generic Dirichlet priors $\bM \sim Dirichlet(\bV)$, i.e,  $\bM_{i*} \ind Dirichlet(\bV_{i*})$, and $\bp \sim Dirichlet(\bv)$ where $\bV$ and $\bv$ respectively are a matrix and a vector of positive hyper-parameters Specific choices with desirable shrinkage properties are discussed in Section \ref{sec:shrink}. This gives the following Gibbs updates:
\begin{align*}
\bz_{r} | \cdot &\sim \begin{cases} Multinomial\left(1, \frac{1}{\sum_{i}M_{ij}p_i}(M_{1j}p_{1},\ldots, M_{Cj}p_{C})\right),\ r \in \mathcal{U},\ d_{rt} = j \\
Multinomial\left(1, \frac{1}{\sum_{i}M_{ij}b_{ri}}(M_{1j}b_{r1},\ldots, M_{Cj}b_{rC})\right),\ r \in \mathcal{L},\ d_{rt} = j
\end{cases}\\
M_{i} | \cdot &\sim Dir\left( \tilde{\text{V}}_{i1}, \ldots, \tilde{\text{V}}_{iJ}\right),\
\tilde{\text{V}}_{ij} = \text{V}_{ij} + \frac{1}{T}\left(\sum_{r \in \mathcal{U}, \mathcal{L}}\sum_{t=1}^{T}I(d_{rt}=j)I(z_{rt}=i)\right)\\
p | \cdot &\sim Dir\left(\tilde{\text{v}}_{1}, \ldots, \tilde{\text{v}}_{C}\right),\ \tilde{\text{v}}_{i} = \text{v}_{i} + \frac{1}{T} \left(\sum_{r \in \mathcal{U}}\sum_{t=1}^{T}I(z_{rt}=i)\right).
\end{align*}

If there are hyper-parameters $\bgamma$ in $\bV$ and $\bv$ that need to be assigned a prior, they can be sampled using a Metropolis-Hastings step.
We note that
the full conditional distributions for the \(z_{rt} \) for
\(r \in \mathcal{U}, d_{rt} = j\) are identical, which enables them to be jointly
sampled. Furthermore, the \(z_{rt}\) for \(r \in \mathcal{L}\) do not
need to be updated if there is a \(i\) such that \(b_{ri} = 1\). 

\subsection{Choice of coarsening factor}\label{sec:factor}
Our Gibbs sampler relies on rounding and coarsening $\nu$ using an integer factor $T$. In this Section we discuss theory guiding the choice of this factor $T$. 
\cite{kessler2015bayesian} have used a similar data-augmented Gibbs sampling approach for compositional regression. However, their approach only incorporates the rounded likelihood for the pseudo-data $\mathbf{d}_{r}$ and does not coarsen. Rounding inflates the sample size by a factor of $T$ resulting in underestimation of the posterior variance and the coarsening is needed to adjust for this. We first show that the coarsening adjustment by a factor  $T$ growing with the sample-size ensures asymptotic equivalence of the rounded and coarsened posterior $\nu_{coarse}$ with the original posterior $\nu$. 

We denote the coarsened and rounded version of the loss function $f_{N}$ using a factor $T$ 
as $\tilde{f}_{N}$. As for $x \geq 0$, $0 \leq \frac{\ceil{Tx}}T - x \leq \frac 1T$, with $n=\xi N$, we have 
{\small \begin{align*}
	|\tilde{f}_{N}(\tM, \tp) - f_{N}(\tM, \tp)| = &-\frac{1}{N}\sum_{r=1}^{N}\sum_{j=1}^{C}\left( T \ceil[\Big]{\frac{a_{rj}}{T}} - a_{rj} \right)\log\left(\sum_{i=1}^{C}M_{ij}p_{i}\right)-\frac{1}{N}\sum_{r=1}^{\xi N}\sum_{j=1}^{C} \left( T \ceil[\Big]{\frac{a_{rj}}{T}} - a_{rj} \right) \log\left(\sum_{i=1}^{C}M_{ij}b_{ri}\right) \\
	& \leq - \frac 1T \sum_{j=1}^C \log\left(\sum_{i=1}^{C}M_{ij}p_{i}\right) - \frac \xi T \sum_{j=1}^C  \log\left(\min_{i} M_{ij}\right). 
	\end{align*}}
Hence, for any $\bM,\bp$ on the interior of the parameter space, we have $\tilde f_N$ goes to the same point-wise limit $f$ as $T=T_{N} \rightarrow \infty$. This immediately leads to the analogues of the asymptotic results of Theorems \ref{th:post}-\ref{th:cp} for $\nu_{coarse}$. 

\begin{corollary}\label{cor:coarse} Let $\nu_{coarse,N}$ denote the rounded and coarsened generalized posterior using a factor $T_N$ with $T_N \rightarrow \infty$. Then, under Assumption 1 and 2, we have 
	\begin{enumerate}[(i)]
		\item $P_{\nu_{coarse,N}(\btheta)}(B_\eps(\btheta^0)) \rightarrow 1$ for any $\ell_1$ ball $B_\eps(\btheta^0)$ of radius $\eps$ around $\btheta^0$ for any $\eps >0$ such that the prior $\Pi(\bp,\bM)$ be gives positive support to $B_\eps(\btheta^0)$.
		\item The Gibbs posterior mean $\hat\btheta_{coarse}$ using $\nu_{coarse}$ is asymptotic normal i.e., 
		$\sqrt N \bOmega_N(\btheta^0)^{-1/2}\bJ(\btheta^0)(\hat \btheta_{coarse} - \btheta^0) \to_d N(0,\bI)$ where $\bJ(\btheta^0)=\nabla^2_{\btheta^0} f$, and   $\bOmega_N(\btheta^0)$ is a positive definite matrix created by replacing the population distribution of $\ba_r$ with that of $\ceil{T_N\ba_r}/T_N$ in $\bOmega(\btheta^0)$ of Theorem \ref{th:norm}. 
		\item For any differentiable function $g(\btheta)$, one can compute interval estimates of $g(\btheta)$ with valid asymptotic coverage in the same way as Theorem \ref{th:cp} using the same $\widehat \bJ$ and and $\widehat \bOmega_N$ that by replaces the samples $\ba_r$ with $\ceil{T_N\ba_r}/T_N$ in the estimate $\widehat\bOmega$.
	\end{enumerate}
\end{corollary}

The asymptotic results only need $T_N \to \infty$. For practical guidance on the choice of $N$, we now look at how the finite sample posterior concentration rate for $\nu_{coarse}$ compare with that of $\nu$  
in Theorem \ref{th:rate}. 

\begin{theorem}[Coarsened posterior concentration rate]\label{th:coarserate}
	Let $\nu_{N,T}$ denote the coarsened posterior of (\ref{eq:coarse}) with rounding and coarsening factor $T=O(N^\beta)$ for any $\beta \geq 1$. Then under Assumptions 1-2 and the conditions of Theorem \ref{th:rate} we have with universal constants $R$ and $R_0 \geq 1$, 
	{\small $$P_{\nu_{N,T}}\left(D_{N,\alpha}(\btheta,\btheta^0) \geqslant 
		(D+3t)NR\eps\right) \leq \exp(-tNR\eps) \mbox{
			with } P_N^0\mbox{-probability atleast }1 - \frac {1+R_0^\frac{\alpha N}{T}}{NR\eps\min\{{(D-1+t)^2,t}\}}.$$}
\end{theorem}

It is clear that when $T=O(N^\beta)$ for $\beta \geq 1$, the coarsened posterior concentrates at the same rate (upto a constant) around the true value with the same probability as the uncoarsened posterior, while when $T=\infty$, then the probability in Theorem \ref{th:coarserate} is same as that is Theorem \ref{th:rate} which is not surprising as the coarsened posterior with $T=\infty$ is the original uncoarsened posterior. Larger $T$ would involve an creating and updating a larger dimensional pseudo-data in the Gibbs sampler. Hence we recommend using $T=O(N)$ the smallest scaling which ensures same concentration rate as the original posterior. 

\hypertarget{simulations}{%
	\section{Simulations}\label{simulations}}

We conduct multiple simulation studies to 
assess a) 	accuracy of GBQL in estimating \(\mathbf{p}\) in the
presence of moderate amounts of labeled data, b) 
compare robustness of estimation-equation-based GBQL to  Dirichlet model-based approach using different data generating mechanisms, c) compare	computation efficiency compared to Dirichlet models, and d) assess  
estimation accuracy 
when there is uncertainty in true labels in \(\mathcal{L}\).

To mimic the real data application we present in Section \ref{phmrc-dataset-analysis}, we used \(N=1000\), \(n=300\), \(C=5\), \(\mathbf{p}_{\mathcal{L}} = E_\calL(y_r) = (\frac{1}{C}, \ldots, \frac{1}{C})\)', and the following four
different values of \(\mathbf{p}\) representing each of the four countries in the PHMRC
dataset (Section \ref{phmrc-dataset-analysis}). The values of $\bp$ and $\bM$ are presented below. 
\begin{equation}\label{eq:simp}
\begin{aligned}
\mathbf{p1} &= (.20, .19, .27, .27, .07)\\
\mathbf{p2} &= (.11, .11, .40, .29, .09)\\
\mathbf{p3} &= (.09, .18, .52, .19, .02)\\
\mathbf{p4} &= (.13, .30, .35, .19, .03)
\end{aligned},\, \bM=
\begin{bmatrix}
0.65&0.35&0&0&0 \\
0&0.35&0.65&0&0 \\
0.1&0.1&0.6&0.1&0.1 \\
0&0&0&0.8&0.2 \\
0&0.4&0&0&0.6 \\
\end{bmatrix}.
\end{equation}
We generated true labels 
$y_{r} | \mathbf{p} \sim Multinomial(1, \mathbf{p}),\ r \in \mathcal{U}$ and $y_{r} | \mathbf{p}_\calL \sim Multinomial(1, \mathbf{p}_{\mathcal{L}}),\ r \in \mathcal{L}
$.
For the first analyses, we allow for full knowledge of these labels for $r \in \mathcal{L}$, which means that $\mathbf{b}_{r} | y_{r} = i$ equals $\be_{i}$ for $r \in \mathcal{L}$.
We then simulated outputs \(\mathbf{a}_{r} | y_{r}\) directly
from a model, so that we know the true data generating mechanism of the
dataset shift. We use two data generating mechanisms for \(\mathbf{a}_{r} | y_{r}\). The first mechanism corresponds to a zero-inflated Dirichlet mixture model:
\begin{align*}
a_{rj}^{*} | y_{r} = i, \mathbf{M}_{i*} \sim \begin{cases}
0,\ \text{if } M_{ij} = 0\\
Gamma(5M_{ij}, 1),\ \text{else}
\end{cases}\ j=1,\ldots, C, \,\,
a_{rj} = \frac{a_{rj}^{*}}{\sum_{k=1}^{C}a_{rk}^{*}} \;.
\end{align*}

The second data generating mechanism introduced subject-level over-dispersion, replacing the Gamma shape parameter $5M_{ij}$ above with $\tau_r M_{ij}$, where $\tau_r$ is subject-specific over-dispersion generated from the mixture distribution 
$\tau_{r} \sim .5 \cdot \mbox{Unif}(.1, 1) + .5 \cdot \mbox{Unif}(10,20)$. 
Instances with $\tau_{r} \leq 1$ will have responses $a_{rj}$ close to 0 and 1, while instances with large $\tau_{r}$ will have $a_{rj}$ clustered closer to the non-zero entries of $\mathbf{M}$.

We GBQL estimates of \(\mathbf{p}\) with estimates from following Bayesian
Dirichlet mixture model which assumes the first data generating mechanism as truth
$$
y_{r} | \mathbf{p} \sim Multinomial(1, \mathbf{p}),\,\,
\mathbf{a}_{r} | y_{r} = i \sim Dirichlet(\tau_{i} \cdot \mathbf{M}_{i*}),\,\,
\tau_{i} \sim Normal(0, 25).
$$
For both the Dirichlet model and GBQL, we used Dirichlet priors for $\mathbf{M}$ shrinking towards $\bI$, and uninformative Dirichlet prior for $\mathbf{p}$. 
Since the Dirichlet distribution does not support zeros, for running the Dirichlet model, 0 values were replaced with $\epsilon = .001$ and each $\mathbf{a}_{r}$ was re-normalized. Posterior sampling for this model was performed using RStan Version 2.19.2 \citep{rstan}. 
Note that this model becomes misspecified for the second true data generating mechanism. For both models, we
ran three chains each with a total of 6,000 draws and a burn-in of 1,000
draws. We used the posterior mean of \(\mathbf{p}\) as
\(\hat{\mathbf{p}}\).

To compare estimates of \(\mathbf{p}\), we use a chance corrected
version of the normalized absolute accuracy (NAA) \citep{gao2016classification} for estimating a compositional vector. NAA is defined as
\[
1- \frac{\sum_{i=1}^{C}|p_{i} - \hat{p}_{i}|}{2(1-min_{i}\{p_{i}\})}.
\]
To represent random guessing of \(\mathbf{p}\) with a score of 0, and
perfect estimation of \(\mathbf{p}\) with a score of 1, we follow
\cite{flaxman2015measuring} and use the Chance Corrected NAA (CCNAA) = $(NAA - .632)/(1 - .632)$.

We repeat our simulations 500 times for each choice of $\bp$ and show the average CCNAA across
this simulations in Figure \ref{fig:unnamed-chunk-2}. For case 1 (left panel) when the
likelihood is correctly specified for the Dirichlet model, both methods produce accurate estimates of \(\mathbf{p}\) and have approximately the same CCNAA.  When we introduce over-dispersion to the distribution of
the $\mathbf{a}_{r} | y_{r}=i$ (right panel), we see that the performance the GBQL model is hardly affected, and substantially outperforms the now misspecified Dirichlet model in all
cases. 
\begin{figure}[]
	\centering \includegraphics[width=.75\linewidth]{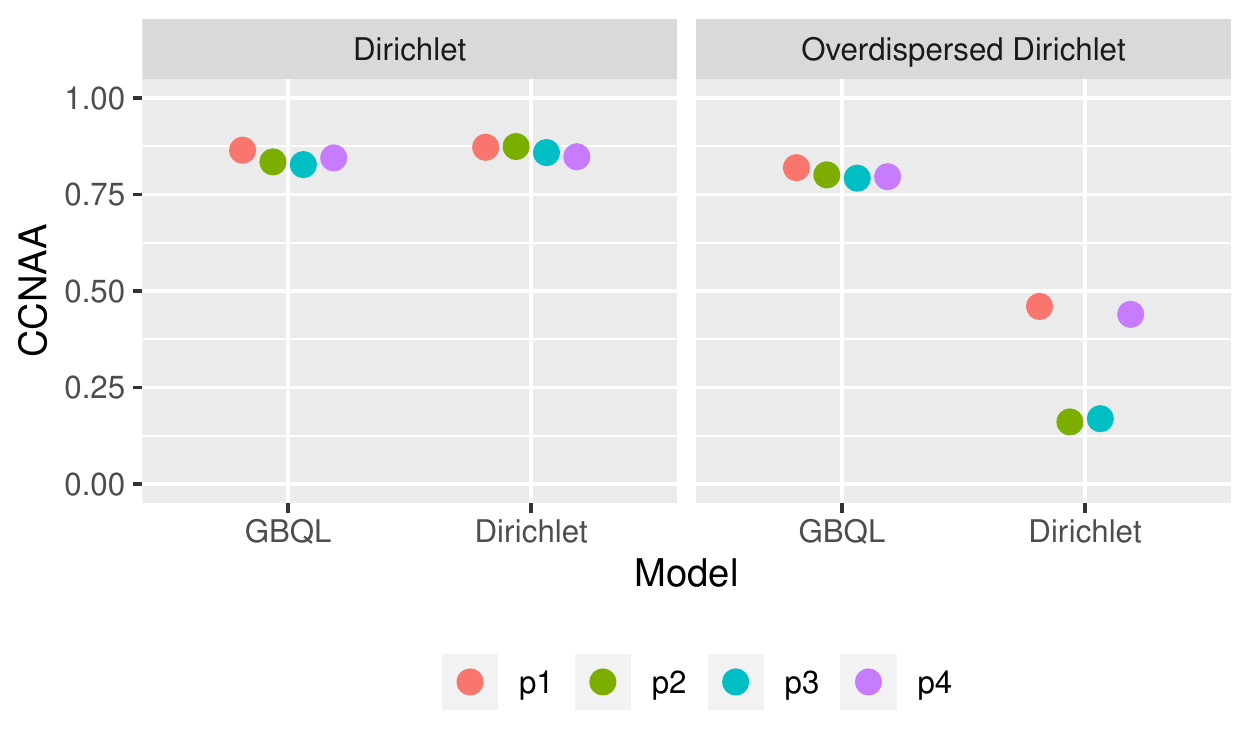} 
	\caption{Quantification performance in simulated data. Columns shows results for the two different data generating mechanisms, while each color represents each of the four scenarios (four true values of p). The GBQL model produces high values of CCNAA for each of the scenarios, while assuming a Dirichlet mixture model likelihood only produces acceptable estimates of $\mathbf{p}$ when the likelihood is correctly specified. 
	}\label{fig:unnamed-chunk-2}
\end{figure}

When we investigated the Stan output for the Dirichlet models, many of the chains
failed to converge when the likelihood was misspecified as indicated in the Gelman-Rubin diagnostic ($\hat R$) values in Table \ref{table:r_hat_runtime}. Furthermore, on average the Stan Dirichlet model took nearly 200 times longer to run than the GBQL method (Table \ref{table:r_hat_runtime}). For GBQL, all $\hat R$ values indicated convergence. Thus, GBQL
accurately estimates \(\mathbf{p}\), removes the need to correctly
specify the likelihood, is fast, and does not require
fine-tuning for the posterior samples to converge.

\begin{table}[ht]
	\resizebox{\textwidth}{!}{
		\centering
		\begin{tabular}{p{3cm}cccc}
			\hline
			Choice of $\mathbf{p}$ & Average $\hat{R}$ GBQL & Average $\hat{R}$ Dirichlet & Average Runtime (minutes) GBQL &  Average Runtime (minutes) Dirichlet \\ 
			\hline
			$\mathbf{p1}$ & 1.03 & 3.32 & 0.15 & 29.79 \\ 
			$\mathbf{p2}$ & 1.02 & 3.43 & 0.16 & 29.70 \\ 
			$\mathbf{p3}$ & 1.03 & 3.12 & 0.16 & 28.84 \\ 
			$\mathbf{p4}$& 1.03 & 3.46 & 0.15 & 29.88 \\ 
			\hline
	\end{tabular}}
	\caption{Average $\hat{R}$, as a measure of posterior sampling convergence, and runtime in minutes for each value of $\mathbf{p}$ was computed for when there is over-dispersion in the data generating mechanism. }\label{table:r_hat_runtime}
\end{table}
We now examine the behavior of the GBQL model in the case of
uncertain labels. To induce this uncertainty, we generate the compositional \(\bb_{r}\) from the following over-dispersed Dirichlet distribution $\mathbf{b}_{r} \sim Dirichlet(\tau_{r} \mathbf{p}),\ r \in \mathcal{U}$ and $\mathbf{b}_{r} \sim Dirichlet(\tau_{r} \mathbf{p_\calL}),\ r \in \mathcal{L}$ where $\tau_{r} \sim .5 \cdot \mbox{Unif}(.1, 1
) + .5 \cdot \mbox{Unif}(10,20)$,  
and generate $y_{r} | \mathbf{b}_{r} \sim Multinomial(1, \mathbf{b}_{r})$. The data generating process for the $\mathbf{a}_{r}$ is the same as in the simulations with known labels. The compositional $\bb_{r}$ are used as the uncertain labels for $r \in \mathcal{L}$. Figure \ref{fig:known_vs_unknown} plots the average CCNAA from GBQL with known labels $y$ against CCNAA of GBQL with unknown labels $\bb$ for each value of $\mathbf{p}$ and data generating mechanism. It can be seen that introducing uncertainty in the labels results in slightly lower (upto 10\%) CCNAA values indicating the small price we pay for the added uncertainty.

\begin{figure}
	\centering \includegraphics[width=.75\linewidth]{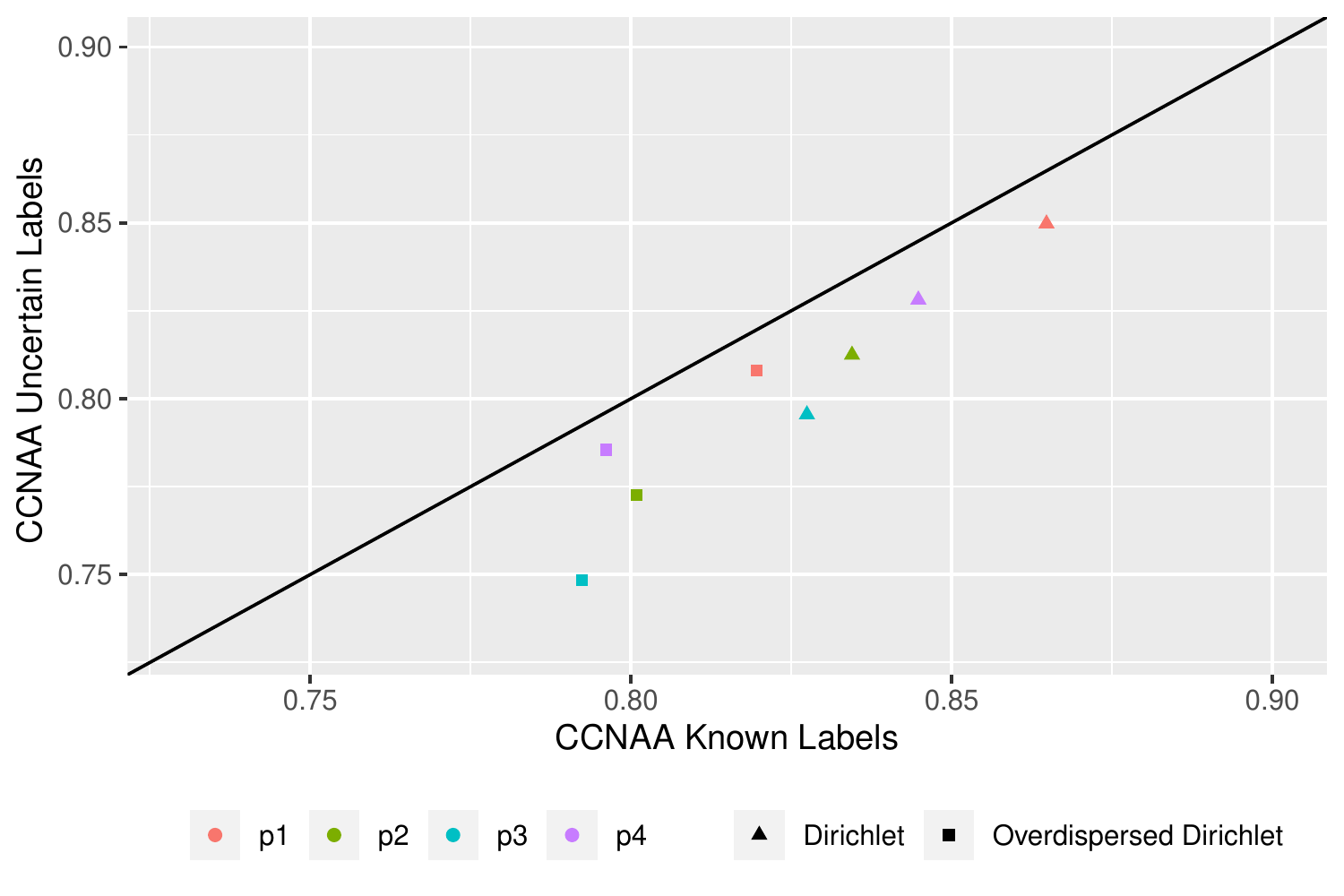}
	\caption{CCNAA for known versus uncertain labels using GBQL. Each color represents a different value for $\mathbf{p}$, while the shapes represent the two different data generating mechanisms.}\label{fig:known_vs_unknown}
\end{figure}

\subsection{Additional Simulation studies} We conducted additional simulation studies on comparison of the estimates, MCMC convergence, and computation time using the Gibbs sampler and direct implementation of GBQL in Stan, comparison of performance of our Gibbs sampler for different choices of the coarsening factor,  evaluation of coverage probabilities of our interval estimates, comparison of the sparse model of Section \ref{sec:para} with the full model. These are provided in Section \ref{sec:addsim} of the Supplement. 

\section{PHMRC Dataset Analysis}\label{phmrc-dataset-analysis}
High-quality cause-of-death information is lacking for 65\% of the world’s population due to scarcity of diagnostic autopsies in low-and middle-income countries (LMICs) \citep{nichols20182016}. So, estimating subnational and national cause-specific-mortality fractions (CSMF) and burden of disease numbers rely to a large extent on simple aggregation (classify-and-count) of verbal autopsy predicted cause-of-deaths. We apply GBQL to improve estimation of CSMF using predicted cause-of-death data from verbal autopsy classifiers. An example of dataset shift is in the Population Health Metrics Research Consortium (PHMRC) gold standard  dataset \citep{murray2011phmrc}, which contains 168 reported
symptoms and gold-standard underlying causes of death for adults in 4
countries. There are 21 total causes of death, that are then aggregated to 5 broader cause of death categories. Figure \ref{fig:phmrc_resp_rates} shows the percentage of subjects within each country 
and cause of death that report each symptom. The $x$-axis is an enumeration of the entire list of symptoms $x$ and the $y$-axis plots $p(x \given y)$ for each symptom $x$. With no dataset shift, we would expect the conditional response rates for each question within each cause of death to be similar for every country. However, as the country-specific lines are quite distinct in each sub-figure, it is clear that this assumption is violated. 
This leads to poor performance of verbal autopsy classifiers trained on symptoms and cause of
death labels from 3 countries to predict the cause of death distribution
for the remaining country \citep{mccormick2016insilico}. 

\begin{figure}[!ht]
	\centering \includegraphics[width=.85\linewidth]{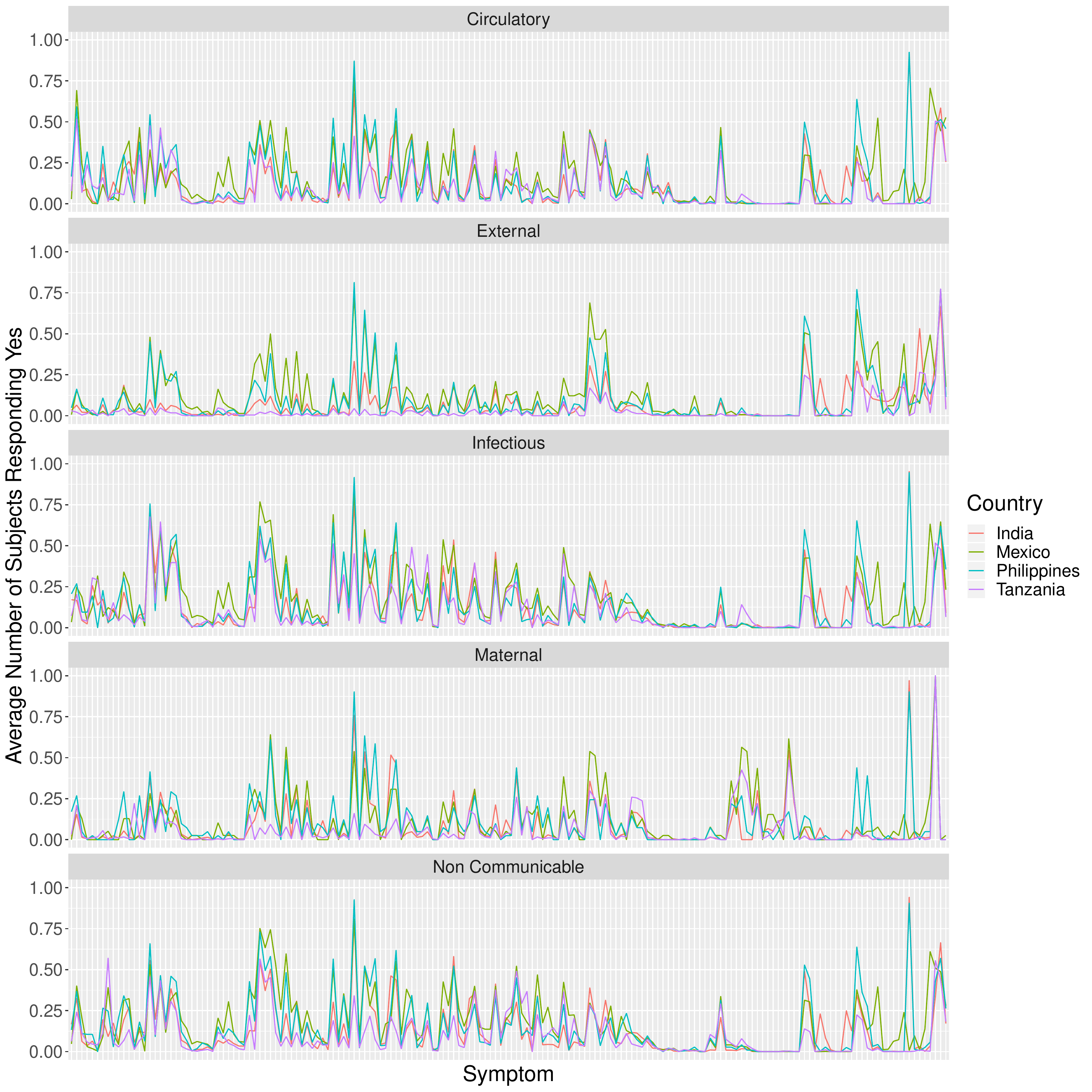} 
	\caption{Dataset shift in PHMRC verbal autopsy data. Percent of subjects with each of 168 reported symptoms within each of the 5 gold-standard underlying cause of death categories, by country.}\label{fig:phmrc_resp_rates}
\end{figure}

We now apply GBQL using limited local data from each country to improve CSMF estimation from verbal autopsy classifier trained on data from the other 3 countries. 
The number of observations within India, Mexico, Philippines, and
Tanzania are 2973, 1586, 1259, and 2023, respectively.  To address country-specific dataset shift, for each country,
we used the three remaining countries as
training data for four methods
commonly used for cause of death predictions: InterVA \citep{interva4},
InSilicoVA \citep{mccormick2016insilico}, NBC \citep{nbc}, and Tariff \citep{tariff2_0}. The first three methods are probabilistic, while Tariff produces a score for each cause that needed to be normalized to be in $[0,1]$. Model training was done using the openVA package version 1.0.8 \citep{openva}. We considered both compositional predictions and classifications (single-class categorical predictions based on the plurality rule). 
For GBQL in the test country, we then sampled labeled data $\calL$ of varying sizes (n=25, 100, 200, 400) to
investigate the effect of increasing the number of known labels. Sampling was performed such that $\bp_{\calL} = (\frac{1}{5}, \ldots, \frac{1}{5})$, as in Section \ref{simulations}. For comparisons, we obtained estimates using the Probabilistic Average \citep[PA,][]{bella2010quantification} method for compositional predictions, which should align with the GBQL estimate for $n=0$ (Section \ref{sec:shrink}) for our choice of priors, as well as estimates using the Adjusted PA method. We
repeated this 500 times for each size of \(n\). Results for the average
CCNAA when using compositional predictions are shown in Figure \ref{fig:phmrccomp}.

When no labeled instances are available, we see that the APA method performs worse than the PA method across almost all countries and algorithms, demonstrating why,  in presence of dataset shift, it is not appropriate to estimate $\bM$ using the training data. We see that obtaining $n=25$ labeled instances (an average of only 5 labeled deaths
per class) does not effectuate any improvement in the performance over not having any labeled test data ($n=0$). However, increasing this to
100 labels (an average of 20 labeled deaths per class) leads to large increase in CCNAA indicating substantial improvement in estimation of \(\mathbf{p}\) across all
countries and algorithms. As there are 168 covariates used for building
these classifiers, using just 100 observations to build a reliable
classifier would be difficult, if not impossible. Quantification accuracy
continues to increase with a larger number of labeled observations
across all countries and algorithms, although the extent of this
improvement is quite variable. Figure \ref{fig:catvscomp} compares the CCNAA for GBQL using compositional predictions versus GBQL using single-class categorical predictions. We see that using the original compositional scores offers improvement over categorization for all algorithms except Tariff for Philippines and Tanzania. 

\begin{figure}[]
	\centering{
		\begin{subfigure}{16cm}
			\centering\includegraphics[width=15cm]{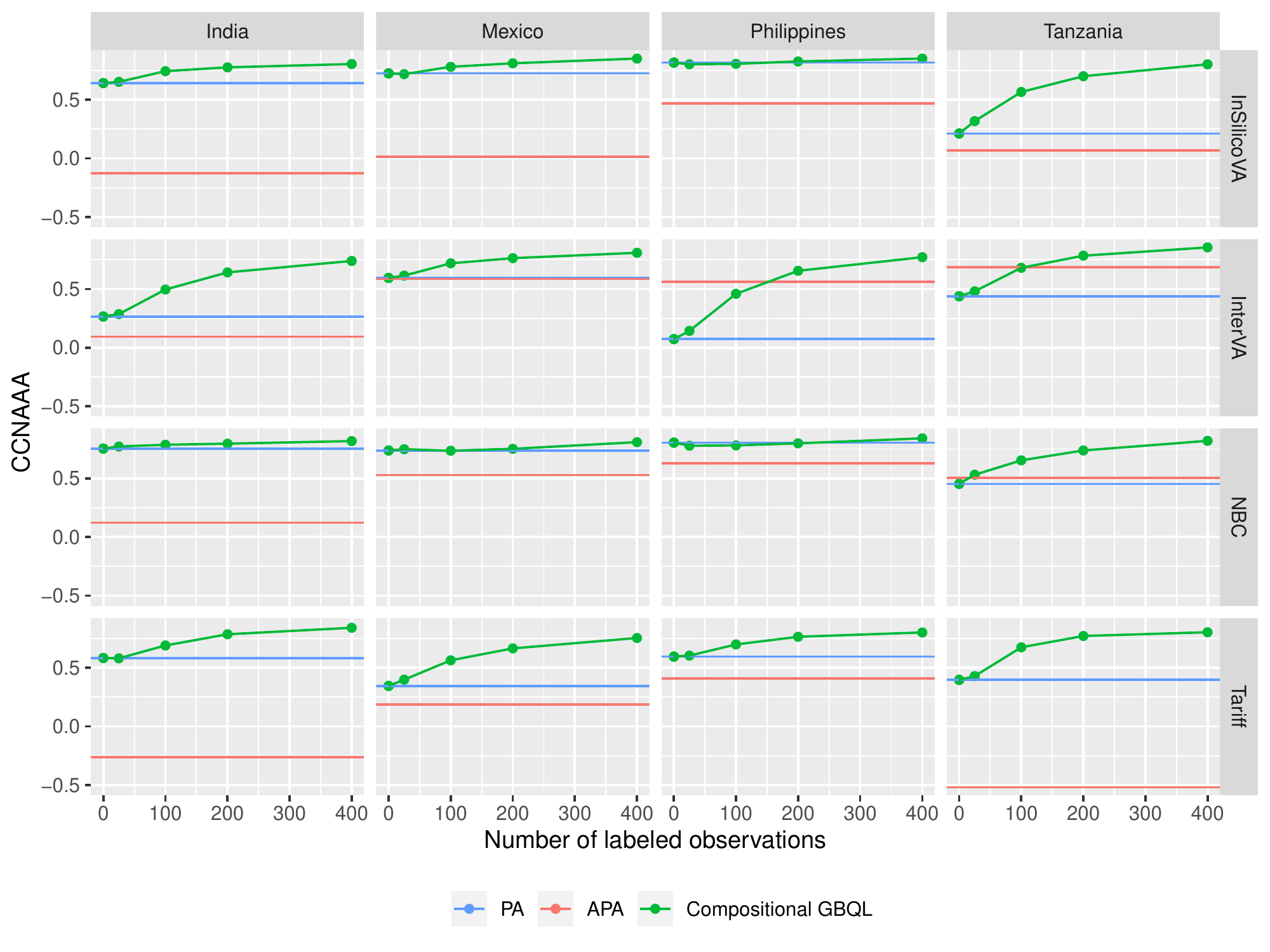}
			\caption{Average CCNAA for increasing numbers of labeled observations across all countries in the PHMRC dataset for GBQL using four common VA classifiers. Average CCNAA for GBQL using compositional predictions is shown in green. For comparison, performance of the PA (blue) and APA (red) methods are also shown.}\label{fig:phmrccomp}
		\end{subfigure}
		
		\begin{subfigure}{16cm}
			\centering\includegraphics[width=15cm]{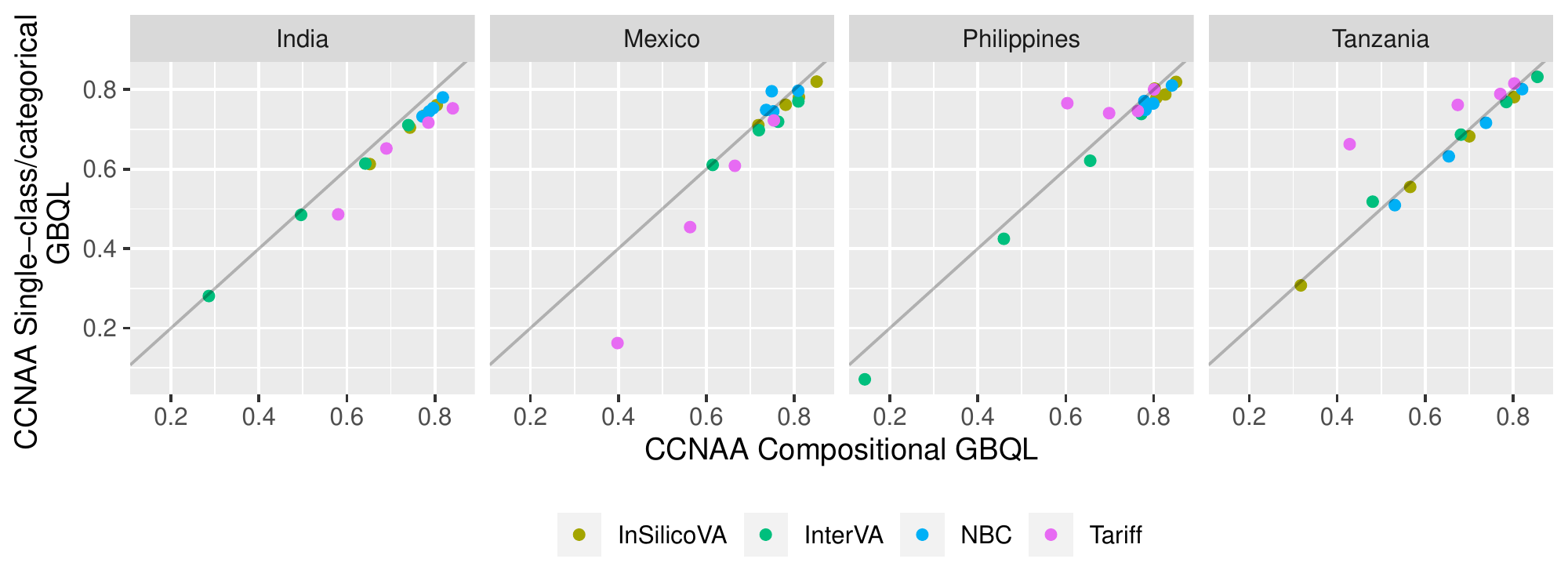}
			\caption{Comparison of CCNAA between GBQL using compositional predictions versus single-class/categorical predictions. Each point represents a different value of $n$, with the black line representing the identity line.}\label{fig:catvscomp}
		\end{subfigure}
	}
	\caption{PHMRC data analysis using different quantification methods. }
\end{figure}

Figure \ref{fig:phmrccomp} also shows that classifier performance varies widely across settings. For example among the PA estimates, InSilicoVA is best for Philippines, whereas NBC is most accurate for Tanzania. We now look at the performance of
our ensemble method which uses predictions from all four algorithms.
Figure \ref{fig:ensemble} shows the CCNAA for the ensemble GBQL and GBQL using individual classifiers, for different numbers of labeled observations and each country.
With only 25 labeled observations, the ensemble CCNAA is
approximately an average of the CCNAA for each of the other algorithms,
which is what we would expect, as for $n=0$ it is exactly the average as discussed in Section \ref{sec:shrink}. With more labeled observations, the
ensemble begins to either outperform all of the methods, or has 
CCNAA very close to that of the top performing method. Importantly, the
ensemble method significantly outperforms the worst method for all
combinations of country, output format and numbers of labeled observations, showing that including multiple
algorithms and using the ensemble quantification protects against inadvertently selecting the worst algorithm. 

\begin{figure}[]
	{\centering \includegraphics[width=1\linewidth]{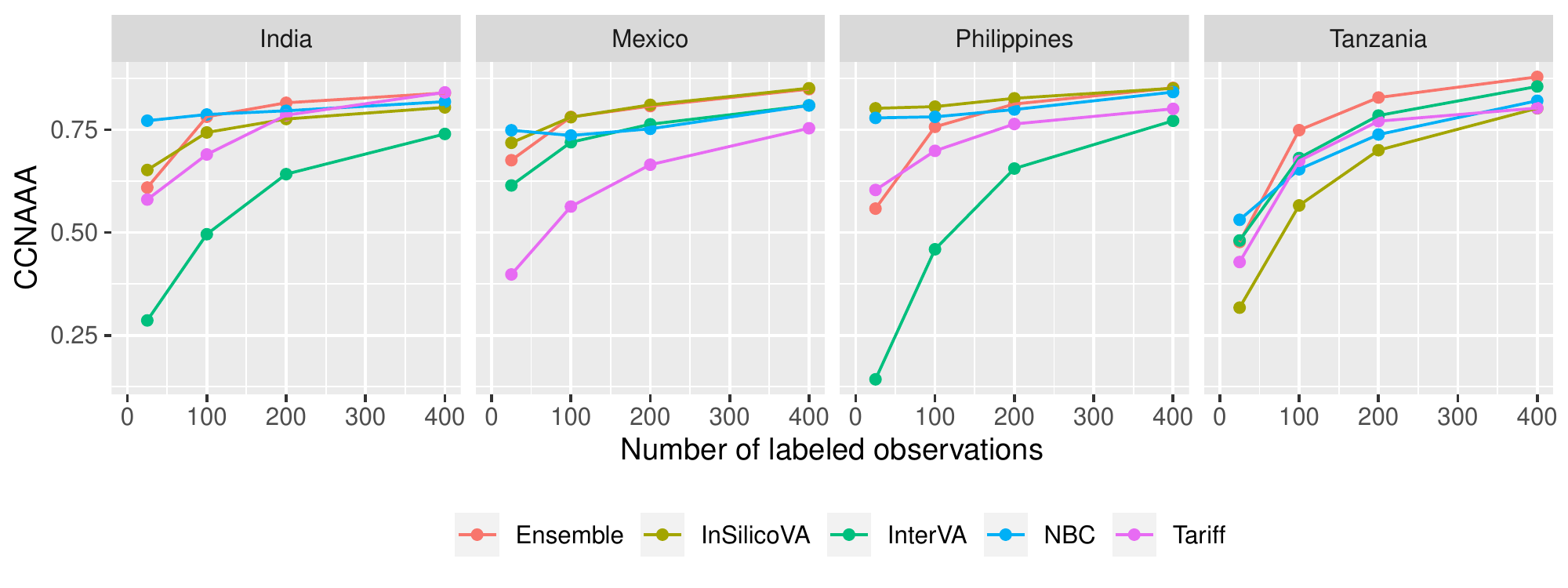} 
	}
	\caption{CCNAA comparing the ensemble GBQL (red) with GBQL using the 4 individual classifiers across countries. 
	}\label{fig:ensemble}
\end{figure}
Finally, to illustrate the efficacy of GBQL even when true labels are observed with uncertainty, we create a toy dataset by randomly
pairing individuals within a country in the PHMRC data. To introduce label uncertainty into the analysis, for a pair of individuals, \(r1\)
and \(r2\), we let

\[
b_{r1i} = b_{r2i} = \frac{1}{2}(I(y_{r1i}=1) + I(y_{r2i}=1) ), 
\]

By using two  individuals each with a single (but possibly different) true label, we create two individuals each with uncertain observed labels in such a way that the total number of individuals with a given cause remains same in this new dataset as that in the actual PHMRC dataset. 
The data generation satisfies the assumption that \(p(y_{r}=i | \bb_{r}) = b_{ri}\). We
then used these beliefs instead of the true labels as input for our
method. Figure \ref{fig:phmrc_known_vs_unknown} compares the CCNAA for the individual methods across each value of $n$ for compositional predictions when
using the known labels versus representing uncertainty in the labels
through beliefs. The performance of GBQL is similar for both types of inputs. The CCNAA were slightly worse when labels are observed with uncertainty, as in Section \ref{simulations}. 

\begin{figure}[]
	\centering
	\includegraphics[width=1\linewidth]{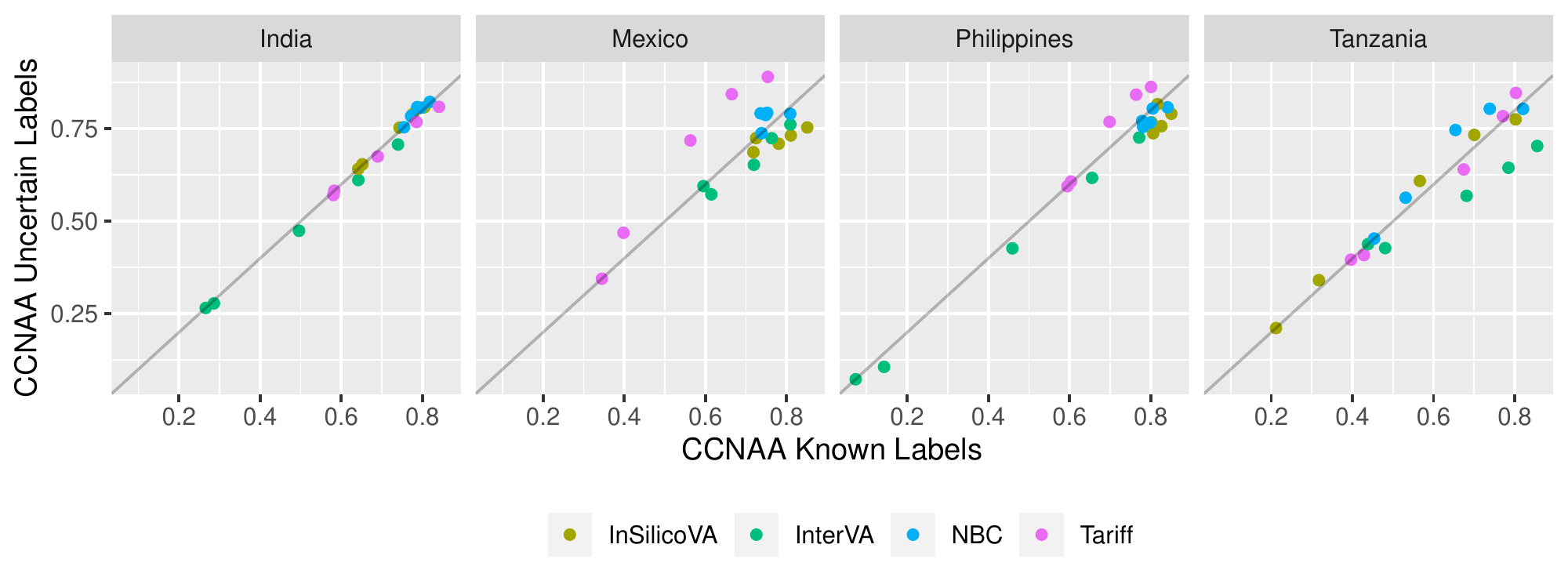}
	\caption{Comparison of CCNAA for when using known labels in PHMRC data versus labels with uncertainty in the synthetic data created from the PHMRC data. Each point represents a different value of $n$, with the black line representing the identity line.}
	\label{fig:phmrc_known_vs_unknown}
\end{figure}

\hypertarget{discussion}{%
	\section{Discussion}\label{discussion}}

Quantification is an important and challenging problem that has only recently gained the attention it deserves. There are important limitations of the commonly used methods; CC \citep{forman2005}, ACC \citep{forman2005}, PA \citep{bella2010quantification}, and APA \citep{bella2010quantification} as they do not account for dataset shift. In absence of local labeled data, GBQL with specific choices of priors yields model-based analogs for each of these methods and provides a probabilistic framework around these approaches to conduct inference beyond point-estimation. 
In presence of local test data GBQL leverages it and substantially improves quantification over these previous approaches. In such settings, GBQL extends BTL \citep{datta2018local} which does not allow uncertainty in either the predicted or the true labels. In summary, GBQL generalizes all these methods, allowing for unified treatment of both categorical and compositional classifier output, incorporation of training data (through priors) and labeled test data, and uncertain knowledge of labeled data classes. 

Appealing to the generalized Bayes framework, our Bayesian estimating equations and the KLD loss functions rely only on a simple first-moment assumption for compositional data that circumvents the need for full model specification even in a Bayesian setting.
The loss function approach easily extends to harmonize output from multiple classifiers, leading to a unified ensemble method which is a pragmatic solution guarding against inadvertent inclusion of a poorly performing classifier in the pool of algorithms. The Bayesian paradigm enables use of shrinkage priors to inform the estimation of $\mathbf{M}$ and $\mathbf{p}$ when limited labeled data from the test set is available. The GBQL 
Gibbs posterior can be approximated using our customized rounded and coarsened Gibbs sampler 
leading to fast posterior sampling compared to off-the-shelf samplers. 

There was no theory justifying the BTL method and more generally, to our knowledge, there is no theory for quantification learning methods under dataset shift. We offer a comprehensive theory for GBQL including posterior consistency, asymptotic normality, valid coverage of interval estimates, and finite-sample concentration rate. All results only assume a first moment assumption and are  robust to not having knowledge of full data distribution. 
Finally, extensive simulations and PHMRC data analysis show that the GBQL model is robust to model misspecification, and uncertainty in true labels, and significantly improves quantification in the presence of dataset shift. 

The estimating equation (\ref{eq:compreg}) is of independent importance beyond quantification learning. It offers a novel and direct method to generalized Bayes regression of compositional response $\ba_r$ on compositional covariate $\bb_r$, allowing $0$'s and $1$'s in both variables and without requiring full distributional specification (like the Dirichlet distribution) or data transformations (like the hard-to-interpret log-ratio transformations) customary in analysis of compositional data. We do not pursue this further here as it is beyond the scope of the paper but would like to point out that efficient posterior sampling algorithm for such a Bayesian composition-on-composition regression follows directly from a part of the coarsened sampler. Similarly, all the general theoretical results for quantification learning presented here, i.e., asymptotic consistency, normality, validity of coverage, and finite sample concentration rates can immediately be applied to this setting to ensure analogous guarantees for the Bayesian composition-on-composition regression.

A future direction is to extend the methodology for a continuous version of the problem, i.e., instead of predicting probabilities for $C$-categories, for each datapoint the classifiers now predict 
$a=(a^{(1)},\ldots,a^{(S)})$ -- a sample of $S$ predicted (real-valued) labels. A simple solution for the continuous case would be to discretize the domain into $C$ bins $B_1, \ldots, B_C$ and compute the empirical proportion of predicted labels in each bin, thereby transforming the sample value data to compositional data and subsequently using GBQL. A more general solution that prevents unnecessary discretization would be to write the continuous version of (\ref{eq:ql}) as 
$ f(a) = \sum_y f(a \given y) p(y)$ where $f(a)$ and $f(a \given y)$ respectively denotes the marginal and conditional densities. This leads to the moment equations $E(a^j) = \sum_y E(a^j \given y) p(y)$ for all $j \geq 1$ for which $E(a^j)$ exists. We can calculate $q_j=E(a^j)$ from the unlabeled set $\calU$ as $\frac 1N \sum_{r=1}^N \frac 1S \sum_{s=1}^S (a_r^{(s)})^j$. Similarly, we can calculate $M_{ij}=E(a^j \given y=i)$ as $\frac 1{|\{r \in \calL, y_r=i\}|} \sum_{r \in \calL, y_r=i} \frac 1S \sum_{s=1}^S (a_r^{(s)})^j$. Since the moment generating function uniquely defines the distribution, letting $\bq=(q_1,q_2,\ldots)$, $\bM=(M_{ij})$, we can solve for the quantity of interest $p(y)=\bp$ as the minimizer of some norm $\|\bq - \bM'\bp\|$ subject to $\bp$ lying on the simplex. We can also extend this to continuous true labels $y$ by using the equation $ f(a) = \int_y f(a \given y) f(y)$. Expressing $f(y)$ generally as a mixture density, $f(y) = \sum_{h=1}^H w_h f_h(y ; \mu_h)$ for some known densities $f_h$, we can solve for the unknown parameters $\mu_h$ and the unknown weights $w_h$ using the moment equation $q_k = \sum_h w_h \int_t  m_{tk} f_h(t; \mu_h) dt$.

\bibliography{bibliography}

\pagebreak

\section*{\begin{center}Supplementary Material for\\ ``Generalized Bayes Quantification Learning under dataset shift"\end{center}}

\renewcommand\theequation{S\arabic{equation}}
\renewcommand\thesection{S\arabic{section}}
\renewcommand\thetheorem{S\arabic{theorem}}
\renewcommand\thelemma{S\arabic{lemma}}
\renewcommand\thefigure{S\arabic{figure}}
\renewcommand\thetable{S\arabic{table}}
\setcounter{table}{0}
\setcounter{figure}{0}
\setcounter{equation}{0}
\setcounter{section}{0}
\setcounter{theorem}{0}
\setcounter{lemma}{0}

\section{Gibbs sampler for ensemble quantification}
\begin{align*}
z_{rt}^{(k)} | \cdot &\sim \begin{cases} Mult\left(1, \frac{1}{\sum_{i}M_{ij}^{(k)}p_i}(M_{1j}^{(k)}p_{1},\ldots, M_{Cj}^{(k)}p_{C})\right),\ r \in \mathcal{U},\ d_{rt} = j \\
Mult\left(1, \frac{1}{\sum_{i}M_{ij}^{(k)}b_{ri}}(M_{1j^{(k)}}b_{r1},\ldots, M_{Cj}^{(k)}b_{rC})\right),\ r \in \mathcal{L},\ d_{rt} = j
\end{cases}\\
M_{i}^{(k)} | \cdot &\sim Dir\left( \tilde{\text{V}}_{i1}^{(k)}, \ldots, \tilde{\text{V}}_{iJ}^{(k)}\right),\
\tilde{\text{V}}_{ij}^{(k)} = \text{V} + \frac{1}{T}\left(\sum_{r \in \mathcal{U}, \mathcal{L}}\sum_{t=1}^{T}(I(d_{rt}^{(k)}=j)I(z_{rt}^{(k)}=i)\right)\\
p | \cdot &\sim Dir\left(\tilde{\text{v}}_{1},\cdots \tilde{\text{v}}_{C}\right),\ \tilde{\text{v}}_{i} = v_{i} + \frac{1}{T} \cdot \left(\sum_{k=1}^{K}\sum_{r \in \mathcal{U}}\sum_{t=1}^{T}I(z_{rt}^{(k)}=i)\right) \\
\end{align*}

\section{Additional theoretical results}\label{sec:addth}
\subsection{Technical statement of Theorem \ref{th:cp}}
\begin{theorem*} Let $\widehat \bV_{A,\calU}$ denote the sample covariance of $\ba_r$'s, $r \in \calU$ and define estimates $\widehat{\bD_r}$,   $\widehat{\bU^{0'}_{\bM}}$ and $\widehat{\bU^{0'}_{\bp}}$ respectively of $\bD_r$, 
	$\bU^{0'}_{\bM}$ and $\bU^{0'}_{\bp}$ by plugging in the Gibbs posterior mean $\widehat\bM, \widehat\bp$ in place of $\bM^0,\bp^0$ in (\ref{eq:gr}), (\ref{eq:um}) and (\ref{eq:up}). 
	Let $\widehat g_r=(\bb_r \otimes \widehat\bD_r)\ba_r$, and $\widehat{V_{\hat g,\calL}}$ be the sample covariance matrix of  $\widehat g_r$ for $r \in \calL$. Let 
	\begin{equation}\label{eq:omegaj}
	\widehat\bOmega = \left(
	\begin{array}{c}
	\widehat{\bU^{0'}_{\bM}}  \\
	\widehat{\bU^{0'}_{\bp}}
	\end{array}\right) \widehat{\bV_{A,\calU}} (\widehat{\bU^{0}_{\bM}}, \widehat{\bU^{0}_{\bp}}) + \xi \left(\begin{array}{cc}
	\widehat{\bV_{\hat g,\calL}}  & \bO \\
	\bO   & \bO
	\end{array}
	\right) \mbox{ and } \widehat \bJ = \nabla^2_{\widehat\btheta} f_N.
	\end{equation}
	For any differentiable function $g(\btheta)$ and $0 < s <1$, define $C_{g,s}=z_{1-s/2}\frac{\sqrt{\nabla_{\hat\btheta} g'\widehat{\bJ}^{-1}\widehat\bOmega\widehat{\bJ}^{-1}\nabla_{\hat\btheta} g}}{\sqrt{N}}$. Then 
	$$P(g(\hat\btheta) - C_{g,s} < g(\btheta^0) < g(\hat\btheta) + C_{g,s} ) \to 1 - s.$$
\end{theorem*}

\subsection{Technical statement of Corollary \ref{cor:ens}}
\begin{corollary*} Let $K$ predictions are available for each instance from $K$ classifiers, and Assumptions 1 and 2 are satisfied for each classifier. With   $\btheta=(\tM^{(1)},\ldots,\tM^{(K)}, \tp)$ the Gibbs posterior $\nu$ for ensemble GBQL, given by 
	\[
	\nu_N(\btheta) \propto \exp\left( - \alpha N f_N(\btheta) \right) \Pi(\btheta) \mbox{ where } f_N(\btheta)=\sum_{k=1}^{K}\left[\sum_{r \in \mathcal{U}}D_{KL}(\mathbf{a}_{r}^{k} || {\mathbf{M}^{(k)}}^{'}\mathbf{p}) + \sum_{r \in \mathcal{L}}D_{KL}(\mathbf{a}_{r}^{k} || {\mathbf{M}^{(k)}}^{'}\bb_{r})\right].
	\]
	satisfies the following properties:
	\begin{enumerate}[(i)]
		\item 	(Posterior consistency.) Let $B_\eps(\btheta^0)$ be an $\ell_1$ ball of radius $\eps$ around $\btheta^0$,  and $\Pi(\btheta)$ be any prior which gives positive support to $B_\eps(\btheta^0)$ for any $\eps >0$.  as $N,n \rightarrow \infty$ and $n/N$ to some limit $\xi$, 
		for any $\eps > 0$, $P_{\nu_N}(B_\eps(\btheta^0)) \rightarrow 1$.
		\item (Asymptotic normality.) Let $\bA_r=(\ba^{1'}_r,\ldots,\ba_r^{K'})'$ denote the vector stacking up the $K$ prediction vectors for a case, $\bV_{A,\calU}=$Cov$_{r \in \calU}(\bA_r)$, $\bq^k_r = \bM^{(k0)'}\bb_r$, $\bD_r=$ a block diagonal matrix with blocks $\bD_{r,1},\ldots,\bD_{r,K}$ where $\bD_{r,k}=\left[diag(1/\bq^k_{r,1:C-1});-1/\bq^k_{r,C}\bone_{C-1 \times 1}\right]$, $\bg_r =$ the $KC(C-1) \times 1$ vector $\left(\bb_r \otimes \bD_r \right)\bA_r$ and $\bV_{g,\calL} = Cov(\bg_{r})_{r\in \calL}$. 
		
		Let $\bJ_k$ denote the matrix of Theorem \ref{th:norm} for the $k^{th}$ classifier and $\bJ_{ens}=\sum_{k=1}^K \bJ_k$. Then there exists a matrix $\bU(\btheta^0)$ (specified in the proof) such that with $\bOmega_{ens}$ defined as  
		\begin{equation*}
		\bOmega_{ens} = \bU(\btheta^0)' \bV_{A,\calU}\bU(\btheta^0) + \xi \left(\begin{array}{cc}
		\bV_{g,\calL}  & \bO \\
		\bO   & \bO
		\end{array}
		\right).
		\end{equation*}
		the mean $\hat\btheta$ of the Gibbs posterior distribution $\nu_N$  is asymptotically normal i.e., 
		$\sqrt N \bOmega_{ens}^{-1/2}\bJ_{ens}(\hat \btheta - \btheta^0) \to_d N(0,\bI)$. 
		\item (Asymptotic coverage.) Let $\widehat \bJ_k$ denote the estimate corresponding to $\widehat \bJ$ in (\ref{eq:omegaj}) for the $k^{th}$ classifier and let  $\widehat\bJ_{ens}=\sum_{k=1}^K \widehat\bJ_k$. Let $\widehat \bg_r$ and denote values of $\bg_r$ by plugging in $\widehat \btheta$ for $\btheta^0$. Let $\widehat \bV_{\hat g,\calL}$ and $\widehat \bV_{A,\calU}$ respectively denote the sample variances of $\widehat \bg_r, r\in \calL$ and $\bA_r, r \in \calU$, and $\widehat \bOmega_{ens}$ denote the estimate of $\bOmega_{ens}$ by plugging in $\bU(\widehat\btheta)$, $\widehat \bV_{\hat g,\calL}$ and $\widehat \bV_{A,\calU}$. For any differentiable function $g(\btheta)$ and $0 < s <1$, define $C_{g,s}=z_{1-s/2}\frac{\sqrt{\nabla_{\hat\btheta} g'\widehat{\bJ}^{-1}\widehat\bOmega_{ens}\widehat{\bJ}_{ens}^{-1}\nabla_{\hat\btheta} g}}{\sqrt{N}}$. Then 
		$$P(g(\hat\btheta) - C_{g,s} < g(\btheta^0) < g(\hat\btheta) + C_{g,s} ) \to 1 - s.$$
		\item (Finite-sample concentration rate.) Let $\eps:=\eps_N > 0$ be such that $B_\eps(\btheta^0)$ -- an $\ell_1$ ball of radius $\eps$ around $\btheta^0$ lies in the interior of $\Theta$, and $\Pi_N(\btheta)$ be any (possibly $N$-dependent) prior that gives positive mass of atleast $\exp(-NR\eps)$ to $B_\eps(\btheta^0)$ for some universal constant $R$. Then, under Assumptions 1-2, 
		we have for any $\eps > 0$, $\alpha \in (0,1)$, $D > 1$, and $t > 0$, 
		$P_{\nu_N}\left(D_{N,\alpha}(\btheta,\btheta^0) \geqslant 
		(D+3t)NKR\eps\right) \leq \exp(-tNKR\eps) \mbox{
			with } P_N^0$-probability atleast $1 - 
		\frac {1+K^{-1}}{NR\eps\min\{{(D-1+t)^2,t}\}}$.
	\end{enumerate}
\end{corollary*}

\pagebreak
\section{Proofs}\label{sec:proofs}
\textbf{Proofs}

For clarity, we prove the results assuming that $n = \xi N$ for $0 \leq \xi \leq 1$, as the proof is similar when the relationship only holds in the limiting sense. Recall the definitions:
\begin{equation}
\begin{array}{rl}
\ellu(\btheta) &= -\frac{1}{N}\sum_{r=1}^{N}\sum_{j=1}^{C}a_{rj}\log\left(\sum_{i=1}^{C}\frac{M_{ij}p_{i}}{a_{rj}}\right)\\
\ellN(\tM) := \elll(\tM) &= -\frac{1}{N}\sum_{r=1}^{\xi N}\sum_{j=1}^{C}a_{rj}\log\left(\sum_{i=1}^{C}\frac{M_{ij}b_{ri}}{a_{rj}}\right),\\
f_N(\btheta) &= \ellN(\tM) + \ellu(\btheta),\\
\ell_\calL(\tM) &= \xi E_{\mathcal{L}}\left[D_{KL}(\mathbf{a} || \mathbf{M}^{'}\mathbf{b})\right]\\
\ell_\calU(\btheta) &= E_{\calU} \left[D_{KL}(\mathbf{a}|| \mathbf{M}^{'}\mathbf{p})\right]\\
f(\btheta) &= \ell_\calL(\tM) + \ell_\calU(\btheta),\\
\end{array}
\end{equation}
We now prove a series of Lemmas about these loss functions. 

\begin{lemma}\label{lem:fl}The following holds for $\ell_\calL$. \begin{enumerate}[(i)]
		\item $\nabla \ell_{\mathcal{L}}(\tM^{0})= 0$, 
		\item $\nabla^2 \ell_{\mathcal{L}}(\tM)$ is positive definite for all $\tM$ under Assumption 1. 
		\item $\nabla^3 \ell_{\mathcal{L}}(\tM)= 0$ is continuous in a neighborhood of $\tM^0$. 
	\end{enumerate}
\end{lemma}

\begin{proof} To show condition (i) holds, we see that
	\begin{equation*}
	\frac{\partial \ell_{\mathcal{L}}(\tM)}{\partial M_{ij}} = -\xi \frac{\partial}{\partial M_{ij}} E_{r \in \mathcal{L}}\left[\sum_{j}a_{rj}\log\left(\sum_{i}M_{ij}b_{ri} \right) \right]
	\end{equation*}
	
	To switch the order of differentiation and expectation, we will use the dominated convergence theorem and show that in a neighborhood of $\tM^{0}$
	\begin{equation*}
	\left|\frac{\partial}{\partial M_{ij}}a_{rj}\log\left(\sum_{i}M_{ij}b_{ri} \right)\right|
	\end{equation*}
	is bounded by some integrable random variable $X$.
	We first note that since $a_{rj} \leq 1$ and $b_{ri}  \leq 1$, we have
	\begin{equation*}
	\left|\frac{\partial}{\partial M_{ij}}a_{rj}\log\left(\sum_{i}M_{ij}b_{ri} \right)\right| = b_{ri} \left|\frac{a_{rj}}{\sum_{i=1}^{C}M_{ij}b_{ri}} - \frac{a_{rC}}{\sum_{i=1}^{C}M_{iC}b_{ri}} \right| \leq \frac{1}{\sum_{i=1}^{C}M_{ij}b_{ri}} + \frac{1}{\sum_{i=1}^{C}M_{iC}b_{ri}} 
	\end{equation*}
	
	and because $\tM^{0}$ is an interior point, we can choose a small enough neighborhood $N_{\epsilon}(\tM^{0})$ such that  $\forall$ $\tM \in N_{\epsilon}(\tM^{0})$, $\min\limits_{ij} M_{ij} > K$, where $K$ is a constant that depends on $\epsilon$ and $\tM^{0}$. Thus we have
	\begin{equation*}
	\sum_{i=1}^{C}M_{ij}b_{ri} > K \sum_{i=1}^{C}b_{ri} = K \mbox{ implying that } \frac{1}{\sum_{i=1}^{C}M_{ij}b_{ri} } < \frac 1K\, \forall i,j,r.
	\end{equation*}
	
	Hence the dominated convergence theorem applies. We now have
	\begin{align*}
	\frac{\partial \ell_{\mathcal{L}}(\tM)}{\partial M_{ij}}\bigg\rvert_{\bM^{0}} &= -\xi E_{r \in \mathcal{L}}\left[\frac{a_{rj}b_{ri}}{\sum_{i=1}^{C}M^0_{ij}b_{ri}} - \frac{a_{rC}b_{ri}}{\sum_{i=1}^{C}M^0_{iC}b_{ri}} \right]\\
	&= -\xi E_{\mathbf{b}_{r}, r \in \mathcal{L}, \bM^{0}}\left[E_{\mathbf{a}_{r} | \mathbf{b}_{r}, r \in \mathcal{L}, M^{0}}\left[\frac{a_{rj}b_{ri}}{\sum_{i=1}^{C}M^0_{ij}b_{ri}} - \frac{a_{rC}b_{ri}}{\sum_{i=1}^{C}M^0_{iC}b_{ri}} \right] \right]\\
	&= -\xi E_{\mathbf{b}_{r}, r \in \mathcal{L}, M^{0}}b_{ri} \left[\frac{\sum_{i=1}^{C}M_{ij}^{0}b_{ri}}{\sum_{i=1}^{C}M_{ij}^{0}b_{ri}}-\frac{\sum_{i=1}^{C}M_{iC}^{0}b_{ri}}{\sum_{i=1}^{C}M_{iC}^{0}b_{ri}} \right]\\
	&= 0. \numberthis
	\end{align*}
	
	Moving to part (ii), we first prove that under Assumption 1, $E_{r \in \mathcal{L}}\left[\mathbf{b}_{r}\mathbf{b}_{r}^{'}\right] \succ 0$. 
	
	Let $\bx (\neq \bzero) \in \mathbb R^C$ and $i = \arg\min_j |x_j|$. Without loss of generality, let $x_i > 0$ (otherwise we can work with $-\bx$). Choose $N_i=\{ \bb \in \widetilde S^C \given b_i \geq 1 - \eps\}$. For small enough $\eps=\eps(\bx)$ we have for any $\bb \in N_i$, we have $\bx'\bb > x_i (1-\eps) - \eps \|\bx\|_1 > \delta$ for some $\delta > 0$.  
	\begin{align*}
	\bx'E_{r \in \mathcal{L}}\left[\mathbf{b}_{r}\mathbf{b}_{r}^{'}\right]\bx = \int (\bx'\bb)^2 dF_{b,\calL} \geq \int_{N_i} (\bx'\bb)^2 dF_{b,\calL} \geq \delta^2 F_{b,\calL}(N_i) > 0. 
	\end{align*}
	The last inequality follows from Assumption 1 which guarantees positive mass around all such $N_i$. This proves $E_{r \in \mathcal{L}}\left[\mathbf{b}_{r}\mathbf{b}_{r}^{'}\right] \succ 0$. 
	
	Now we once again use the same reasoning as part (i) to switch the orders of expectation and differentiation in a neighborhood of $\bM^{0}$. We have
	\begin{equation}\label{eq:hessl}
	\frac{\partial^{2} \ell_{\mathcal{L}}(\tM)}{\partial M_{ij} \partial M_{i^{'}j^{'}}}\bigg\rvert_{\bM} = \xi E_{r \in \mathcal{L}} \left[  b_{ri}b_{ri^{'}}\left( \frac{I(j=j^{'})a_{rj}}{(\sum_{i=1}^{C}M_{ij}b_{ri})^2} +  \frac{a_{rC}}{(\sum_{i=1}^{C}M_{iC}b_{ri})^2}\right)\right]
	\end{equation}
	Let $\mathbf{H}(\tM)$ denote the Hessian at $\tM$. 
	From (\ref{eq:hessl}), $\bH(\tM)$ has $C \times C$ blocks of the form $\bH_j(\tM) + \bH_C(\tM)$ where for $j=1,\ldots, C$, where 
	\begin{align*}
	\bH_j(\tM) = E_{r \in \mathcal{L}} \left(  \left[  \bb_{r}\bb_{r^{'}}\left( \frac{a_{rj}}{(\sum_{i=1}^{C}M_{ij}b_{ri})^2} \right)\right] \right) \geq E_{r \in \calL} \left(a_{rj}\bb_r\bb_r' \right).
	\end{align*}
	As we have proved $E_{r \in \calL} \bb_r\bb_r' \succ 0$, we have 
	\begin{equation}\label{eq:poscov}
	E_{r \in \calL} \left(a_{rj}\bb_r\bb_r'\right) = E_{r \in \calL} \left[(\sum_{i=1}^{C}M^0_{ij}b_{ri})\bb_r\bb_r'\right] \geq (\min_{ij} \bM^0 )  E_{r \in \calL} \bb_r\bb_r' \succ 0.
	\end{equation}
	Hence, $\bH_j(\tM) \succ 0$ for all $j$. $\bH_C(\tM) \succ 0$ implies $\bH(\tM)$ dominates the block-diagonal matrix with blocks $\bH_j(\tM)$. As $\bH_j(\tM) \succ 0$ for all $j$, this block diagonal matrix is positive definite and so is $\bH(\tM)$.
	
	For part (iii), once again application of DCT, as reasoned above, yields
	$$
	\frac{\partial^{3} \ell_{\mathcal{L}}}{\partial M_{ij} \partial M_{i^{'}j^{'}}\partial M_{i^{''}j^{''}}} = -2   E_{r \in \calL} \left(I(j=j^{'}=j^{''})\frac{a_{rj}b_{ri}b_{ri^{'}}b_{ri^{''}}}{(\sum_{i=1}^{C}M_{ij}b_{ri})^{3}} -  \frac{a_{rC}b_{ri}b_{ri^{'}}b_{ri^{''}}}{(\sum_{i=1}^{C}M_{iC}b_{ri})^{3}}\right)$$
	In a neighborhood around $\tM^0$, as $M_{ij}$ is bounded below by $K$ and $a_{rj}$ and $b_{ri}$'s are bounded by 1, we have this to be absolutely bounded by $4/K^3$.   
\end{proof}

\begin{lemma}\label{lem:fnl} The following holds for $\ellN$.
	\begin{enumerate}[(i)]
		\item $\lim_{N \rightarrow \infty}\ellN(\tM) =  \ell_{\mathcal{L}}(\tM)$,
		\item 
		under Assumption 1, 
		$\ellN(\tM)$ is strictly convex for all $\tM$ for large enough $N$,
		\item $\nabla^3 \ellN$ exists, is a continuous and uniformly bounded function in a neighborhood of $\tM^0$,
		\item (Identifiability of $\tM$)\label{lem:idM} If Assumption 1 holds, then 
		$\liminf_n \inf_{\tM \notin B_\eps(\tM^0)} \ellN/n > E_\calL(D_{KL}(\ba || {\bM^0}'\bb))$.
	\end{enumerate}
\end{lemma}

\begin{proof}Part (i) is proved as
	\begin{align*}
	\ellN(\tM) &= -\frac{1}{N}\sum_{r=1}^{\xi N}\sum_{j=1}^{C}a_{rj}\log\left(\sum_{i=1}^{C}\frac{M_{ij}b_{ri}}{a_{rj}}\right)\\
	&= -\frac{N\xi}{N}\frac{1}{N\xi}\sum_{r=1}^{\xi N}\sum_{j=1}^{C}a_{rj}\log\left(\sum_{i=1}^{C}\frac{M_{ij}b_{ri}}{a_{rj}}\right)\\
	&\xrightarrow{N \rightarrow \infty} \xi E_{r \in \mathcal{L}}\left[-\sum_{j=1}^{C}a_{rj}\log\left(\sum_{i=1}^{C}\frac{M_{ij}b_{ri}}{a_{rj}}\right)\right]\\
	&= \xi E_{\mathcal{L}}\left[D_{KL}(\mathbf{a} || \mathbf{M}^{'}\mathbf{b})\right]\\
	&= \ell_{\mathcal{L}}(\tM) 
	\end{align*}
	
	We prove part (ii) by showing that the Hessian $\bH_N(\tM)$ of $\ellN(\tM) \succ 0$. We have
	\begin{equation}\label{eq:gradfln}
	\begin{aligned}
	\frac{\partial \ellN}{\partial M_{ij}} &= - \frac{1}{N}\sum_{r=1}^{\xi N} \left[\frac{a_{rj} b_{ri}}{\sum_{i=1}^{C}M_{ij}b_{ri}} - \frac{a_{rC} b_{ri}}{\sum_{i=1}^{C}M_{iC}b_{ri}} \right]\\
	\frac{\partial^{2} \ellN}{\partial M_{ij} \partial M_{i^{'}j^{'}}} &= I(j=j^{'}) \frac{1}{N}\sum_{r=1}^{\xi N} \frac{a_{rj}b_{ri}b_{ri^{'}}}{(\sum_{i=1}^{C}M_{ij}b_{ri})^{2}} + \frac{1}{N}\sum_{r=1}^{\xi N} \frac{a_{rC}b_{ri}b_{ri^{'}}}{(\sum_{i=1}^{C}M_{iC}b_{ri})^{2}}\\
	\frac{\partial^{3} \ellN}{\partial M_{ij} \partial M_{i^{'}j^{'}}\partial M_{i^{''}j^{''}}} &= -I(j=j^{'}=j^{''}) \frac{2}{N}\sum_{r=1}^{\xi N} \frac{a_{rj}b_{ri}b_{ri^{'}}b_{ri^{''}}}{(\sum_{i=1}^{C}M_{ij}b_{ri})^{3}} + \frac{2}{N}\sum_{r=1}^{\xi N} \frac{a_{rC}b_{ri}b_{ri^{'}}b_{ri^{'i}}}{(\sum_{i=1}^{C}M_{iC}b_{ri})^{3}}
	\end{aligned}
	\end{equation}
	
	Hence $\bH_N(\tM)$ has $C \times C$ blocks of the form $\bH_{N,j}(\tM)+\bH_{N,C}(\tM)$ for $j=1,\ldots,C-1$ where 
	\begin{align*}
	\bH_{N,j}(\tM) = \frac{1}{N}\sum_{r=1}^{\xi N} \frac{a_{rj}\bb_{r}\bb_{r^{'}}}{(\sum_{i=1}^{C}M_{ij}b_{ri})^{2}} \geq \frac{1}{N}\sum_{r=1}^{\xi N} a_{rj}\bb_{r}\bb_{r^{'}}. 
	\end{align*}
	We have shown in (\ref{eq:poscov}) that $E_{r \in \calL}\; a_{rj}\bb_{r}\bb_{r^{'}} \succ 0$. Hence, there exists a set $\calA$ with $P(\calA) = 1$ such that $\frac{1}{N}\sum_{r=1}^{\xi N} a_{rj}\bb_{r}\bb_{r^{'}} \succ 0$ on $\calA$ for large enough $N$. On $\calA$, we thus have $\bH_{N,j}(\tM) \succ 0$ for $j=1,\ldots,C$ and all $\tM$. This in turn implies $\bH_N(\tM) = \mbox{bdiag}(\bH_{N,j}(\tM)) + \left((\bone\bone') \otimes \bH_{N,C}(\tM)\right) \succ 0$ on $\calA$. So $\ellN(\tM)$ is strictly convex.

	Finally, for part (iii), we note from (\ref{eq:gradfln}) that $\nabla^3 \ellN$ is continuous in a neighborhood of $\tM^0$. We also note that the denominators involve terms of the form are $(\sum_{i=1}^{C}M_{iC}b_{ri})^3$. As argued in the proof of Lemma \ref{lem:fl}, these are uniformly (free of $\tM$ and $b_r$) bounded away from zero in a neighborhood around $\tM^0$. Hence, $\nabla^3 \ellN$ is an uniformly bounded function in a neighborhood around $\tM^0$. 
	
	For part (iv), Lemma \ref{lem:fl} and Lemma \ref{lem:fnl} (parts  (i) and (ii)) are sufficient to establish this using Theorem 2.3 (Condition 3 $\implies$ Condition 1) of \cite{miller2019asymptotic}.
\end{proof}



\begin{lemma}\label{lem:fu} The following holds for $\ell_\calU$. 
	\begin{enumerate}[(i)]
		\item $\ell_\calU(\btheta) \geq \ell_\calU(\btheta^0)$ for all $\btheta \in \Theta$.
		\item $\ell_\calU(\tM,\tp)$ is twice continuously differentiable in a neighborhood around $\theta^0$. 
	\end{enumerate}
\end{lemma}

\begin{proof} For part (i), we have $\ell_\calU(\btheta)=\ell_\calU(\bM,\bp)=\ell_\calU(\bI,\bM'\bp) \geq \inf_\bq \ell_\calU(\bI,\bq)$. Now, $\ell_\calU(\bI,\bq) = E_{r \in \calU}\; d_{KL}(\ba_r \| \bq)$ is minimized as $\bq = E_{r \in \calL}(\ba_r) = \bM^{0'}\bp^0$. Hence, $\ell_\calU(\btheta) \geq \ell_\calU(\bI,\bM^{0'}\bp^0) = \ell_\calU(\bM^0,\bp^0) = \ell_\calU(\btheta^0)$. 
	
	For part (ii), similar to Lemma \ref{lem:fl} part (iii) we can differentiate within the expectation signs, and the resulting third derivative will have numerators as functions of $a_{rj}$'s which are bounded by 1, and denominators as functions of $\sum_i M_{ij}p_i$ which is bounded from below by $K$ in a neighborhood around $\btheta^0$. Hence, $\nabla^3 \ell_\calU$ will be bounded and continuous in the neighborhood. 
\end{proof}

\begin{lemma}\label{lem:fnu} The following holds for $\ellu$:
	\begin{enumerate}[(i)]
		\item $\lim_N \ellu(\btheta)/N = \ell_\calU(\btheta) = E_\calU(D_{KL} (\ba || {\bM}'\bp))$.
		\item $\nabla^3 \ellu(\btheta)/N$ exists and is uniformly bounded in a neighborhood around $\btheta^0$. 
		\item (Weak identifiability of $\tM,\tp$). $\lim_N \inf_{\btheta \in \Theta} \ellu(\btheta)/N$ exists and equals  $\ell_\calU(\btheta^0) = E_\calU(D_{KL} (\ba || {\bM^0}'\bp^0))$.
	\end{enumerate}
\end{lemma}

\begin{proof}
	Part (i) follows similar to part (i) of Lemma \ref{lem:fnl}. 
	Proof of part (ii) is also similar to proof of Lemma \ref{lem:fnl} part (iii) as the denominators of the third derivative will involve $(\sum_i M_{ij}p_i)^3$ which is bounded away from zero in a neighborhood of $\bM^0$. 
	
	For part (iii), we first note that $
	\inf_{\tM, \tp}\ellu(\tM, \tp) = \inf_{\tM, \tp}\ellu(\widetilde \bI, \widetilde{\bM'\bp}) =  \inf_{\tilde{\bq}}\tilde{\ell}_{\mathcal{U}, N}(\tilde{\bq})
	$
	where 
	\begin{align*}
	\tilde{\ell}_{\mathcal{U}, N}(\tilde{\bq}) =  \ellu(\tilde{\mathbf{I}}, \tilde{\bq})
	= -\frac{1}{N}\sum_{r=1}^{N}\sum_{j=1}^{C}a_{rj}\frac{\log(q_{j})}{a_{rj}}. \numberthis
	\end{align*}
	
	Clearly, $\tilde{\ell}_{\mathcal{U}, N}(\tilde{\bq})$ is minimized at $\hat{\tilde{\bq}}$ where $\hat{\tilde{q}}_{j} = \frac{1}{N}\sum_{r=1}^{N}a_{rj}$. Then
	\begin{align*}
	\inf_{\btheta \in \Theta}\ellu(\tM, \tp) &= \tilde{\ell}_{\mathcal{U}, N}(\hat{\tilde{\bq}})\\
	&= -\sum_{j=1}^{C}\log(\hat{\tilde{q}}_{j})\hat{\tilde{q}}_{j} + \frac{1}{N}\sum_{r=1}^{N}\sum_{j=1}^{C}a_{rj}\log(a_{rj})\\
	&\xrightarrow{N \rightarrow \infty}-\sum_{j=1}^{C}\left[\log(\sum_{i=1}^{C}M_{ij}^{0}p_{i}^{0})\sum_{i=1}^{C}M_{ij}^{0}p_{i}^{0} - E_{r \in \mathcal{U}}\left[a_{rj}\log(a_{rj}) \right] \right]\\
	&= -\sum_{j=1}^{C} E_{r \in \mathcal{U}}\left[a_{rj}\log\left(\frac{\sum_{i=1}^{C}M_{ij}^{0}p_{i}^{0}}{a_{rj}} \right) \right]\\
	&= E_{\calU} D_{KL}(\mathbf{a} || \bM^{0'}\bp^0).
	\end{align*}
\end{proof}

\noindent We now return to the loss function for the full data, $f_{N}$, to show 
\begin{lemma}\label{lem:f} The following holds for $f_N$ and $f$.
	\begin{enumerate}[(i)]
		\item $\lim_N  f_N(\tM, \tp) = f(\tM, \tp)$.
		\item There exists 
		$\bOmega:=\bOmega(\btheta^0) \succ 0$ and $ \bOmega^{-1/2}\nabla f_N(\btheta^0)/\sqrt N \to_d N(0,\bI)$ 
		\item $f_N$ and $f$ are twice continuously differentiable and the third derivative of $f_N$ is uniformly bounded around any small neigbhborhood of $\btheta^0$. 
		\item Under Assumptions 1 and 2, $f(\btheta) > f(\btheta^0)$ for all $\btheta \neq \btheta^0$. If $\bJ(\btheta)=\nabla^2 f(\btheta)$ then $\bJ := \bJ(\btheta^0) \succ 0$. 
	\end{enumerate}
\end{lemma}

\begin{proof} Part (i) is proved directly from parts (i) of Lemmas \ref{lem:fnl} and \ref{lem:fnu}. 
	
	For part (ii), for any vector $\bx$, denoting $1/\bx =(1/x_1,\ldots,1/x_C)'$, we have 
	\begin{equation}\label{eq:gr}
	\begin{aligned}
	\frac 1 {\sqrt {N}} \frac{\partial \ellN}{\partial M_{ij}}\Big|_{\bM^0} =& \frac {\sqrt \xi}{\sqrt{\xi N}} \sum_{r=1}^{\xi N} b_{ri} \left(\frac {a_{rj}}{\sum_i M^0_{ij}b_{ri}} - \frac {a_{rC}}{\sum_i M^0_{iC}b_{ri}}  \right)\\
	\implies \frac 1 {\sqrt {N}} \frac{\partial \ellN}{\partial \tM}\Big|_{\bM^0} =& \frac {\sqrt \xi}{\sqrt {\xi N}} \sum_{r=1}^{\xi N} \left(\bb_r \otimes \bD_r \right)\ba_r \mbox{ where } \bD_r=\left[diag(1/\bq_{r,1:C-1});-1/\bq_{r,C}\bone_{C-1 \times 1}\right] \mbox{ with } \bq_r = \bM^{0'}\bb_r\\
	& \to_d N(0,\xi \bV_{g,\calL}) \mbox{ where } \bV_{g,\calL} = Cov(\bg_{r})_{r\in \calL} \mbox{ with } \bg_r = (\bb_r \otimes \bD_r)\ba_r.
	\end{aligned}
	\end{equation}
	
	Letting $\bq=\bM'\bp$ and $\bV_{A,\calU}=Cov(\ba_{r})_{r \in \calU}$, we have
	\begin{align*}
	\frac 1 {\sqrt {N}} \frac{\partial \ellu}{\partial M_{ij}}\Big|_{\btheta^0} =& \frac 1{\sqrt N} \sum_{r=1}^N p^0_i \left(\frac {a_{rj}}{\sum_i M^0_{ij}p^0_i} - \frac {a_{rC}}{\sum_i M^0_{iC}p^0_i}  \right)\\
	=& \frac 1{\sqrt N} \sum_{r=1}^N p^0_i \left(\frac {a_{rj}}{ q^0_{j}} - \frac {a_{rC}}{ q^0_{C}}  \right) \\
	=& \frac 1{\sqrt N} \sum_{r=1}^N \bu_{M_{ij}}^{0'}\ba_r \mbox{ where } \bu_{M_{ij}}^{0}=p^0_i * (0,\ldots,0,1/q_j^0,0,\ldots,0,-1/q_C^0)'\\
	& \to_d N(0,\bu_{M_{ij}}^{0'} \bV_{A,\calU} \bu^0_{M_{ij}}) 
	\end{align*}
	
	Let $\otimes$ denote the Kronecker product. We have, $\frac 1 {\sqrt {N}} \nabla_{\bM^0} \ell_\calU \to_d N(0,\bU^{0'}_{\bM} \bV_{A,\calU} \bU^0_{\bM})$ where 
	\begin{equation}\label{eq:um}
	\begin{aligned}
	\bU^{0'}_{\bM}&=(\bu^0_{M_{11}},\ldots,\bu^0_{M_{1,C-1}},\bu^0_{M_{21}},\ldots,\bu^0_{M_{C,C-1}})'\\
	&= \bp^0 \otimes \bD \mbox{ where } \bD_{C-1 \times C}=\left[diag(1/\bq^0_{1:C-1})_{C-1 \times C-1}; - 1/q^0_C \bone_{C-1 \times 1} \right]. 
	\end{aligned}   
	\end{equation}
	Similarly, we have 
	\begin{align*}
	\frac 1 {\sqrt {N}}  \frac{\partial \ellu}{\partial p_{i}}\Big|_{\btheta^0} =& \frac 1{\sqrt N} \sum_{r=1}^N \sum_{j=1}^C a_{rj} \left(\frac {M^0_{ij}-M^0_{Cj}}{\sum_i M^0_{ij}p^0_i}\right) \\
	=& \frac 1{\sqrt N} \sum_{r=1}^N \bu^{0'}_{p_i}\ba_r \mbox{ where } \bu^{0}_{p_i}=\left(\frac {M^0_{i1}-M^0_{C1}}{ q^0_1}, \frac {M^0_{i2}-M^0_{C2}}{ q^0_2}, \ldots, \frac {M^0_{i,C}-M^0_{C,C}}{ q^0_C}\right)'\\
	& \to_d N(0,\bu_{p_{i}}^{0'} \bV_{A,\calU} \bu^0_{p_{i}})
	\end{align*}
	Consequently, letting $\oslash$ denote the elementwise division of matrices, we have $\frac 1 {\sqrt {N}} \nabla_{\bp^0} \ell_\calU \to_d N(0,\bU^{0'}_{\bp} \bV_{A,\calU} \bU^0_{\bp})$ where 
	\begin{equation}\label{eq:up}
	\begin{aligned}
	\bU^{0'}_{\bp}&=(\bu^0_{p_{1}},\ldots,\bu^0_{p_{C-1}})'\\
	&= (\bM^0_{1:C-1,1:C} - \bone_{C-1 \times 1} \otimes \bM^{0'}_{C*}) \oslash (\bone_{C-1 \times 1} \otimes \bq^{0'})
	\end{aligned}
	\end{equation}
	
	Combining, all this we we have $f_N/\sqrt N \to_d N(0,\Omega(\btheta^0))$ where  
	\begin{equation}\label{eq:omega}
	\bOmega(\btheta^0) = \left(
	\begin{array}{c}
	\bU^{0'}_{\bM}  \\
	\bU^{0'}_{\bp}
	\end{array}\right) \bV_{A,\calU} (\bU^{0}_{\bM}, \bU^{0}_{\bp}) + \xi \left(\begin{array}{cc}
	\bV_{g,\calL}  & \bO \\
	\bO   & \bO
	\end{array}
	\right).
	\end{equation}
	
	Part (iii) is immediately proved due to Lemmas \ref{lem:fl}(iii), \ref{lem:fnl}(iii), \ref{lem:fu}(ii) and \ref{lem:fnu}(ii). 
	
	For part (iv), we have from Lemma \ref{lem:fl} (i) and (ii) that $\ell_\calL(\tM) > \ell_\calL(\tM^0)$ for any $\tM \neq \tM^0$, and from \ref{lem:fu} (i) that $\ell_\calU(\btheta) \geq \ell_\calU(\btheta^0)$ for any $\btheta \neq \btheta^0$. Combining, we have $f(\btheta) \geq f(\btheta^0)$ for all $\btheta \neq \btheta^0$.
	
	To prove sharp inequality, let $\btheta^1=(\tM^1,\bp^1)$ be such that $f(\btheta^1) = f(\btheta^0)$. Since $\ell_\calL(\tM^1) > \ell_\calL(\tM^0)$ for $\tM^1\neq \tM^0$, to have equality, $\tM^1=\tM^0$. 
	
	Now let $\tilde \ell_\calU(\bq)=\ell_\calU(\bI,\bq)=E_{r \in \calU} d_{KL}(\ba_r || \bq)$. We know $\ell_\calU$ is a strictly convex function in $\bq$ minimized at $\bq^0 = \bM^{0'}\bp$. 
	Since, $f(\btheta^1) = f(\btheta^0)$ and $\ell_\calL(\tM^1) = \ell_\calL(\tM^0)$, we must have $\tilde\ell_\calU(\bM^{1'}\bp^1)=\tilde\ell_\calU(\bM^{0'}\bp^0) = \inf_\bq\; \tilde \ell_\calU(\bq)$ implying $\bM^{1'}\bp^1=\bM^{0'}\bp^0$ and consequently $\bp^1=\bp^0$ due to Assumption 2. 
	
	This proves that $\btheta^0$ is the global minima of $f$. As $f$ has already been proved to be twice differentiable in part (ii), we have $\bJ \succ 0$. 
\end{proof}


\begin{proof}[Proof of Theorem \ref{lem:id}]
	Let $B_\eps(\btheta^0)$ be the $\ell_1$ ball of radius $\eps$ around $\btheta^0$. For part (i), we 
	first note that
	\begin{equation*}
	\{||\btheta - \btheta^{0}||_1 > \epsilon\} \subseteq \{||\tM - \tM^{0}||_1 > h \} \cup \{||\tp - \tp^{0}||_1  > \epsilon/2, ||\tM - \tM^{0}||_1 <h \}
	\end{equation*}
	where $h < \eps/2$ is a fixed, but suitably small constant which we will specify later. 
	We begin with the fact that
	\begin{equation}\label{eq:rhsterms}
	\begin{aligned}
	\text{lim inf}_{N} \text{ inf}_{\btheta \notin B_{\epsilon}(\btheta^{0})}f_{N}(\btheta) \geq \min \{&\text{lim inf}_{N} \text{ inf}_{\{\btheta:||\tM - \tM^{0}||_1 > h\}}f_{N}(\tM, \tp),\\ &\text{lim inf}_{N} \text{ inf}_{\btheta:||\tp - \tp^{0}||_1 >\epsilon/2, ||\tM - \tM^{0}||_1 <h}f_{N}(\tM, \tp) \} \end{aligned}
	\end{equation}
	
	and will show that each of the two terms in the right hand side of (\ref{eq:rhsterms}) is greater than $f(\tM^{0}, \tp^{0})$. Using Lemma \ref{lem:fnl}(iv) and Lemma \ref{lem:fnu}(iii) 
	we immediately have
	$\liminf_{N} \inf_{\{\btheta:||\tM - \tM^{0}||_1 > h\}}f_{n}(\tM, \tp) > f(\tM^{0}, \tp^{0})$.
	
	Focusing on the other term, we note that 
	\begin{equation}
	f_{N}(\tM, \tp) - f(\tM^{0}, \tp^{0}) = \ellN(\tM) - \ell_{\mathcal{L}}(\tM^{0}) + \ellu(\tM, \tp) - \ellu(\tM^{0}, \tp) + \ellu(\tM^{0}, \tp)  - \ell_{\mathcal{U}}(\tM^{0}, \tp^{0}).
	\end{equation}\label{eq:appendixineq}
	Letting $A = \{||\tp - \tp^{0}||_1 > \epsilon/2, ||\tM - \tM^{0}||_1 <h \}$  and using (\ref{eq:appendixineq}) we have
	\begin{equation}\label{eq:appendixineq2}
	\begin{aligned}
	\text{lim inf}_{N}\text{ inf}_{A}\;f_{N}(\tM, \tp) - f(\tM^{0}, \tp^{0}) > &-\text{lim sup}_{N}\text{ sup}_{A}|\ellN(\tM^{0}) - \ellN(\tM)|\\
	&-\text{lim sup}_{N}|\ellN(\tM^{0}) - \ell_{\mathcal{L}}(\tM^{0})|\\
	&-\text{lim sup}_{N}\text{sup}_{A}|\ellu(\tM, \tp) - \ellu(\tM^{0}, \tp)|\\
	&+\text{lim inf}_{N}\text{ inf}_{\{||\tp -\tp^{0}||_1 > \epsilon / 2\}}\ellu(\tM^{0}, \tp) - \ell_{\mathcal{U}}(\tM^{0}, \tp^{0})
	\end{aligned}
	\end{equation}
	Note that the second term in (\ref{eq:appendixineq2}) is 0 as $\ellN(\tM^{0}) \rightarrow \ell_{\mathcal{L}}(\tM^{0})$. Focusing on the first term, we use the mean value theorem to have:
	\begin{align*}
	\text{lim sup}_{N}\text{ sup}_{A}|\ellN(\tM^{0}) - \ell_{\mathcal{L}}(\tM^{0})| &\leq \text{lim sup}_{N}\text{ sup}_{\{||\tM - \tM^{0}||_1 < h\}}\left|\max\limits_{i,j}\frac{\partial \ellu(\tM)}{\partial M_{ij}} \right|\times \text{ sup}_{\{||\tM - \tM^{0}||_1 < h\}} ||\tM - \tM^{0}||_{1}\\
	&\leq Ch \times \text{lim sup}_{N}\text{sup}_{\{||\tM - \tM^{0}||_1 < h\}}\max\limits_{i,j}\frac{1}{N} \sum_{r=1}^{N\xi} \left(\frac{a_{rj}b_{ri}}{\sum_{i=1}^{C}M_{ij}b_{ri}} + \frac{a_{rC}b_{ri}}{\sum_{i=1}^{C}M_{iC}b_{ri}}\right) 
	\end{align*}
	
	Since $\tM^{0}$ is an interior point, for small enough $\epsilon$, we have $\forall i, j$, $M_{ij} \geq K(\epsilon)$, and thus
	$ \text{sup}_{\{||\tM - \tM^{0}||_1 < h \}}\frac{1}{\sum_{i=1}^{C}M_{ij}b_{ri}} \leq \frac{1}{K(\epsilon)}. $
	which further implies
	\begin{align*}
	\text{limsup}_{N}\text{sup}_{A}|\ellN(\tM^{0}) - \ell_{\mathcal{L}}(\tM^{0})| \leq \text{limsup}_{N} \frac{Ch}{K(\epsilon)}\sum_{i,j}\frac{1}{N}\sum_{r=1}^{N\xi}a_{rj}b_{ri} \leq \frac{C^3h}{K(\epsilon)} \;.
	\end{align*}
	
	Using the same logic as above, we also have
	\begin{equation*}
	\text{limsup}_{N}\text{sup}_{A}|\ellu(\tM, \tp) - \ellu(\tM^{0}, \tp)| \leq \frac{C^3h}{K(\epsilon)} \;.
	\end{equation*}
	
	Combining this, we have 
	\begin{equation}
	\text{lim inf}_{N}\text{ inf}_{A}f_{N}(\tM, \tp) - f(\tM^{0}, \tp^{0}) > \frac{-2C^3h}{K(\epsilon)} + \text{lim inf}_{N}\text{ inf}_{\{||\tp -\tp^{0}||_1 > \epsilon / 2\}}\ellu(\tM^{0}, \tp) - \ell_{\mathcal{U}}(\tM^{0}, \tp^{0})
	\end{equation}
	
	We define $\ellu^{*}(\tp) = \ellu(\tM^{0}, \tp)$ which as $N \rightarrow \infty$ goes to $\ell_{\mathcal{U}}^{*}(\tp)=\ell_\calU(\tM^0,\tp)$. 
	First we will show that $\ellu^{*}(\tp)$ is convex. We have
	\begin{equation*}
	\ellu^{*} (\tp)= -\frac{1}{N}\sum_{r=1}^{N}\sum_{j=1}^{C}a_{rj}\log\left( \frac{\sum_{i=1}^{C}M_{ij}^{0}p_{i}}{a_{rj}}\right)
	\end{equation*}
	which implies
	\begin{equation*}
	\frac{\partial^{2}\ellu^{*} (\tp)}{\partial p_{i}\partial p_{i^{'}}} = \sum_{j=1}^{C}\hat{\tilde{q}}_{j}\frac{(M_{ij}^{0} - M_{Cj}^{0})(M_{i^{'}j}^{0} - M_{Cj}^{0})}{(\sum_{i=1}^{C}M_{ij}^{0}p_{i})^2}
	\end{equation*}
	
	Now letting
	\begin{align*}
	\mathbf{D} = \text{diag}(d_{j}) \mbox{ where } d_{j} &= \frac{\hat{\tilde{q}}_{j}}{(\sum_{i=1}^{C}M_{ij}^{0}p_{i})^2}, \\
	\mathbf{U} = (\mathbf{u}_{1}, \ldots, \mathbf{u}_{C}) \mbox{ where } \mathbf{u}_{j} &= \bM_{1:(C-1), j}^{0} - M_{Cj}^{0} \bone_{C-1}
	\end{align*}
	we have $
	\nabla^{2}\ellu^{*} (\tp) = \sum_{j=1}^{C}d_{j}\mathbf{u}_{j}\mathbf{u}_{j^{'}} = \mathbf{U}\mathbf{D}\mathbf{U}^{'} $. As $\hat{\tilde{q}}_{j} = \frac 1N \sum_{r=1}^N a_{rj} \to E_\calU(a_{rj}) = \sum_i M^{0}_{ij}p^0_i > 0$, there exists a set $\calA$ with $P(\calA)=1$, such that on $\calA$, $\hat{\tilde{q}}_{j} > 0$ for large enough $N$. 
	Now note that the rows of $\mathbf{U}$ are linear combinations of rows of $\bM^{0}$:
	\begin{equation*}
	\mathbf{U} = \begin{pmatrix}
	e_{1}^{'} - e_{C}^{'}\\
	e_{2}^{'} - e_{C}^{'}\\
	\vdots\\
	e_{C-1}^{'} - e_{C}^{'}\\
	\end{pmatrix}\bM^{0}
	\end{equation*}
	And thus by assumption 2, that $\bM^{0}$ is full rank, $U$ is also full row-rank and hence on $\calA$, $\ellu^{*} (\tp)$ is convex.
	
	Next, we now look at the properties of $\nabla \ell_{\mathcal{U}}^{*}(\tp^{0})$ and $\nabla^{2} \ell_{\mathcal{U}}^{*}(\tp^{0})$.
	We have
	\begin{align*}
	\ell_{\mathcal{U}}^{*}(\tp)&= -\sum_{j=1}^{C}\left[E_{r \in \mathcal{U}}\left[\hat{\tilde{q}}_{j}\log(\sum_{i=1}^{C}M_{ij}^{0}p_{i})- a_{rj}\log(a_{rj}) \right] \right]\\
	&= -\sum_{j=1}^{C}\left[(\sum_{i=1}^{C}M_{ij}^{0}p_{i}^{0})\log(\sum_{i=1}^{C}M_{ij}^{0}p_{i})-E_{r \in \mathcal{U}}\left[a_{rj}\log(a_{rj})\right] \right] \numberthis
	\end{align*}
	and thus
	\begin{align*}
	\frac{\partial \ell_{\mathcal{U}}^{*}(\tp)}{\partial p_{i}}\bigg\rvert_{\tp = \tp^{0}} &= -\sum_{j=1}^{C}\frac{(\sum_{i=1}^{C}M_{ij}^{0}p_{i}^{0}) (M_{ij}^{0}-M_{Cj}^{0})}{(\sum_{i=1}^{C}M_{ij}^{0}p_{i}^{0})}\\
	&= -\sum_{j=1}^{C}M_{ij}^{0} + \sum_{j=1}^{C}M_{Cj}^{0}\\
	&= 0 \numberthis
	\end{align*}
	
	And finally looking at $\nabla^{2} \ell_{\mathcal{U}}^{*}(\tp^{0})$ we have
	\begin{equation}
	\nabla^{2} \ell_{\mathcal{U}}^{*}(\tp^{0})\bigg\rvert_{\tp = \tp^{0}} = \mathbf{U}E_{\mathcal{U}}\mathbf{D}\mathbf{U}^{'} \succ 0
	\end{equation}
	
	and thus by Theorem 2.3 of \citep{miller2019asymptotic} we have
	
	\begin{equation}
	\liminf_{N}\inf_{\{||\tp -\tp^{0}||_1 > \epsilon / 2\}}\ellu(\tM^{0}, \tp) - \ell_{\mathcal{U}}(\tM^{0}, \tp^{0}) > \delta
	\end{equation}
	
	for some $\delta > 0$. We now return to the constant $h$ and note that by choosing $h < \frac{\delta K(\epsilon)}{2C^{3}}$, we have
	
	\begin{equation}\label{eq:idp}
	\liminf_{N}\inf_{A}f_{N}(\tM, \tp) - f(\tM^{0}, \tp^{0}) > 0
	\end{equation}
	
	Thus we have proved that both terms in the right hand side of (\ref{eq:rhsterms}) are greater than $0$. Using $\kappa$ to be the minimum of the two terms, 
	we have
	proved the first statement of the Lemma. 
	To prove the second part of the Lemma, let $G_N=\inf_{\|\btheta-\btheta^0\| \geq \eps}f_{N}(\tM, \tp) - f_N(\tM^{0}, \tp^{0})$. Note that $\liminf_N I(G_N \geq \kappa) = I(\liminf_N G_N \geq \kappa)$. We then have
	\begin{align*}
	\liminf_N P(\inf_{\|\btheta-\btheta^0\| \geq \eps}\, f_{N}(\tM, \tp) - f_N(\tM^{0}, \tp^{0}) \geq \kappa) =& \liminf_N E(I(G_N \geq \kappa)) \\
	\geq& E (\liminf_N I(G_N \geq \kappa))\\
	=& E (I(\liminf_N G_N \geq \kappa)) =1.
	\end{align*}
	Here the first inequality is from Fatou's lemma, and the last equality comes from the fact that using the first statement of this Lemma we can say 
	\begin{align*}
	\liminf_N \inf_{\|\btheta-\btheta^0\| \geq \eps}f_{N}(\tM, \tp) - f_N(\tM^{0}, \tp^{0}) &\geq \lim \inf_N \inf_{\|\btheta-\btheta^0\| \geq \eps}f_{N}(\tM, \tp) - f(\tM^{0}, \tp^{0}) \\
	& \qquad + \liminf_N f(\tM^{0}, \tp^{0}) - f_{N}(\tM^0, \tp^0)  \geq \kappa + 0.
	\end{align*}
\end{proof}

\begin{proof}[Proof of Theorem \ref{th:post}]
	Lemma \ref{lem:f}(i) proves pointwise limit of $f_N$, Theorem \ref{lem:id} part (i) proves the identifiability. Hence, the theorem is proved by applying Condition 1 of Theorem 2.3 of \cite{miller2019asymptotic}. 
\end{proof}

\begin{proof}[Proof of Theorem \ref{th:norm}] The posterior mean
	equals $\arg \min_{\bzeta} \int \rho(\btheta-\bzeta)\nu(\btheta)$ where $\nu$ is the Gibbs posterior, and
	$\rho(\bx)=\|\bx\|^2$. We already know that $\btheta^0$ is an interior point of the parameter space, and that $f_N$ satisfies the identifiability condition of Theorem \ref{lem:id} part (ii). Hence, $\btheta^0$, $\rho$ and $f_N$ respectively satisfy conditions 1, 2 and 3 of \cite{chernozhukov2003mcmc}. 
	
	Now, from Lemma \ref{lem:f} (ii) we have $\nabla^3 (f_N - f)$ is uniformly bounded (by some $K$) in any small enough neighborhood $B_\delta(\btheta^0)$. Hence, on $B_\delta(\btheta^0)$, $|\nabla^2 (f_N - f)(\btheta)| \leq |\nabla^2 (f_N - f)(\btheta^0)| + K \delta$. 
	Using $\delta=\eps/(2K)$ we have
	\begin{align*}
	P(\sup_{\btheta \in B_\delta(\btheta^0)} |\nabla^2 (f_N - f)(\btheta)| > \eps) \leq P( |\nabla^2 (f_N - f)(\btheta^0)| > \eps/2) \to 0. 
	\end{align*}
	
	This proves that $f_N$ satisfies Lemma 2 (and hence Condition 4) of \cite{chernozhukov2003mcmc}. The asymptotic normality now follows from Theorem 2 of \cite{chernozhukov2003mcmc}.
\end{proof}

\begin{proof}[Proof of Theorem \ref{th:cp}] We have already proved asymptotic normality of $\sqrt{N}g(\hat\btheta)$ (and hence consistency of $g(\hat\btheta)$) in Theorem \ref{th:norm}. Hence, we only need to prove $\widehat\bJ \to_p \bJ$ and $\widehat\bOmega \to_p \bOmega$ to ensure asymptotically valid coverage of the intervals by Theorem 4 of \cite{chernozhukov2003mcmc}. 
	From Theorem \ref{th:norm}, we have $\hat\btheta \to_p \btheta^0$, hence $\widehat \bq = \widehat\bM'\widehat\bp \to_p \bq^0$. As $\bM^0$ is interior, $q^0_j$ is bounded away from $0$ for all $j$, hence, $1/\hat q_j \to_p q^0_j$. So we have, $$\widehat{\bU^{0'}_{\bM}}= \widehat\bp \otimes \left[diag(1/\widehat \bq_{1:C-1})_{C-1 \times C-1}; - 1/\hat q_C \bone_{C-1 \times 1} \right] \to_p \bU^{0'}_{\bM}.$$
	Similarly, we have $\widehat{\bU^{0'}_{\bp}} \to_p \bU^{0'}_{\bp}$. As $\ba_r$'s for $r \in \calU$ are iid with covariance $\bV_{\ba,\calU}$, we immediately have the sample covariance $\widehat{\bV_{\ba,\calU}} \to_p \bV_{\ba,\calU}$. Hence the first term in the expression of $\widehat\bOmega$ goes in probability to the corresponding term of $\bOmega$. 
	
	Next, we need to show $\widehat{\bV_{\hat g,\calL}} \to_p \bV_{g,\calL}$. As $g_r=(\bb_r \otimes \bD_r)\ba_r$ are iid for $r \in \calL$, their sample covariance $\widehat{\bV_{g,\calL}} \to_p \bV_{g,\calL}$. Hence, it is enough to show $\widehat{\bV_{\hat g,\calL}}- \widehat{\bV_{g,\calL}} \to_p \bO$. Recall that both these matrices are of dimension $C(C-1)$ with rows and columns indexed by the pairs $ij$, for $i=1,\ldots,C$, $j=1\ldots,C-1$. Let $x_r$ and $\hat x_r$ denote the respective components of $g_r$ and $\hat g_r$ corresponding to $M_{ij}$. From (\ref{eq:gr}), 
	$x_r = b_{ri} (\frac {a_{rj}}{\sum_i M^0_{ij}b_{ri}} - \frac {a_{rC}}{\sum_i M^0_{iC}b_{ri}})$. So $|x_r|$ is bounded by $2/K$ where $0 < K=\min_{ij} M^0_{ij}$. Similarly, $|\hat x_r|$ is bounded by $2/\min_{ij} \widehat M_{ij}$. Now, letting $\bar x$ and $\bar{\hat x}$ respectively to be the sample means of $x_r$ and $\hat x_r$ for $r 
	\in \calL$, we have 
	\begin{align*}
	|\widehat{\bV_{\hat g,\calL}}- \widehat{\bV_{g,\calL}}|_{ij,ij} &= |(\frac 1{\xi N} \sum_{r=1}^{\xi N} (x_r^2) - \bar x^2) - (\frac 1{\xi N} \sum_{r=1}^{\xi N} (\hat x_r^2) - \bar{\hat x}^2)|  \\
	&\leqslant (4/K + 4/\min_{ij} \widehat M_{ij}) \max_{r \in \calL} |x_r - \hat x_r|\\
	& \leqslant (4/K + 4/\min_{ij}\widehat M_{ij})  \max_{r \in \calL} b_{ri} \left( a_{rj}\big|\frac {1}{\sum_i M^0_{ij}b_{ri}} - \frac {1}{\sum_i \widehat M_{ij}b_{ri}} \big| + a_{rC}\big|\frac {1}{\sum_i M^0_{iC}b_{ri}} - \frac {1}{\sum_i \widehat M_{iC}b_{ri}}  \big|\right)\\
	&\leqslant (4/K + 4/\min_{ij}\widehat M_{ij}) \frac {4 \max_{i} |M^0_{ij} - \widehat M_{ij}| + 4\max_{i} |M^0_{iC} - \widehat M_{iC}|}{K\min_{ij}\widehat M_{ij}}.
	\end{align*}
	Here the last inequality follows from the fact that $a_{rj}$'s and $b_{ri}$'s are not greater than 1, and that $|\sum_i (\widehat M_{ij} - M^0_{ij})b_{ri}| \leq \max_i |M_{ij} - M^0_{ij}| \sum_i b_{ri} = \max_i |M_{ij} - M^0_{ij}|$ as $\sum_i b_{ri}=1$. As $\widehat M_{ij} \to_p M^0_{ij}$, and there are only $C(C-1)$ such terms, we have $\min_{ij} \widehat M_{ij} \to_p \min_{ij} M^0_{ij}=K > 0$, and $\max_{i} |M^0_{ij} - \widehat M_{ij}| \to_p 0$ for all $j$, proving $|\widehat{\bV_{\hat g,\calL}}- \widehat{\bV_{g,\calL}}|_{ij,ij} \to_p 0$. Repeating this for all the other entries of the matrices, proves $|\widehat{\bV_{\hat g,\calL}}- \widehat{\bV_{g,\calL}}| \to_p \bO$, and hence $\widehat\bOmega \to_p \bOmega$. 
	
	Next to show that $\widehat \bJ \to_p \bJ$, note from (\ref{eq:gradfln}) that 
	$\frac{\partial^{2} \ellN}{\partial M_{ij} \partial M_{i^{'}j^{'}}} = I(j=j^{'}) \frac{1}{N}\sum_{r=1}^{\xi N} \frac{a_{rj}b_{ri}b_{ri^{'}}}{(\sum_{i=1}^{C}M_{ij}b_{ri})^{2}} + \frac{1}{N}\sum_{r=1}^{\xi N} \frac{a_{rC}b_{ri}b_{ri^{'}}}{(\sum_{i=1}^{C}M_{iC}b_{ri})^{2}}$.
	Once again noting $\min_{ij} \widehat M_{ij} \leq \sum_i \widehat M_{ij}b_{ri} \leq 1$, and $K \leq \sum_i M^0_{ij}b_{ri} \leq 1$ we have
	\begin{align*}
	|\widehat J_{M_{ij},M_{i'j'}} - J_{M_{ij},M_{i'j'}}| \leqslant & \frac{1}{N}\sum_{r=1}^{\xi N} a_{rj}b_{ri}b_{ri^{'}} \big|\frac{1}{(\sum_{i=1}^{C}\widehat M_{ij}b_{ri})^{2}}-\frac{1}{(\sum_{i=1}^{C}M^0_{ij}b_{ri})^{2}}\big| \\
	& \qquad + \frac{1}{N}\sum_{r=1}^{\xi N}  a_{rC}b_{ri}b_{ri^{'}} \big|\frac{1}{(\sum_{i=1}^{C}\widehat M_{iC}b_{ri})^{2}} -  \frac{1}{(\sum_{i=1}^{C}M^0_{iC}b_{ri})^{2}}|\\
	& \leqslant \xi \frac{2}{(K\min_{ij} \widehat M_{ij})^2} \left(\max_i |\widehat M_{ij} - M^0_{ij}| + \max_i |\widehat M_{iC} - M^0_{iC}| \right) \to_p 0.
	\end{align*}
	We can prove, the same way, the other entries of $\widehat \bJ$ goes to the corresponding entries of $\widehat \bJ$ as the denominator of all the terms in the Hessian involve either $\sum_i \widehat M_{ij}b_{ri}$ or $\sum_i \widehat M_{ij}\hat p_{i}$ both of which are bounded from below by $\min_{ij} \widehat M_{ij}$. 
\end{proof}


%
%
%

\begin{proof}[Proof of Theorem \ref{th:rate}] We will use $R$ as an universal constant whose value may change from line to line. Let $\eps_1=\sqrt{R\eps}$, and define  $B_{N,\eps_1}(\btheta^0)=\left\{\btheta \,\given \int \log \left(\frac{\tilde p_N(\btheta^0)}{\tilde p_N(\btheta)}\right) dP_N^0 < N\eps_1^2, \int \log^2\left(\frac{\tilde p_N(\btheta^0)}{\tilde p_N(\btheta)}\right) dP_N^0 < N\eps_1^2 \right\}$.
	
	Now, 
	\begin{align*}
	\int \log \left(\frac{\tilde p_N(\btheta^0)}{\tilde p_N(\btheta)}\right) dP_N^0 = N \int (f_N(\btheta) - f_N(\btheta^0)) dP_N^0 = N (f(\btheta) - f(\btheta^0)) 
	\end{align*}
	By Lemma \ref{lem:f}, $f(\btheta)$ is continuously differentiable around a neighborhood of $\btheta^0$. Hence, using Lipschitz continuity, 
	$f(\btheta) - f(\btheta^0) \leq R\|\btheta^0 - \btheta\|$ for some $R > 0$, and we have on $B_\eps(\btheta^0)$,  
	$$ \int \log \left(\frac{\tilde p_N(\btheta^0)}{\tilde p_N(\btheta)}\right) dP_N^0 < NR\eps = N\eps_1^2.$$
	
	Next we look at the squared pseudo-KL divergence,
	\begin{align*}
	\int \log^2 \left(\frac{\tilde p_N(\btheta^0)}{\tilde p_N(\btheta)}\right) dP_N^0 &= N E_\calU \left[ \sum_{j=1}^{C}a_{j}\log\left(\frac {\sum_{i=1}^{C} M^0_{ij}p^0_{i}}{\sum_{i=1}^{C} M_{ij}p_{i}} \right)\right]^2 + \xi N E_\calL\left[\sum_{j=1}^{C}a_{j}\log\left(\frac {\sum_{i=1}^{C} M^0_{ij}b_{i}}{\sum_{i=1}^{C} M_{ij}b_{i}} \right)\right]^2\\
	&\leqslant C N E_\calU  \sum_{j=1}^{C}a^2_{j}\log^2\left(\frac {\sum_{i=1}^{C} M^0_{ij}p^0_{i}}{\sum_{i=1}^{C} M_{ij}p_{i}} \right) + \xi C N E_\calL\sum_{j=1}^{C}a^2_{j}\log^2\left(\frac {\sum_{i=1}^{C} M^0_{ij}b_{i}}{\sum_{i=1}^{C} M_{ij}b_{i}} \right)
	\end{align*}
	
	Here the last inequality applies $(\sum_{i=1}^C u_i)^2 \leq C\sum_{i=1}^C u_i^2$. 
	For any $\bx \in \widetilde S^C$, we have on $B_\eps(\btheta^0)$, $|\sum_{i=1}^{C} M_{ij}x_{i} - \sum_{i=1}^{C} M^0_{ij}x_{i}| \leq \eps$. 
	Also, as $\bM^0$ is an interior point, for small enough $\eps$ we have  on $B_\eps(\btheta^0)$, 
	$0 < K \leqslant \sum_{i=1}^{C} M_{ij}x_{i}$ for all $\bx \in \widetilde S^C$. Hence, using Lipschitz continuity of $\partial \log^2$ in any compact interval bounded away from $0$, and that $E(a_j^2) \leq 1$, we have $\log^2\left(\frac {\sum_{i=1}^{C} M^0_{ij}b_{i}}{\sum_{i=1}^{C} M_{ij}b_{i}} \right) < R\eps$ for all $\bb$ on $B_\eps(\btheta^0)$. Similarly, as $\bp^0$ is also interior, and on $B_\eps(\btheta^0)$, $|\sum_i p_ix_i - p^0_i x_i| \leq \eps$ for all $\bx \in \widetilde S^C$, we will have 
	$$\log^2\left(\frac {\sum_{i=1}^{C} M^0_{ij}p^0_{i}}{\sum_{i=1}^{C} M_{ij}p_{i}} \right) \leq 2 \log^2\left(\frac {\sum_{i=1}^{C} M^0_{ij}p^0_{i}}{\sum_{i=1}^{C} M_{ij}p^0_{i}} \right) + 2\log^2\left(\frac {\sum_{i=1}^{C} M_{ij}p^0_{i}}{\sum_{i=1}^{C} M_{ij}p_{i}} \right) \leq R\eps.$$
	
	Combining, we have 
	$$ \int \log^2 \left(\frac{\tilde p_N(\btheta^0)}{\tilde p_N(\btheta)}\right) dP_N^0 < NR\eps.$$
	
	As $P_{\Pi_N}(B_\eps(\btheta^0)) \geq \exp(-NR\eps)$, we thus have $P_{\Pi_N}(B_{N,\eps_1}(\btheta^0)) \geq \exp(-NR\eps) = \exp(-N\eps_1^2)$. The rest of the proof follows the ideas in \cite[proof of Theorem 3.1]{bhattacharya2019bayesian}. Let $U_n = \{\btheta \given D_{N,\alpha}(\btheta,\btheta^0) \geqslant (D+3t)N\eps_1^2\}$.
	We have
	\begin{align*}
	P_{\nu_N}\left(D_{N,\alpha}(\btheta,\btheta^0) \geqslant (D+3t)N\eps_1^2\right) = \frac{\int_{U_n} \exp(-\alpha Nf_N(\btheta))d\Pi_N}{\int_{\Theta} \exp(-\alpha Nf_N(\btheta))d\Pi_N} \leqslant \frac{\int_{U_n} \exp(\alpha Nf_N(\btheta^0)-\alpha Nf_N(\btheta))d\Pi_N}{\int_{B_{N,\eps_1}(\btheta^0)} \exp(\alpha Nf_N(\btheta^0)-\alpha Nf_N(\btheta))d\Pi_N}.
	\end{align*}
	
	
	We first consider the numerator. We have 
	\begin{align*}
	E^N_0 \int_{U_N} \exp(\alpha Nf_N(\btheta^0)-\alpha Nf_N(\btheta))d\Pi_N &= \int \int_{U_N} \exp(\alpha Nf_N(\btheta^0)-\alpha Nf_N(\btheta))d\Pi_N dP^0_N\\
	&= \int_{U_N} \int  \exp(\alpha Nf_N(\btheta^0)-\alpha Nf_N(\btheta)) dP^0_N d\Pi_N \\
	&= \int_{U_N} \exp(-  D_{N,\alpha}(\btheta,\btheta^0)) d\Pi_N \\ &\leqslant \exp(-(D+3t))N\eps_1^2. 
	\end{align*}
	Application of Markov inequality yields
	\begin{align*}
	P_N^0&\left(\int_{U_N} \exp(\alpha Nf_N(\btheta^0)-\alpha Nf_N(\btheta))d\Pi_N  > \exp(-(D+2t)N\eps_1^2)\right) \\
	& \qquad \leqslant \exp((D+2t)N\eps_1^2) E^N_0\left(\int_{U_N} \exp(\alpha Nf_N(\btheta^0)-\alpha Nf_N(\btheta))d\Pi_N\right) \\
	& \qquad \leqslant \exp(-Nt\eps_1^2). 
	\end{align*}
	For the denominator, as we have established $P_{\Pi_N}(B_{N,\eps_1}(\btheta^0)) \geq \exp(-N\eps_1^2)$, we can directly use the result of \cite[proof of Theorem 3.1]{bhattacharya2019bayesian} to have 
	\begin{align*}
	P_N^0&\left(\int_{B_{N,\eps_1}(\btheta^0)} \exp(\alpha Nf_N(\btheta^0)-\alpha Nf_N(\btheta))d\Pi_N  \leqslant \exp(-\alpha(D+t)N\eps_1^2) \right) \leq \frac 1{N(D+t-1)^2\eps^2_1}.
	\end{align*}

	Combining the probabilities for the numerator and denominator, we have 
	\begin{align*}
	P_{\nu_N}\left(D_{N,\alpha}(\btheta,\btheta^0) \geqslant (D+3t)N\eps_1^2\right) \leq \exp(-(D+2t)N\eps_1^2) \times \exp(\alpha(D+t)N\eps_1^2) \leq \exp(-tN\eps_1^2)
	\end{align*}
	with $P_N^0$-probability $\geq 1 - \frac 1{N(D-1+t)^2\eps_1^2} - \exp(-Nt\eps_1^2) \geq 1 - \frac 1{N(D-1+t)^2\eps_1^2} - \frac 1{Nt\eps_1^2} \geq 1 - \frac 2{NR\eps\min\{{(D-1+t)^2,t}\}}$. 
\end{proof}

\begin{proof}[Proof of Corollary \ref{cor:paramrate}] Let $\Pi(\bp,\bM)$ denote the independent Dirichlet prior for $\btheta=(\bp,\bM)$, i.e., 
$$ \Pi(\bp,\bM) = \mbox{Dirichlet}(\bp \given \balpha_p) \times \prod_{i=1}^C \mbox{Dirichlet}(\bM_{i*} \given \balpha_{Mi}).$$
We first verify that this prior satisfies the prior mass condition of Theorem \ref{th:rate} that $\Pi(B_{\eps_N}(\btheta^0)) > \exp(-NR\eps_N)$ for the choice of $\eps_N=\log N / N$. Let $\eps_N'=\eps_N/(C+1)$. When $\|\bp - \bp^0\|_1 < \eps_N'$ and $\|\bM_{i*} - \bM^0_{i*}\|_1 < \eps_N'$ for all $i=1,\ldots,C$, we have $\|\btheta - \btheta^0\|_1 \leq \eps_N$. Hence, using the independence of the priors, we have
$$
\begin{aligned}
\Pi(\btheta \in B_{\eps_N}(\btheta^0)) &\geq \Pi(\bp \in B_{\eps_N'}(\bp^0)) \times \prod_{i=1}^C \Pi(\bM_{i*} \in B_{\eps_N'}(\bM_{i*}^0))\\
&\geq R_1^{C+1} \exp\left(-(C+1)R_2 \log \left(\frac 1{\eps_N'}\right)\right)
\end{aligned}
$$
Here, the last inequality is directly taken from the concentration bound of the Dirichlet distribution derived in Lemma 6.1 of \cite{ghosal2000convergence}. $R_1$ and $R_2$ are universal constants depending solely on the prior hyper-parameters $\balpha_p$ and $\balpha_{Mi}$.
Letting $R_2^*=(C+1)R_2$ and $R_1^* = R_1^{C+1}/\exp(R_2^*\log(C+1))$ we have the prior mass bounded from below by 
$$R_1^* \exp\left(-R_2^*\log\left(\frac 1{\eps_N}\right)\right).$$ 
Using a value of $R$ in Theorem \ref{th:rate} which is greater than $R_2^*$, as $\log N$ dominates $\log \log N$, we have for large enough $N$,
$$
\begin{aligned}
\log (R_1^*) + R \log N \geq& R_2^* (\log N - \log \log N)\\
\implies \log (R_1^*) + NR\eps_N \geq& R_2^* \log (1 / \eps_N) \\
\implies R_1^* \exp(-R_2^*\log(\frac 1{\eps_N})) \geq& \exp(-NR\eps_N).
\end{aligned}$$
This proves the prior mass condition of Theorem \ref{th:rate} with $\eps_N= \log N/N$. 

Letting $\Xi_N = P_{\nu_N}\left(\frac{D_{N,\alpha}(\btheta,\btheta^0)}N \geqslant 
M\frac{\log N}N \right)$ where $M=(D+3t)R$, we have from Theorem \ref{th:rate},
$$P^0_N \left(\Xi_N > \delta_N \right) < \eta_N$$ for large enough $N$ where  $\delta_N = \exp(-NR\eps_N) = \exp(-R\log N) \to 0$ and $\eta_N = \frac 1{NR\eps_NM'} = \frac 1{M' \log N} \to 0$ (with $M'=\min\{(D-1+t)^2,t\}/2$). This proves the Corollary. 
\end{proof}

\begin{proof}[Proof of Corollary \ref{cor:ens}]
	\noindent\textbf{Part (i).} Let $f_{N,k}$ denote the loss function for the $k^{th}$ classifier, and $f_{,k}$ denote its corresponding limit. Then $f_N=\sum_{k=1}^K f_{N,k}$ and $f=\sum_{k=1}^K f_{,k}$. We can write $\{\btheta \given \|\btheta-\btheta^0\| > \eps \} \subset \cup_{k=0}^K \calA_k$ where $\calA_k=\{\btheta \given \|\tM^{(k)} - \tM^{(k0)} \| > h\}$ for $k=1,\ldots,K$ and $\calA_0=\{\btheta \given \|\tM^{(k)} - \tM^{(k0)} \| < h\, \forall k=1,\ldots,K , \|\btheta-\btheta^0\| > \eps \}$. So we have 
	$$ \liminf_{N}\inf_{\|\btheta-\btheta^0\| > \eps} f_{N}(\btheta) - f(\btheta^{0}) \geq \min_{k \in \{0,1,\ldots,K\}} \liminf_{N}\inf_{\calA_k} f_{N}(\btheta)- f(\btheta^{0})$$. 
	For $k \geq 1$, on $\calA_k$, we have from Lemma \ref{lem:fnl}(iv), $\liminf_{N}\inf_{\calA_k} f_{N,k}(\btheta)- f_{,k}(\btheta^{0}) > \kappa_{k}$ for some $\kappa_k >0$. So, 
	\begin{align*}
	\liminf_{N}\inf_{\calA_k} f_{N}(\btheta)- f(\btheta^{0}) &\geq \liminf_{N}\inf_{\calA_k} f_{N,k}(\btheta)- f_{,k}(\btheta^{0}) + \sum_{k' \neq k} \liminf_{N}\inf_{\calA_k} f_{N,k'}(\btheta)- f_{,k'}(\btheta^{0})\\
	& \geq \liminf_{N}\inf_{\calA_k} f_{N,k}(\btheta)- f_{,k}(\btheta^{0}) + \sum_{k' \neq k} \liminf_{N}\inf_{\bTheta} f_{N,k'}(\btheta)- f_{,k'}(\btheta^{0}) \\
	& \geq \kappa_k.
	\end{align*}
	
	On $\calA_0$, as $\|\tM^{(k)} - \tM^{(k0)} \| < h$  for all $k=1,\ldots,K$, for small enough $h$, we have $\|\tp - \tp^0\| > \eps/2$. We thus have 
	\begin{align*}
	\liminf_{N}\inf_{\calA_0} f_{N}(\btheta)- f(\btheta^{0}) &\geq \sum_{k=1}^K \liminf_{N}\inf_{\calA_0} f_{N,k}(\btheta)- f_{,k}(\btheta^{0}) \\
	&\geq \sum_{k=1}^K \liminf_{N}\inf_{\{\|\tM^{(k)} - \tM^{(k0)} \| < h,\|\tp - \tp^0\| > \eps/2\}} f_{N,k}(\btheta)- f_{,k}(\btheta^{0}) \\
	&\geq K\delta \mbox{ for some } \delta > 0 \mbox{ and } h < \mbox{ some } h_0.
	\end{align*}
	Here the last inequality comes from the proof of Theorem \ref{lem:id}. Combining, we have 
	\begin{equation}\label{eq:idens}
	\liminf_{N}\inf_{\|\btheta-\btheta^0\| > \eps} f_{N}(\btheta) - f(\btheta^{0}) \geq  \min\{\kappa_1,\ldots,\kappa_K,K\delta\} > 0.
	\end{equation}
	We have thus established the analogue of the identifiability result of Theorem \ref{lem:id}(i) for ensemble GBQL. The posterior consistency result for part (i) of this corollary now follows similar to the proof of Theorem \ref{th:post}.
	
	\noindent\textbf{Part (ii).} Using (\ref{eq:idens}) we can prove the analogue of Theorem \ref{lem:id} (ii) for ensemble GBQL in the same way Theorem \ref{lem:id} (i) was used to prove Theorem \ref{lem:id} (ii). Next, we calcuate the $\bOmega$ and $\bJ$ matrices for asymptotic normality result. As $f_{N} = \sum_{k=1}^K f_{N,k}$, it is immediate that $\bJ_{ens}=\nabla_{\btheta^0}^2 f_N = \sum_{k=1}^K \nabla_{\bM^{(k0)},\bp^0}^2 f_{N,k} = \sum_{k=1}^K \bJ_i$. 
	
	Let $\ell_{\calU,N,k}$ denote the loss-function for $\calU$ for the $k^{th}$ classifier and $\ell_{\calU,N} = \sum_{k=1}^K \ell_{\calU,N,k}$. Similarly, define $\ell_{\calL,\xi N, k}$ and $\ell_{\calL,\xi N}$. Then following the proof of Lemma \ref{lem:f}, we have $\partial_{\tM^{(10)},\ldots,\tM^{(K0)}} \ell_{\calL,\xi N} \to_d N(0,\xi \bV_{g,\calL})$. 
	
	As $\bM^{(k)}$ only appears in $\ell_{\calU,N}$ through the loss $\ell_{\calL,N,k}$, we have  $ \frac 1 {\sqrt {N}} \frac{\partial \ell_{\calU,N}}{\partial M^{(k)}_{ij}}=\frac 1 {\sqrt {N}} \frac{\partial \ell_{\calU,N,k}}{\partial M^{(k)}_{ij}}$ and consequently we have, $\frac 1 {\sqrt {N}} \nabla_{\bM^{(k0)}} \ell_{\calU,N} \to_d N(0,\bU^{0'}_{\bM^{(k)}} \bV_{A,k,\calU} \bU^0_{\bM^{(k)}})$ where $\bU^{0'}_{\bM^{(k)}}$ is similar to (\ref{eq:um}) using $\bM^{(k)}$ and $\bp$, and $\bV_{A,k,\calU}=$Cov$_{r \in \calU}(\ba_r^k)$. 
	Finally, we have 
	\begin{align*}
	\frac 1 {\sqrt {N}} \frac{\partial \ell_{\calU,N}}{\partial p_{i}} =& \frac 1{\sqrt N} \sum_{r=1}^N \sum_{j=1}^C \sum_{k=1}^K a^k_{rj} \left(\frac {M^{(k0)}_{ij}-M^{(k0)}_{Cj}}{\sum_i M^{(k0)}_{ij}p^0_i}\right) = (\bu^{10'}_{p_i}, \ldots, \bu^{K0'}_{p_i})' \bA_r.
	\end{align*}
	Consequently, letting $\bU^{k0'}_{\bp}$ denote the matrix similar to (\ref{eq:up}) using $\bM^{k0}$ and $\bp$, and defining $\bU^{0'}_{\bp}=(\bU^{10'}_{\bp},\ldots,\bU^{K0'}_{\bp})$
	we have, $\frac 1 {\sqrt {N}} \nabla_{\bp^0} \ell_\calU \to_d N(0,\bU^{0'}_{\bp} \bV_{A,\calU} \bU^0_{\bp})$.
	
	Combining, all this we we have $f_N/\sqrt N \to_d N(\bzero,\bOmega_{ens})$ where  
	\begin{equation*}
	\bOmega_{ens} = \bU' \bV_{A,\calU}\bU + \xi \left(\begin{array}{cc}
	\bV_{g,\calL}  & \bO \\
	\bO   & \bO
	\end{array}
	\right), \mbox{ where } \bU'=\left(
	\begin{array}{cccc}
	\bU^{0'}_{\bM^{(1)}} & \bO & \cdots & \bO \\
	\bO & \bU^{0'}_{\bM^{(2)}} & \bO & \cdots \\
	\bO & \cdots & \bO & \bU^{0'}_{\bM^{(K)}}  \\
	\multicolumn{4}{c}{\bU^{0'}_{\bp}}
	\end{array}\right).
	\end{equation*}
	
	\textbf{Part (iii).} For valid asymptotic coverage of the confidence intervals, one only needs to prove consistency of $\widehat \bJ_{ens}$ and $\widehat \bOmega_{ens}$. The proof is exactly identical to that of Theorem \ref{th:cp} and is skipped. 
	
	\textbf{Part (iv).} Let $\tilde p_N(\btheta) = \exp(-Nf_N(\btheta))$ and $\eps_1=\sqrt{KR\eps}$. As $f_N$ is sum of $K$ loss functions, one for each classifier, following the proof and notation of Theorem \ref{th:rate}, we immediately have the following
	\begin{equation*}
	\begin{aligned}
	\int \log \left(\frac{\tilde p_N(\btheta^0)}{\tilde p_N(\btheta)}\right) dP_N^0  < NKR\eps = N\eps_1^2, \int \log^2 \left(\frac{\tilde p_N(\btheta^0)}{\tilde p_N(\btheta)}\right) dP_N^0 & < N\eps_1^2\\
	P_N^0\left(\int_{U_N} \exp(\alpha Nf_N(\btheta^0)-\alpha Nf_N(\btheta))d\Pi_N  > \exp(-(D+2t)N\eps_1^2)\right) 
	& \leqslant \exp(-Nt\eps_1^2)\\
	P_N^0\left(\int_{B_{N,\eps_1}(\btheta^0)} \exp(\alpha Nf_N(\btheta^0)-\alpha Nf_N(\btheta))d\Pi_N  \leqslant \exp(-\alpha(D+t)N\eps_1^2) \right) & \leq \frac K{N(D+t-1)\eps^2_1}.
	\end{aligned}
	\end{equation*}
	
	Combining the probabilities for the numerator and denominator, we have 
	\begin{align*}
	P_{\nu_N}\left(D_{N,\alpha}(\btheta,\btheta^0) \geqslant (D+3t)N\eps_1^2\right) \leq \exp(-(D+2t)N\eps_1^2) \times \exp(\alpha(D+t)N\eps_1^2) \leq \exp(-tN\eps_1^2)
	\end{align*}
	with $P_N^0$-probability $\geq 1 - \frac K{N(D-1+t)^2\eps_1^2} - \exp(-Nt\eps_1^2) \geq 1 - \frac 1{NR\eps(D-1+t)^2} - \frac 1{NKRt\eps} \geq 1 - \frac {1+K^{-1}}{NR\eps\min\{{(D-1+t)^2,t}\}}$. 
\end{proof}

\begin{proof}[Proof of Corollary \ref{cor:coarse}] From Section \ref{sec:factor}, 
	\begin{align*}
	\tilde{f}_{N}(\tM, \tp) - f_{N}(\tM, \tp) = &-\frac{1}{N}\sum_{r=1}^{N}\sum_{j=1}^{C}\left( T_{N} \ceil[\Big]{\frac{a_{rj}}{T_{N}}} - a_{rj} \right)\log\left(\sum_{i=1}^{C}M_{ij}p_{i}\right)\\
	&-\frac{1}{N}\sum_{r=1}^{\xi N}\sum_{j=1}^{C} \left( T_{N} \ceil[\Big]{\frac{a_{rj}}{T_{N}}} - a_{rj} \right) \log\left(\sum_{i=1}^{C}M_{ij}b_{ri}\right)
	\end{align*}
	and because $\log\left(\sum_{i=1}^{C}M_{ij}p_{i}\right) < 0$ and $\ceil{Tx}/T - x > 0$, we have
	$ \tilde{f}_{N}(\tM, \tp) \geq f_{N}(\tM, \tp)$
	which along with Theorem \ref{lem:id} part(i) shows that
	\begin{equation}\label{eq:idcoarse}
	\liminf_{N}\inf_{\btheta \notin B_{\epsilon}(\btheta^{0})} \tilde{f}_{N}(\tM, \tp) > f(\tM^{0}, \tp^{0}).
	\end{equation}
	Since $\tilde f_N \to f$ as $N,T_N \to \infty$, this proves identifiability of using $\tilde f_N$ analogous to Theorem \ref{lem:id} part (i) for $f_N$. The posterior consistency of part (i) of this Corollary is now immediate like Theorem \ref{th:post}.
	
	Equation (\ref{eq:idcoarse}) also leads to the analogue of Theorem \ref{lem:id} part (ii) for $\tilde f_N$. This in turn proves the asymptotic normality. As the pointwise limit of $\tilde f_N$ is $f$, same as that of $f_N$. The matrix $\bJ$ remains the same. The matrix $\bOmega_N$ is the variance of $\nabla_{\btheta^0} \tilde f_N$. Note that as we are using the approximate function $\tilde f_N$, we will no longer have $E \nabla_{\btheta^0} \tilde f_N = 0$. However, using the same bounds of $0 < \ceil{Tx}/T -x < 1/T$, we will have $E \nabla_{\btheta^0} \tilde f_N = O(1/T_N)$ which suffices as $T_N \to \infty$. This proves part (ii). 
	
	Part (iii) follows immediately by replacing $\ba_r$ with $\ceil{T_N \ba_r}/T_N$ in the expression of $\widehat \bOmega$. 
\end{proof}

\begin{proof}[Proof of Theorem \ref{th:coarserate}] Define $f_{N,T}(\btheta)$ as the rounded and coarsened version of the loss function, $\tilde p_{N,T}(\btheta)=\exp(- Nf_{N,T}(\btheta))$ and $B_{N,T,\eps_1}(\btheta^0)=\left\{\btheta \,\given \int \log \left(\frac{\tilde p_{N,T}(\btheta^0)}{\tilde p_{N,T}(\btheta)}\right) dP_N^0 < N\eps_1^2, \int \log^2\left(\frac{\tilde p_{N,T}(\btheta^0)}{\tilde p_{N,T}(\btheta)}\right) dP_N^0 < N\eps_1^2 \right\}$. Then following the proof of Theorem \ref{th:rate}, we have
	\begin{align*}
	P_{\nu_{N,T}}\left(D_{N,\alpha}(\btheta,\btheta^0) \geqslant (D+3t)N\eps_1^2\right) \leqslant \frac{\int_{U_n} \exp(\alpha Nf_{N,T}(\btheta^0)-\alpha Nf_{N,T}(\btheta))d\Pi_N}{\int_{B_{N,T,\eps_1}(\btheta^0)} \exp(\alpha Nf_{N,T}(\btheta^0)-\alpha Nf_{N,T}(\btheta))d\Pi_N}.
	\end{align*}
	Let $X$ denote the numerator and $Y$ be the denominator. Then using Fubini's Theorem, 
	\begin{align*}
	E_N^0 X = \int_{U_n} \int \exp(\alpha Nf_N(\btheta^0)-\alpha Nf_N(\btheta)) \frac{\exp(\alpha Nf_{N}(\btheta)-\alpha Nf_{N,T}(\btheta))}{\exp(\alpha Nf_N(\btheta^0)-\alpha Nf_{N,T}(\btheta^0))} dP_N^0\,d\Pi_N. 
	\end{align*}
	
	Now, $0 < \frac {\ceil{Ta_{rj}}}T - a_{rj}$ and $\sum_i M_{ij}p_i \leq 1$. Hence,
	\begin{align*}
	\exp(\alpha Nf_{N}(\btheta)-\alpha Nf_{N,T}(\btheta)) = \prod_{r=1}^N \prod_{j=1}^C (\sum_i M_{ij}p_i)^{\alpha\left(\frac {\ceil{Ta_{rj}}}T - a_{rj}\right)} \prod_{r=1}^{\xi N} \prod_{j=1}^C (\sum_i M_{ij}b_{ri})^{\alpha\left(\frac {\ceil{Ta_{rj}}}T - a_{rj}\right)} \leq 1. 
	\end{align*}
	
	On the other hand, $\frac {\ceil{Ta_{rj}}}T - a_{rj} < \frac 1T$, implying, with $K=\min_{i,j} M^0_{ij} \in (0,1)$, we have 
	\begin{align*}
	\exp(\alpha Nf_{N}(\btheta^0)-\alpha Nf_{N,T}(\btheta^0)) &= \prod_{r=1}^N \prod_{j=1}^C (\sum_i M^0_{ij}p^0_i)^{\alpha\left(\frac {\ceil{Ta_{rj}}}T - a_{rj}\right)} \prod_{r=1}^{\xi N} \prod_{j=1}^C (\sum_i M^0_{ij}b_{ri})^{\alpha\left(\frac {\ceil{Ta_{rj}}}T - a_{rj}\right)} \\
	& \geq \left(K^{C(1+\xi)}\right)^\frac{\alpha N}{T}. 
	\end{align*}
	Letting $R_0=K^{-C(1+\xi)} > 1$ we have $E_N^0 X \leq R_0^\frac{\alpha N}{T} E^N_0\left(\int_{U_N} \exp(\alpha Nf_N(\btheta^0)-\alpha Nf_N(\btheta))d\Pi_N\right)$. Then following the proof of Theorem \ref{th:rate}, we have 
	$ P_N^0\left(X  > \exp(-(D+2t)N\eps_1^2)\right) \leqslant \exp((D+2t)N\eps_1^2) E^N_0 X  \leqslant R_0^\frac{\alpha N}{T} \exp(-Nt\eps_1^2)$. 
	
	Now for the denominator $Y$, note that the only difference between $\tilde p_{N,T}$ and $\tilde p_N$ is that the $a_{rj}$'s are replaced by $\frac {\ceil{Ta_{rj}}}T$. We have shown in Theorem \ref{th:rate} that on $B_\eps(\btheta^0)$, both $\int \log \left(\frac{\tilde p_N(\btheta^0)}{\tilde p_N(\btheta)}\right) dP_N^0$ and $\int \log^2\left(\frac{\tilde p_N(\btheta^0)}{\tilde p_N(\btheta)}\right) dP_N^0$ were less than $N\eps_1^2$. To prove this, the only property of $a_{rj}$'s used were that they were uniformly bounded by $1$. The same holds for $\frac {\ceil{Ta_{rj}}}T$ with the uniform bound $1 + 1/T \leq 2$. Hence, one can exactly replicate that part of the proof of Theorem \ref{th:rate} to show 
	$P_{\Pi_N}(B_\eps(\btheta^0)) \geq \exp(-NR\eps)$, implies $P_{\Pi_{N,T}}(B_{N,\eps_1}(\btheta^0)) \geq \exp(-NR\eps) = \exp(-N\eps_1^2)$ and consequently
	\begin{align*}
	P_N^0&\left(Y  \leqslant \exp(-(D+t)N\eps_1^2) \right) \leq \frac 1{N(D+t-1)\eps^2_1}.
	\end{align*} 
	Combining, as in Theorem \ref{th:rate}, we have 
	\begin{align*}
	P_{\nu_{N,T}}\left(D_{N,\alpha}(\btheta,\btheta^0) \geqslant (D+3t)N\eps_1^2\right) \leq \exp(-(D+2t)N\eps_1^2) \times \exp(\alpha(D+t)N\eps_1^2) \leq \exp(-tN\eps_1^2)
	\end{align*}
	with $P_N^0$-probability $\geq 1 - \frac 1{N(D-1+t)^2\eps_1^2} - R_0^\frac{\alpha N}{T}\exp(-Nt\eps_1^2) \geq 1 - \frac {1+R_0^\frac{\alpha N}{T}}{NR\eps\min\{{(D-1+t)^2,t}\}}$.
\end{proof} 

\pagebreak
\section{Additional simulation studies}\label{sec:addsim}
\subsection{Comparison of Methods for Calculating Credible Intervals}\label{sec:simcp}

While the focus of the GBQL method is the point estimate of $\bp$, we here compare methods for interval estimates for $\bp$, using average coverage probability over many replicate simulations. For these simulations, we use the same 4 choices of $\bp$ specified in (\ref{eq:simp}) and use 

\[\mathbf{M} = \begin{bmatrix}
0.65&0.25&0.02&0.02&0.02 \\
0.06&0.25&0.65&0.02&0.02 \\
0.1&0.1&0.6&0.1&0.1 \\
0.02&0.04&0.04&0.7&0.2 \\
0.02&0.3&0.03&0.05&0.6 \\
\end{bmatrix},\]

so that neither $\bp$ nor the rows of $\bM$ are on the boundary of the unit simplex. We then compare the coverage of the 95\% interval estimates for $\bp$ as obtained by taking a 95\% percentile-based credible interval based on direct samples from the Gibbs posterior, and the delta-method style interval estimates proposed in Theorem \ref{th:cp} around the Gibbs posterior means which have the asymptotic guarantee of well-calibrated coverage. 
Figure \ref{fig:coverage} shows that the credible intervals were always too conservative producing near 100\% coverage probability for every parameter under every scenario. The delta method approach was better calibrated for most parameters across scenarios. There is slight drop in coverage for 2 parameters in scenario 3 which has a true parameter value close to the simplex boundary. We also look at the mean widths of the two sets of interval estimates in Figure \ref{fig:coverage_width}. We see that the credible intervals are uniformly wider than the delta-method intervals  with the difference being more prominent for parameters whose true values are away from the boundary. 
\begin{figure}[H]
	\centering \includegraphics[width=\linewidth]{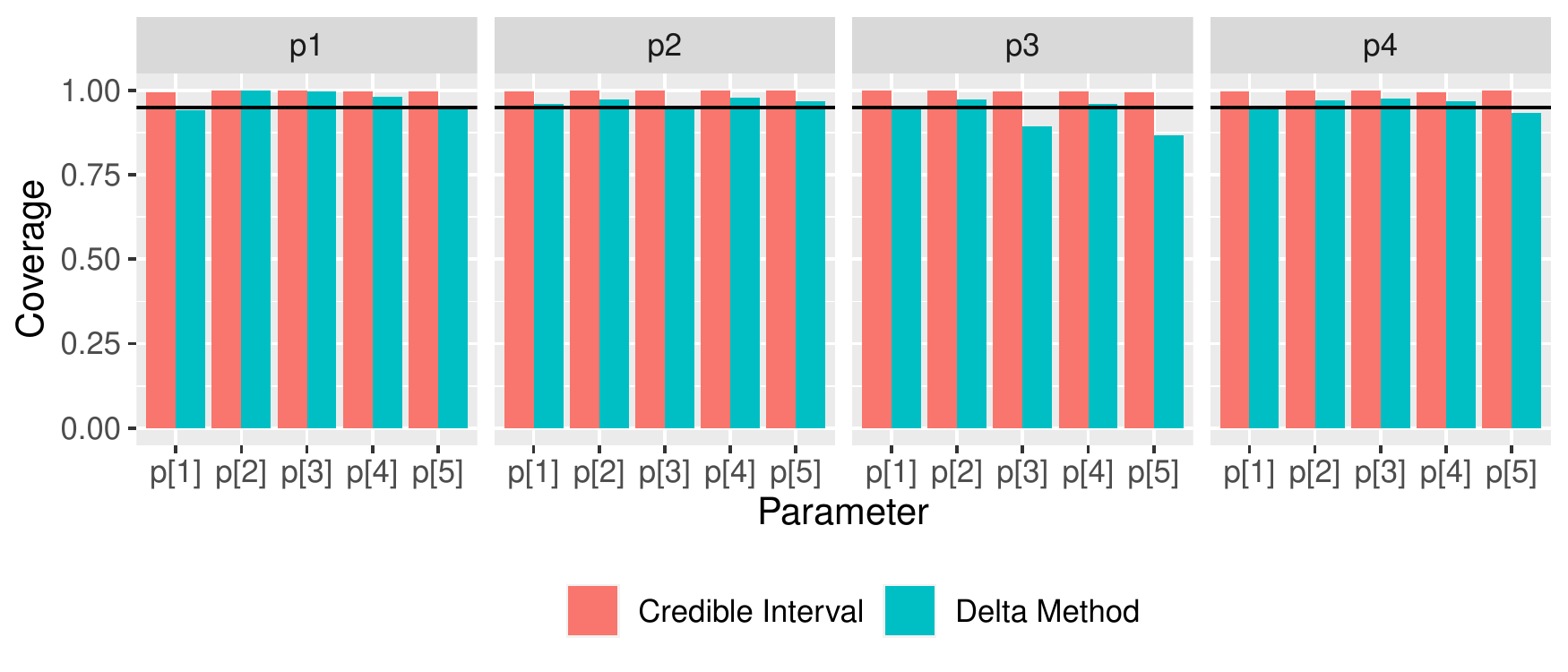}
	\caption{Mean coverage probability of Interval estimates}\label{fig:coverage}
\end{figure}

\begin{figure}[H]
	\centering 
	\includegraphics[width=\linewidth]{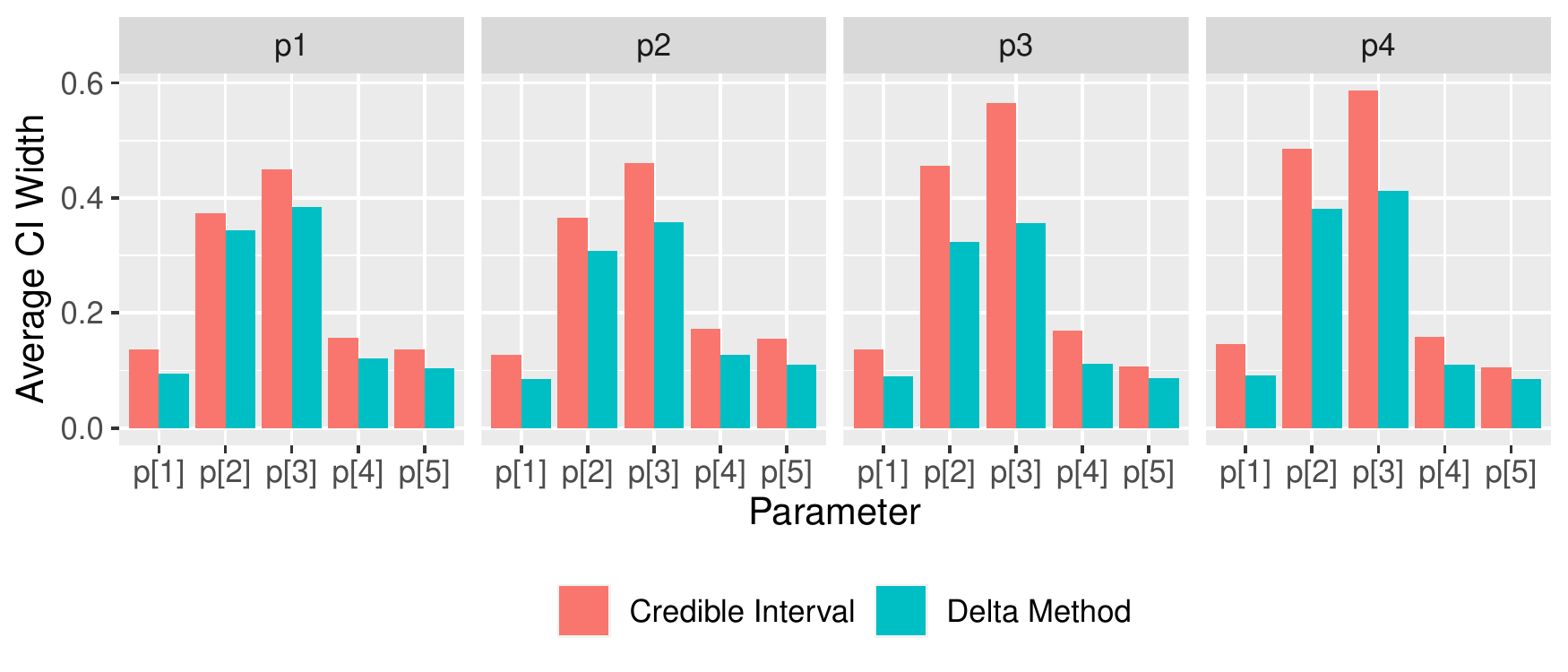}
	\caption{Mean width of Interval estimates}\label{fig:coverage_width}
\end{figure}

\subsection{Approximation of Coarsened Posterior to Full Posterior}\label{sec:simstan}
We empirically assess the accuracy of our conditional-coarsening based  Gibbs sampler compared to directly sampling from the coarsened posterior $\nu_{coarse}$ or the actual (uncoarsened) Gibbs posterior $\nu$ using RStan (Hamiltonian Monte Carlo). 
%
%
We use the same 4 parameter settings of Section \ref{simulations} and $500$ datasets for each setting and plot the average CCNAA across these replicates for our parameter of interest $\bp$. The results are presented in 
Figure \ref{fig:gbql_vs_stan_cnaa} show that each method gives approximately the same CCNAA. 

\begin{figure}[H]
	\centering \includegraphics[width=\linewidth]{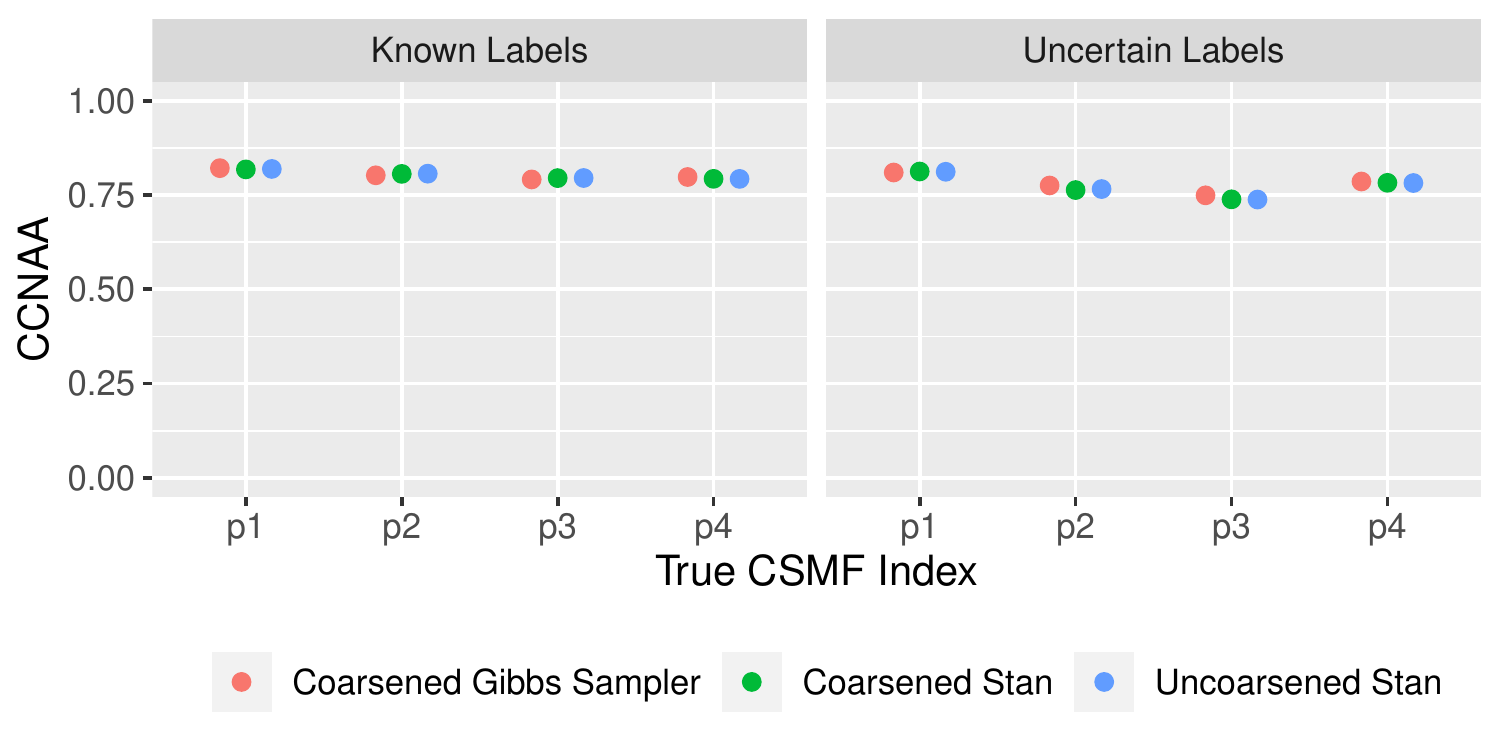}
	\caption{CCNAA of the different sampling methods, for each of the four values of $\bp$ and across known and uncertain labels for $\calL$}\label{fig:gbql_vs_stan_cnaa}
\end{figure}

Next we looked at MCMC convergence for each sampling method. For each method, we took 24,000 posterior samples ($3$ chains and $8,000$ samples per chain) and calculated the Gelman-Rubin diagnostic $\hat R$ for each component of $\bp$. Figure \ref{fig:gbql_vs_stan_rhat} shows the average $\hat{R}$ across $500$ replicate datasets for each index of $\bp$ is below 1.05 for each of the sampling methods, indicating each method generally shows convergence. 

\begin{figure}[H]
	\centering \includegraphics[width=\linewidth]{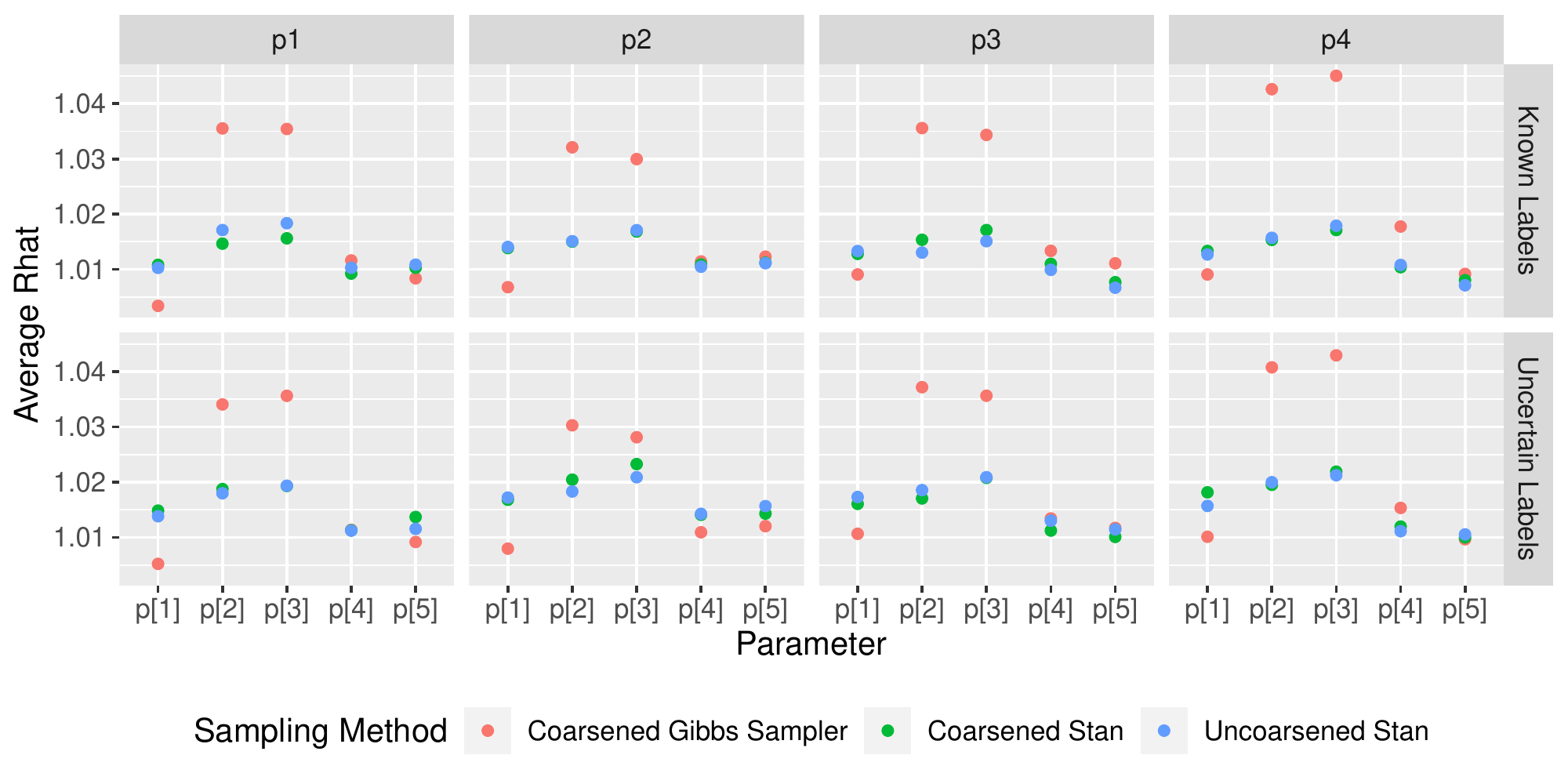}
	\caption{The Gelman-Rubin diagnostic $\hat R$ for each component of $\bp$, across the four different true values for $\bp$ for the different sampling methods.}\label{fig:gbql_vs_stan_rhat}
\end{figure}

Fianlly, we looked at the computational aspects of the 3 sampling methods using run times. Figure \ref{fig:gbql_vs_stan_timing} shows that STAN using Hamiltonian Monte Carlo (HMC) requires dramatically more time than the Gibbs sampler to obtain posterior samples.  The efficiency of the conditional-coarsening based Gibbs sampler is more noticeable (almost 7-8 times faster) for the case with known labels as compared to uncertain labels (3-4 times faster). This is due to the Gibbs Sampler sampling the discrete latent variables for subjects in $\calL$ and $\calU$, while these latent variables are marginalized out in the HMC.

\begin{figure}[H]
	\centering \includegraphics[width=\linewidth]{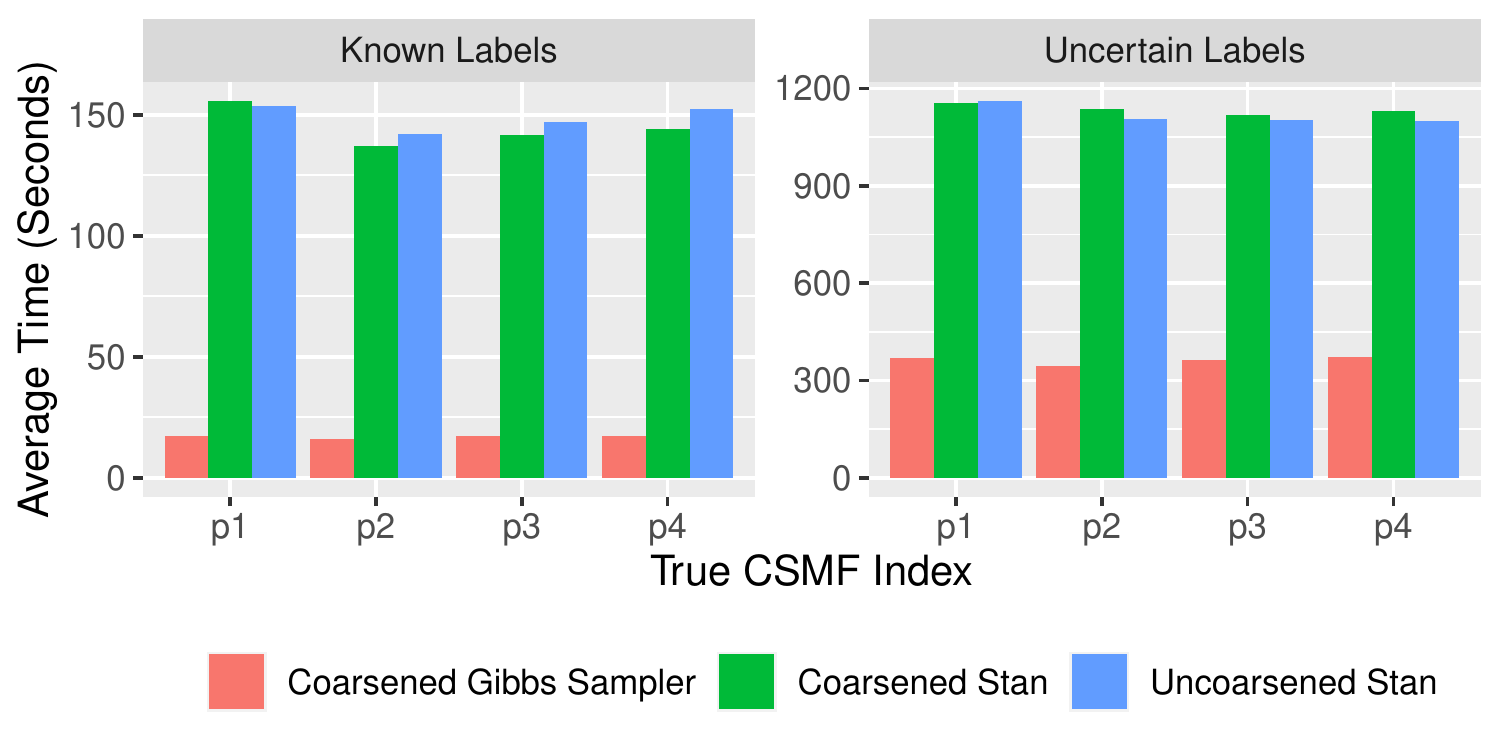}
	\caption{Timing of different posterior sampling methods, across different true values for $\bp$. For each method, we took 24,000 posterior samples (3 chains and 8,000 samples per chain)}\label{fig:gbql_vs_stan_timing}
\end{figure}

In summary, these results show that the GBQL conditional coarsening Gibbs sampler generates posterior estimates nearly indistinguishable from estimates based on direct samples from the coarsened posterior or original GBQL Gibbs posterior, while being substantially faster. 
\subsection{Sensitivity of the Coarsened Gibbs Sampler to the Level of Coarsening}\label{sec:simfactor}

While the previous simulations all used a coarsening factor of $T=100$ and produced accurate estimates of $\bp$, we evaluate the sensitivity of the Gibbs Sampler to different values of $T$. A larger $T$ ensures closer approximation by $\nu_{coarse}$ of the actual Gibbs posterior $\nu$ for GBQL. However, it also increases the number of pseudo-data simulations in the sampler. So conceptually, there is a accuracy-computation trade-off in the choice of $T$. 

In practice, however, the added computation for using larger $T$ is often negligible. To explain with an example, consider an individual with $\ba_{r} = (.1, .9, 0, 0, 0)$. For $T=10$, we will create $10$ pseudo-data $d_{rt}$ according to the model (\ref{eq:round}) such $\sum_{t=1}^{10} I(d_{rt} = 2) = 9$, while for $T=1000$, we create $1000$ pseuo-data such that  $\sum_{t=1}^{1000} I(d_{rt} = 2) = 900$. However, it is clear from the sampler steps in Section \ref{a-gibbs-sampler-for-the-posterior-belief-distribution} that for all the $90$ or $900$ choice of $t$ for which $d_{rt}=2$ the full conditional of $z_{rt} | d_{rt} = 2$ will be the same multinomial distribution. Hence, these $90$ or $900$ can be sampled at once using calls to {\em rmultinom} in R with different sample size. Thus, changing the value of $T$ simply changes how many  multinomial samples to draw. Table \ref{table:multinomtime} shows the median time to take samples from a five-dimensional multinomial distribution  in R is essentially the same for different sample sizes, showing that the sampler scales efficiently with increasing values of $T$ (unless using large $C$ or very large $T$).

\begin{table}[ht]
	\begin{center}
		\begin{tabular}{|l|l|}
			\hline
			\begin{tabular}[c]{@{}l@{}}Sample Size\end{tabular} & \begin{tabular}[c]{@{}l@{}}Median Time\\ (microseconds)\end{tabular} \\ \hline
			1                                                           & 3.45                                                                 \\ \hline
			10                                                          & 3.54                                                                 \\ \hline
			100                                                         & 3.63                                                                 \\ \hline
			1,000                                                       & 3.75                                                                 \\ \hline
			10,000                                                      & 3.66                                                                 \\ \hline
		\end{tabular}
		\caption{The median time in microseconds for R to take samples of various sizes from a uniform multinomial distribution, with 5 categories}\label{table:multinomtime}
	\end{center}
\end{table}

Next, we conduct an actual comparison of accuracy and run times for the Gibbs sampler for different choices of $T$. Figure \ref{fig:coarsening_timing} shows the average time to obtain 24,000 samples from the posterior (when $\bp$ = $\bp_{1}$) is similar across 5 different choices of $T$ of increasing magnitude. The run times to obtain the posterior samples is similar across all values of $T$. 
\begin{figure}[H]
	\centering \includegraphics[width=\linewidth]{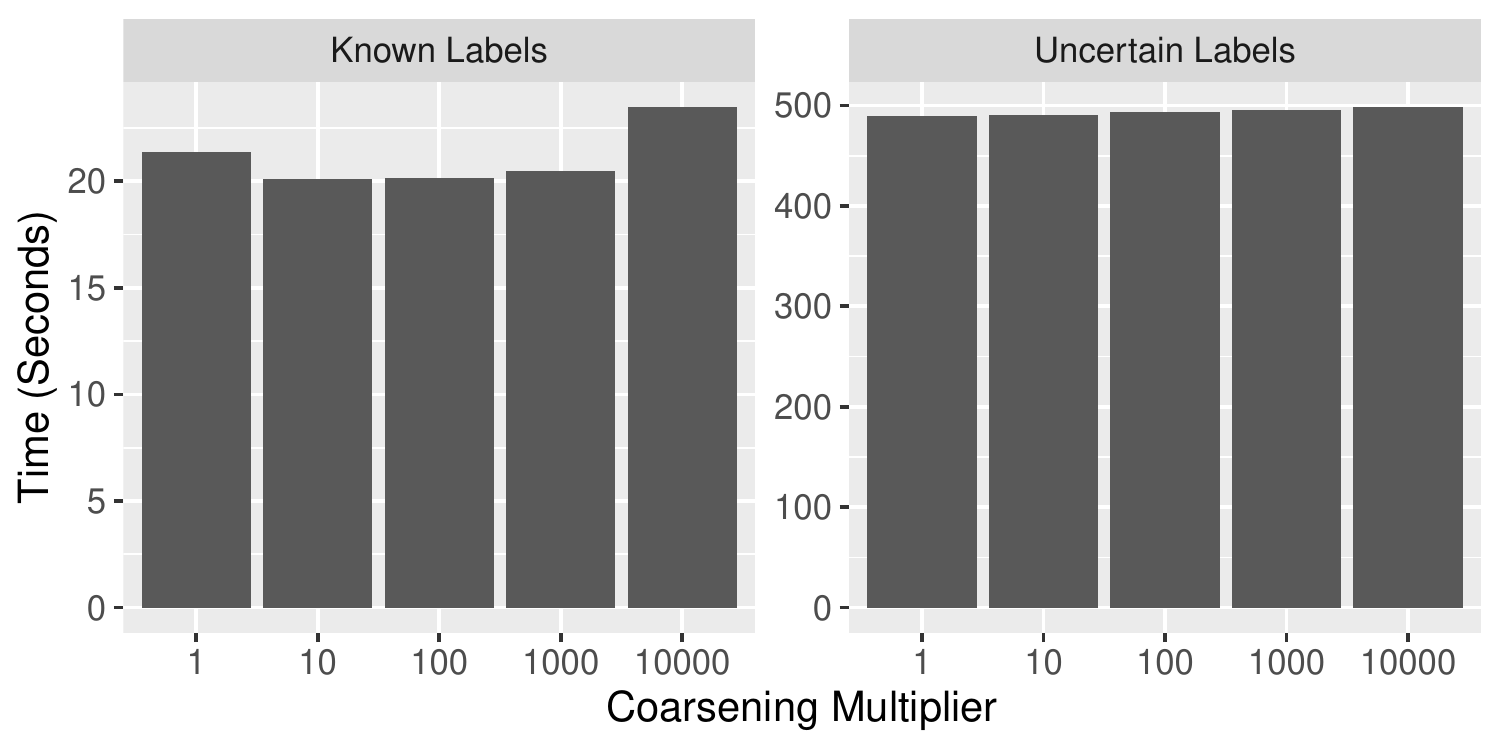}
	\caption{The average time to sample from the posterior, for different values of the coarsening multiplier $T$, when $\bp$ = $\bp_{1}$ (similar results are seen for the other three values for $\bp$).}\label{fig:coarsening_timing}
\end{figure}

\begin{figure}[]
	\centering \includegraphics[width=\linewidth]{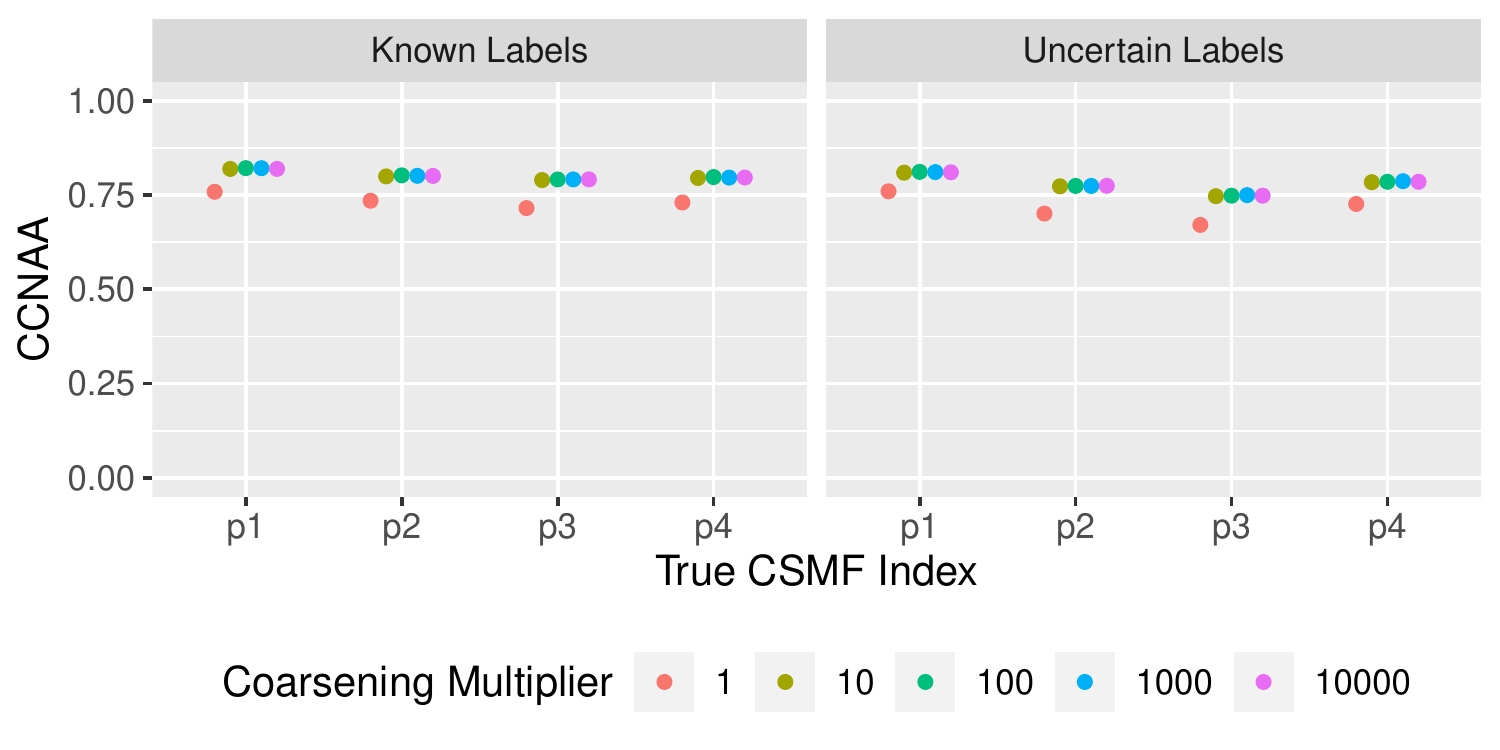}
	\caption{The CCNAA across different values of the coarsening multiplier $T$.}\label{fig:coarsening_ccnaa}
\end{figure}

Figure \ref{fig:coarsening_ccnaa} looks at the accuracy in estimating $\bp$ for the 5 choices of $T$. We see that setting $T = 1$ produces the least accurate estimate of $\bp$, while the results are similar for all the other $4$ values of $T$.

\subsection{Priors for Sparse Misclassification rates}\label{sec:simpara}
We compare the GBQL model with uninformative Dirichlet priors and no enforced sparsity, versus a sparse model discussed in Section \ref{sec:para} where the zero-values of $\bM$ used in Section~\ref{simulations} are correctly set equal to 0. To undersstand the impact on the size of the labeled data $\calL$ on estimation of $\bM$ for the two methods, we use  two choices of $n=|\calL|$. Figure \ref{fig:sparse_ccnaa} shows that when $|\calL|$ is $50$, i.e., only around $10$ cases per cause, the sparse model outperforms the full model, leading to a higher CCNAA. When $|\calL|$ grows to $300$, i.e., around $60$ cases per cause, the performance of the full model is almost indistinguishable, and the methods have a similar CCNAA across all 4 choices of  $\bp$. 

\begin{figure}[H]
	\centering \includegraphics[width=\linewidth]{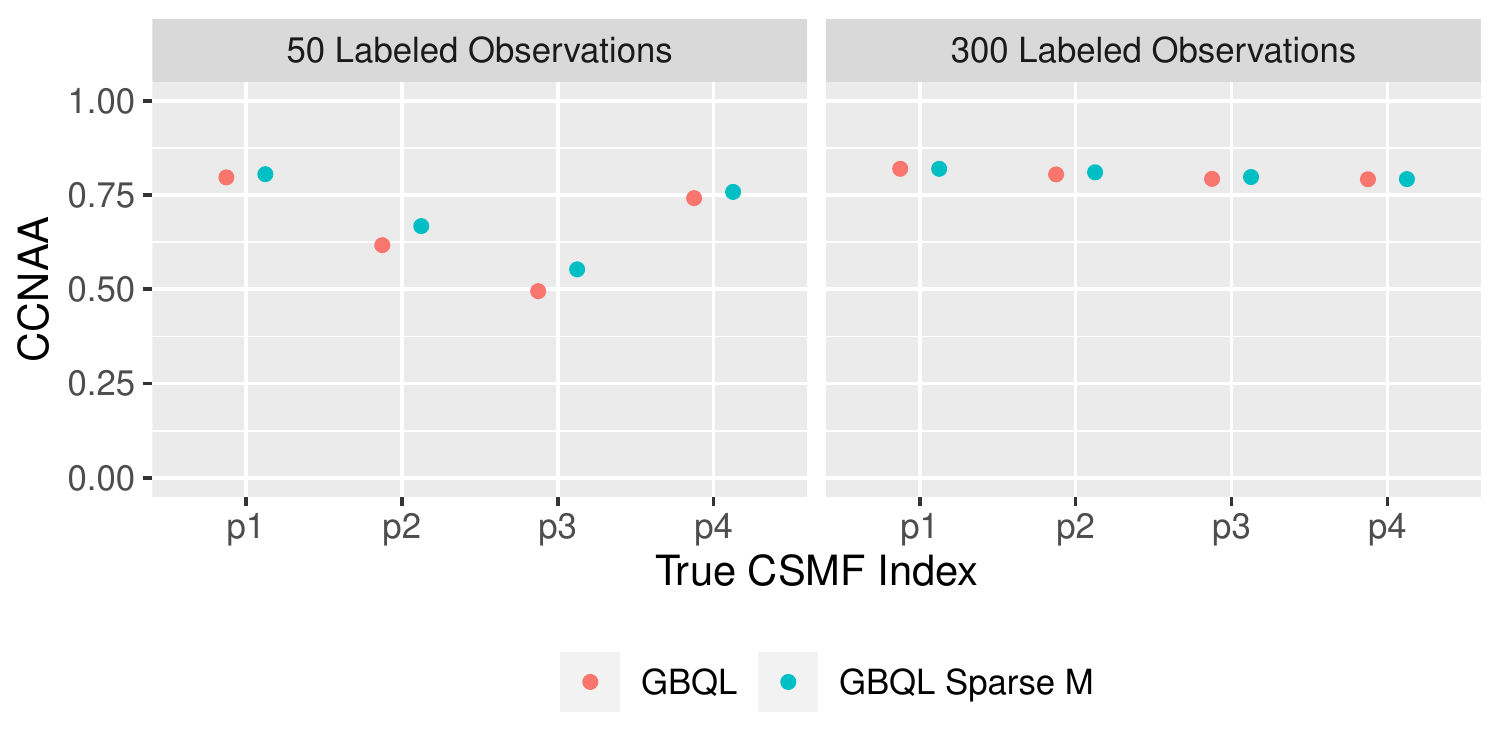}
	\caption{CCNAA of the GBQL model with uninformative priors, versus the GBQL model with sparsity enforced through setting entries of $\bM$ to zero. The two columns show the results for $|\calL|$ = 50 and 300, from left to right, respectively}\label{fig:sparse_ccnaa}
\end{figure}
\end{document}